\documentclass[a4paper, 12pt]{article}
\usepackage{jheppub}

\usepackage{bm}        % for math
\usepackage{amssymb,amsthm}   % for math
\usepackage{amsmath} 
\usepackage{mathrsfs}
\usepackage[british]{babel}

\usepackage{empheq}

\newcommand{\be}{\begin{equation}}
\newcommand{\ee}{\end{equation}}
\newcommand{\bea}{\begin{eqnarray}}
\newcommand{\eea}{\end{eqnarray}}

\newcommand{\pd}{\partial}
\newcommand{\ft}[2]{{\textstyle\frac{#1}{#2}}}

 \def\cB{{\cal B}}  
    \def\cH{{\cal H}}
 \def\cI{{\cal I}} \def\cJ{{\cal J}}  \def\cL{{\cal L}}
 \def\cM{{\cal M}}   
  \def\cR{{\cal R}}

 \def\cS{{\cal S}} \def\cT{{\cal T}}

\newcommand{\bn}{\ensuremath{\boldsymbol n}}

\def \uhat {\hat{u}}
\def \Vhat {\hat{V}}
\renewcommand\({\left(}
\renewcommand\){\right)}

\renewcommand{\Re}{\operatorname{Re}}
\renewcommand{\Im}{\operatorname{Im}}

\newcommand{\dd}{{\rm d}}

\newcommand{\rmi}{\mathrm{i}}
\newcommand{\rme}{\mathrm{e}}
\newcommand{\eqscr}{\;\hat{=}\;}

\newtheorem{theorem}{Theorem}[section]
\newtheorem{proposition}{Proposition}[section]
\newtheorem{lemma}{Lemma}[section]
\theoremstyle{definition}

\theoremstyle{remark}

\numberwithin{equation}{section}
\theoremstyle{plain}

\def\finn{\hfill \rule{2.5mm}{2.5mm}} 

\newcommand*\widefboxb[1]{\fbox{\hspace{.5em}#1\hspace{.5em}}}

\title{Supertranslations: redundancies of horizon data, and global symmetries at null infinity.}

\author[1]{K. Sousa,}
\author[1]{G. Mil\'ans del Bosch}
\author[2,3]{and B. Reina}
\affiliation[1]{Instituto de Fisica Teorica UAM-CSIC
Universidad Autonoma de Madrid, Cantoblanco, 28049 Madrid, Spain}
\affiliation[2]{School of Mathematical Sciences, Dublin City University, Glasnevin, Dublin 9, Ireland}
\affiliation[3]{Departamento de F'sica Te—rica e Historia de la Ciencia, University of the Basque Country UPV/EHU, Apartado 644, 48080 Bilbao, Spain}
\emailAdd{kepa.sousa@csic.es}
\emailAdd{guillermo.milans@csic.es}
\emailAdd{borja.reina@dcu.ie}

\abstract{
We characterise the geometrical nature of smooth supertranslations defined on a generic non-expanding horizon (NEH) embedded in vacuum.  To this end we consider the constraints imposed by the vacuum Einstein's equations on the NEH structure, and discuss the transformation properties of their solutions under supertranslations.
We present a freely specifiable data set  which is both necessary and sufficient to reconstruct  the full horizon geometry, and is  composed of objects   
 which are invariant under supertranslations.   We conclude that smooth supertranslations do not transform the geometry of the NEH,  and that they should be regarded as pure gauge. Our results apply both to stationary, and  non-stationary states of a NEH, the later ones being able to describe radiative processes taking place on the horizon.  
  As a consistency check we repeat the analysis for BMS supertranslations defined on null infinity, $\cI$.
Using the same framework as for  the NEH we recover the well known result that BMS supertranslations act non-trivially on  the free data on $\cI$. The full analysis is done in exact, non-linear, general relativity.
}  

\date{\today}

\begin{document}
\maketitle

\section{Introduction}
%%%%%%%%%%%%%%%%%%%%%%%%%%%%%%%%%%%%%%%%%%%%%%%%%%%%%%%%
%%%%%%%%%%%%%%%%%%%%%%%%%%%%%%%%%%%%%%%%%%%%%%%%%%%%%%%%

The subject of asymptotic symmetries in gravitational theories has been an active field of research over the last years.  
One of the main motivations to study    the asymptotic structure of the spacetime  boundary is the 
characterisation of candidate theories of quantum gravity. A prominent example is the result in   \cite{Brown:1986nw,Henneaux:1985tv},  that any consistent theory of quantum gravity on a spacetime which is asymptotically $\mathrm{AdS}_3$  should be a conformal field theory. This result has led to major developments in the understanding of the microscopic origin of black hole entropy, as in the case of the BTZ black hole 
\cite{Banados:1992wn,Banados:1998wy,Strominger:1997eq}, and for extremal Kerr black holes in four dimensions \cite{Guica:2008mu}.  The success of this approach has inspired many  attempts to extend these results to the case of astrophysical --non-extremal-- black holes in asymptotically flat spacetimes  (see  \cite{Barnich:2010eb,Barnich:2011ct,Barnich:2011mi,Barnich:2011ty,Koga:2001vq,Chung:2010ge,Majhi:2012tf,Donnay:2015abr,Donnay:2016ejv} and references therein).

The symmetry group  of four dimensional asymptotically flat spacetimes at null infinity $\mathcal{I}$ is the so called Bondi-Metzner-Sachs (BMS) group \cite{Bondi:1962px,Sachs:1962wk,Sachs:1962zza}. This group consists of the semidirect product of the Lorentz group times an infinite dimensional abelian normal subgroup which generalises translations, the so called \emph{supertranslations}.  Supertranslations act on future (past) null infinity by shifting the advanced (retarded) time independently for each point of the sphere at infinity.  These diffeomorphisms are particularly interesting because they act non-trivially on the geometric data at null infinity, which encodes the gravitational degrees of freedom of gravitational radiation. More specifically, the radiative vacua of asymptotically flat spacetimes is infinitely degenerate,  and supertranslations act transitively on it, i.e. all  radiative vacua are connected to each other by supertranslations. The implications of this symmetry group on the gravitational S-matrix 
were first studied in the framework of \emph{asymptotic quantization} \cite{Ashtekar:1981sf,Ashtekar:1981bq,Ashtekar:1987tt} (see also \cite{PROP:PROP19780260902,Frolov1979}), and the relation between supertranslations and Weinberg's \emph{soft graviton theorem}  \cite{weinbergSFT} has been explored in \cite{Strominger:2013jfa,He:2014laa,Pasterski:2015tva,Hawking:2015qqa,Hawking:2016sgy}.

Recently it has been argued that in  black hole spacetimes, if the  event horizon is regarded as an inner  boundary,  it is appropriate to enhance the asymptotic symmetry group (ASG) with those diffeomorphisms which leave invariant the near horizon geometry\footnote{A more complete list of related works can be found in \cite{Hawking:2016sgy}.} \cite{Koga:2001vq,Chung:2010ge,Majhi:2012tf,Hawking:2015qqa,Donnay:2015abr,Donnay:2016ejv,Hawking:2016sgy,Averin:2016hhm,Averin:2016ybl,Hawking:2016msc}. For  stationary black holes the corresponding set of diffeomorphisms has been shown to be a reminiscence of the BMS group on $\mathcal{I}$. As in the case of null infinity, the ASG on the horizon  is enhanced with respect to the isometry group of the background with the addition of supertranslations, which in this case shift the advance (retarded) time of the future (past) horizon.  Following  the analogy with the BMS group on null infinity, it has been conjectured that the horizon supertranslations may act non-trivially on the black hole geometry, transforming the black hole to a physically inequivalent one. If this was the case, the  ``supertranslation hair'' could provide some  insight on the microscopic degrees of freedom associated to the black hole entropy.   
 Although these ideas are certainly appealing, the physical nature of the asymptotic symmetry group defined on non-extremal horizons  is still unclear, and the proposal remains controversial  \cite{Mirbabayi:2016axw,Bousso:2017dny,Bousso:2017rsx}. 

On the one hand there are still some discrepancies on the structure  of the asymptotic symmetry group found in different analyses \cite{Koga:2001vq,Chung:2010ge,Majhi:2012tf,Donnay:2015abr,Eling:2016xlx,Setare:2016msj,Donnay:2016ejv}.  These differences could be attributed to alternative choices of boundary conditions for  the metric tensor near the horizon, but the geometric interpretation of the discrepancies is not well understood.   The main difficulty to compare different analyses is that they are based on a coordinate dependent approach. In practice, the boundary conditions  are defined as restrictions on a explicit coordinate expression of the metric tensor in the neighbourhood of the horizon. 
This complicates the comparison between the various works as, in general, a given set of boundary conditions does not retain the same form in different coordinate systems.

On the other hand it remains an open question to determine the effect of the horizon  ASG on the geometry, and in particular  whether if these diffeomorphisms act non-trivially on the physical state of the black hole. This problem was recently addressed in \cite{Hawking:2016sgy}, where the authors performed  a hamiltonian analysis to characterise  the phase space of a  Schwarzschild spacetime.  According to this study  the phase space of the Schwarzschild spacetime is infinite dimensional, and supertranslations  act non-trivially on it. This  conclusion contrasts with the classical result that there is only a three-parameter family of stationary black hole solutions in Einstein-Maxwell theory. 
Actually, the  hamiltonian analysis of  spacetimes containing isolated horizons, such as the  Schwarzschild spacetime, had  been considered before in \cite{Ashtekar:1998sp,Ashtekar:1999wa,Ashtekar:2001is,Ashtekar:2001jb}. In those works  it was shown that the corresponding phase space could be   infinite dimensional in the presence of radiation, but in the stationary case it was argued that the physical state is completely determined by the standard quantities: ADM mass, angular momentum and electric charge.

In the present paper we will study   horizon supertranslations  defined on a  generic non-expanding  horizon (NEH) which is embedded in  vacuum \cite{Ashtekar:1998sp,Ashtekar:1999wa,Ashtekar:2000hw,Ashtekar:2001is,Ashtekar:2001jb}. A non-expanding horizon is a generalisation of a killing horizon which admits  gravitational radiation propagating arbitrarily close to it, and even \emph{on} the NEH itself (but not crossing it).
Our main objective is to provide a coordinate invariant definition of horizon supertranslations, and then to characterise their effect on the NEH geometry.   For this purpose we have taken as a guide the geometric method used by R. Geroch  \cite{Geroch1977} and A.  Ashtekar \cite{Ashtekar:1981hw,Ashtekar:1981bq,Ashtekar:1981sf,Ashtekar:1987tt,Ashtekar:2014zsa} to study the structure and dynamics of null infinity. Following these works,  we describe the  horizon in terms of an  abstract  three-dimensional manifold separated from the spacetime, and which  is   diffeomorphically identified with the horizon. In this framework, the information about  the intrinsic and extrinsic geometry of the  horizon is encoded in tensor fields living on the abstract manifold. The advantage of this method is that the geometric data of the horizon is isolated from the rest of the spacetime, and  moreover, all the gauge redundancies are well characterised. 

 To clarify the geometrical nature of supertranslations we have studied the set  spacetime diffeomorphisms which preserve the horizon as a set of points, and leave invariant the metric tensor on it.  Note that these diffeomorphisms, which we call for short \emph{hypersurface symmetries}, are defined in a coordinate invariant way, i.e. without involving an explicit coordinate expression for the spacetime metric tensor. After checking that horizon supertranslations belong to this class of diffeomorphisms,  we have studied the behaviour of  the complete horizon geometry (including the extrinsic geometry) under an arbitrary hypersurface symmetry.   Our analysis shows that the    effect of these diffeomorphisms on the  horizon can be identified with a gauge redundancy  of the description. In other words, hypersurface symmetries, and in particular supertranslations, leave  invariant both  \emph{the intrinsic  and  the extrinsic} geometry of the horizon up to a gauge redundancy of the description.   Note, however, that  this result is not sufficient to claim that supertranslations act trivially on the  geometry of the horizon. Indeed, supertranslations could be large gauge transformations, i.e.   \emph{global symmetries}, which      can change the dynamical state of a system \cite{Benguria:1976in,Henneaux:1985tv,Brown:1986nw}. Thus, we still need to identify the dynamical degrees of freedom of the horizon, or equivalently, a data set which is both \emph{necessary and sufficient} to reconstruct the full NEH geometry, and then we have  to determine how it transforms   under supertranslations.

In order to identify the dynamical degrees of freedom of the NEH we follow the geometric method of \cite{10230751794} and \cite{Ashtekar:2001jb} (see also \cite{Mars:2013qaa,Gourgoulhon:2005ng}), which consists on   studying the constraints imposed by the vacuum Einstein's equations on its geometric data. By solving these constraint equations it is possible extract a set of  freely specifiable quantities   which contain all the information necessary to reconstruct the NEH geometry \cite{Ashtekar:2001jb}. Thus, the resulting  \emph{free horizon data set} encodes the dynamical degrees of freedom of the horizon, but in general it also involves some gauge redundancies.  In the present work we have reconsidered the analysis in \cite{Ashtekar:2001jb}  for non-expanding  horizons,  discussing in detail the  treatment of the  gauge redundancies, and specifically of supertranslations. 

The main result of this work is the identification of a free data set which does not involve any unfixed gauge degree of freedom and, in particular, which is composed of objects which are \emph{invariant under supertranslations}. 
 An immediate consequence of our result is that  supertranslations do not affect the NEH geometry, as it can be encoded entirely  in quantities 
   which are  invariant under these diffeomorphisms.    In particular     the stationary state of the horizon    is completely determined by its intrinsic geometry  and its angular momentum aspect,  and neither of the two  transform  under horizon supertranslations.  The supertranslation invariant data set is also sufficiently general to represent non-stationary states of the NEH, and thus, it can describe radiative processes occurring at the horizon.     It is important to stress that,  to avoid excluding physically allowed configurations of the horizon, we have not eliminated the freedom to perform supertranslations using gauge fixing conditions.  Instead, guided by the treatment of null infinity  \cite{Ashtekar:1981hw}, we have dealt with the redundancy describing the NEH geometry in terms of variables which are invariant under supertranslations.
    
As a consistency  check we have repeated  the analysis of null infinity done in \cite{Ashtekar:1981hw} using the same framework as for non-expanding horizons. In particular, following \cite{Ashtekar:1981hw}, we have characterised the solutions to the constraint equations of null infinity in terms of variables invariant under \emph{BMS supertranslations}, and we have reproduced the proof of the  degeneracy of the radiative vacuum of asymptotically flat spacetimes. In other words, we find that the  solution space of the constraint equations of $\cI$  in the absence of  radiation is infinite dimensional.  Since the analysis is done in terms of variables which are free of any gauge redundancies, these degenerate vacua must be regarded as physically distinct, and yet they can be shown to be connected to each other   by supertranslations. Therefore, we recover the well known result that BMS supertranslations --contrary to the case of horizons-- act non-trivially on the free data of null infinity, i.e. they represent  a \emph{global symmetry} of the constraint equations.\\

This article is organised as follows. In section \ref{sec:Hypersurfaces}   we  review the formalism to describe the  geometry of null hypersurfaces, together with the constraint equations that restrict the  corresponding geometric data. In section \ref{sec:redundancies} we characterise in detail the gauge redundancies inherent to our description, and we discuss the effect of supertranslations on the horizon geometry.  In section \ref{sec:NEH} we analyse the constraint equations of a non-expanding horizon, and present  a free data set to describe its geometry which is composed of quantities   invariant under supertranslations. In section \ref{sec:Scri} we consider the constraint equations  for null infinity,  and we  reproduce the proof of the  degeneracy of the radiative vacuum of asymptotically flat spacetimes using our framework.  Finally in section \ref{sec:discussion}  we  discuss our results.

\section{Dynamics of null hypersurfaces}
\label{sec:Hypersurfaces}
%%%%%%%%%%%%%%%%%%%%%%%%%%%%%%%%%%%%%%%%%%%%%%%%%%%%%%%%
%%%%%%%%%%%%%%%%%%%%%%%%%%%%%%%%%%%%%%%%%%%%%%%%%%%%%%%%

In this section we will review the geometry and dynamics of null hypersurfaces, what will also serve to present the relevant formulae. A more detailed overview of this subject can be found in \cite{BlauGR,Gourgoulhon:2005ng}, while the specific framework used here is based on  \cite{Mars:2013qaa}. 

We begin setting our notation and general conventions. We will work with  $(3+1)-$dimensional spacetimes  $(\cM,g)$, described by a manifold $\cM$ equipped with a metric tensor $g$ with signature signature $(-,+,+,+)$. We will denote the spacetime coordinates by $\{x^\mu\}$, with the index running over $\mu=0,1,2,3$. 
The  Riemann tensor is defined in terms of the Ricci identity as follows
\be
\nabla_{[\mu} \nabla_{\nu]} V^\sigma \equiv  \nabla_\mu  \nabla_\nu V^\sigma - \nabla_\nu  \nabla_\mu V^\sigma = R^\sigma_{\phantom{\sigma}\rho\mu \nu} V^\rho,
\ee
where $V^\mu$ is an arbitrary  vector field\footnote{We will use the shorthands  $W_{[\mu\nu]} = W_{\mu\nu} -W_{\nu\mu}$ and $W_{(\mu\nu)}= W_{\mu\nu} +W_{\nu\mu}$ to denote the symmetrisation and anti-symmetrisation of indices.}, the Ricci tensor is given by  $R_{\mu\nu} = R^\sigma_{\phantom{\sigma}\mu \sigma\nu}$, and the scalar curvature by $R= R^\mu_\mu$. The Riemann curvature can be split in its trace part, characterised by the Schouten tensor $S_{\mu\nu} \equiv R_{\mu \nu} - \ft16 Rg_{\mu\nu}$, and its  traceless part, encoded in the  Weyl  tensor $C_{\mu\nu\rho\sigma}$   
\be
R_{\sigma\rho\mu \nu } = C_{\sigma\rho\mu\nu} + \ft12(g_{\sigma[\mu} S_{\nu] \rho} - g_{\rho[\mu} S_{\nu]\sigma}).
\ee
We will use  geometrized units $c=G=1$  so that the Einstein's equations read
\be
G_{\mu\nu} \equiv R_{\mu\nu} - \ft12 R g_{\mu\nu} = 8 \pi T_{\mu\nu}.
\label{eq:Einstein}
\ee
In regions where the spacetime geometry is consistent with the vacuum Einstein's equations, $R_{\mu\nu}=0$,  the Schouten tensor must be zero and thus the Riemann curvature is completely determined by the Weyl tensor   $R_{\mu\nu\rho\sigma} = C_{\mu\nu\rho\sigma}$.

\subsection{Geometric data of null hypersurfaces}
\label{sec:nullGeometry}
%%%%%%%%%%%%%%%%%%%%%%%%%%%%%%%%%%%%%%%%

In this section we will review the geometry and dynamics of a null hypersurface $\cH$.  Bearing in mind the case of  black hole horizons and null infinity, we will assume the hypersurface to have the topology $\cH \cong\mathbb{R}\times \mathbb{S}^2$.
We will describe the hypersurface   as the embedding of an abstract three dimensional manifold $\Sigma\cong\mathbb{R}\times \mathbb{S}^2$ on the spacetime via the   diffeomorphism $\Phi: \Sigma \to \cM$,   so that $\Phi(\Sigma) = \cH$  (see \cite{Mars:2013qaa}). The manifold $\Sigma$ acts as a diffeomorphic copy of $\cH$ detached from the spacetime, and it is introduced for convenience in order to isolate the dynamical degrees of freedom (i.e. the free geometric data) of the hypersurface.  

To characterise the intrinsic and extrinsic geometry of the hypersurface it is convenient to introduce a basis of the spacetime tangent space adapted to $\cH$.   For this purpose  let us first   define a coordinate system for the abstract manifold   $\{\xi^a\}$, with the index running over $a=1,2,3$.  The corresponding coordinate  basis of  the tangent space $T_p\Sigma$ is then  given by $\hat \cB \equiv \{\hat e_a =\pd_{\xi^a}\}$, with $p\in \Sigma$. The elements of the basis $\hat \cB$ are identified with a set of linearly independent spacetime vectors  tangent to the hypersurface $e_a\equiv \dd\Phi(\hat e_a)$, via the pushforward map $\dd \Phi$   associated to $\Phi$. Then, we can form a  basis $\cB=\{e_a,\ell\}$ of the spacetime tangent space completing the  set of vectors $\{e_a\}$ with any vector $\ell$ transverse to $\cH$, the so called \emph{rigging}.

The spacetime metric over the hypersurface can be characterised in terms of  
 a  set of tensor fields over $\Sigma$ which, by definition, have the following components on the basis $\hat \cB$  
 \be
\gamma_{ab}\equiv g(e_a,e_b)|_{\Phi(p)}, \qquad  \ell_a \equiv g(\ell,e_a)|_{\Phi(p)}, \qquad \ell^{(2)} \equiv g(\ell,\ell)|_{\Phi(p)}.
\label{eq:metricData0}
\ee
These fields  encode the scalar products of the elements in the  spacetime basis $\cB=\{e_a,\ell\}$, and in particular $\gamma_{ab}$ represents the induced metric  on $\cH$. This set of fields are known as the \emph{hypersurface metric data}.

In order to reduce the  large degree of gauge freedom in this description, namely  the choice of coordinates on $\Sigma$ and the specification of the rigging vector $\ell$, it is useful to introduce some simplifying conventions.  The normal one-form $\bn$ and the normal vector $n$  to the hypersurface are determined by the conditions  $\bn(e_a)=0$ and $n= g^{-1}(\bn,\cdot)$ respectively, and thus they are defined up to $\bn \to \lambda \bn$ and $n \to \lambda n$, where $\lambda$ is a scalar field on $\cH$.
Since the normal vector to a null hypersurface is null, i.e. $g(n,n) = \bn(n)=0$,  $n$ is also tangent to $\cH$,  and we will choose it to be future directed.  Therefore we can partially fix the coordinate system on the abstract manifold  $\Sigma$  defining $\xi^1$ so that $e_1=n$, and parametrising the $\mathbb{S}^2$ component of the hypersurface with the coordinates $\xi^M$, where $M=\{2,3\}$. Moreover we will require the rigging vector $\ell$ to be null $g(\ell,\ell)=0$, and we will fix its normalisation and  direction with respect to the vectors $\{e_a\}$ imposing the conditions $\bn(\ell)=1$ and $g(\ell,e_M)=0$ everywhere on $\cH$, what can be expressed equivalently as $g(\ell,e_a)=\delta_a^1$. 
With these choices the explicit coordinate expressions for the hypersurface metric data in the basis $\hat \cB=\{\hat e_1,\hat e_M\}$ read
\be
\gamma_{ab} = \begin{pmatrix}
0& 0 \\
0 & q_{MN}
\end{pmatrix}, \qquad \ell_a  = (1,0,0), \qquad \ell^{(2)} =0. 
\label{eq:metricData}
\ee
Here $q_{MN} \equiv g(e_M,e_N)|_{\Phi(p)}$ represents the induced metric on the spatial sections $\cS_{\xi^1} \equiv \Sigma|_{\xi^1}$ of the horizon, which are defined by the level sets of the null parameter $\xi^1$.   In the following we will denote the elements of the basis of the spacetime tangent space 
by  $\cB=\{n, \ell, e_M\}$. For later convenience, we also write here the following  identity satisfied by the elements of $\cB$
\be
g^{\mu \nu} = \ell^{(\mu} n^{\nu)} + q^{MN} e_N^\mu  e_M^\nu.
\label{eq:inverseG}
\ee
The Levi-Civita connection at points of the hypersurface can be characterised specifying its action on the elements of   
 the basis $\cB$ 
\begin{flalign}
\nabla_{n} n = \kappa n ,  \qquad  &\nabla_{M} n = \Omega_M n + \Theta_M^{\phantom{N}N} e_N,\nonumber\\
\nabla_{n} e_M =  \Omega_M  n + \Theta_M^{\phantom{N}N} e_N, \qquad &\nabla_{M} e_N =- \Theta_{MN} \ell - \Xi_{MN} n +\overline  \Gamma_{MN}^L e_L, \nonumber \\
\nabla_{n} \ell =-\kappa \ell  -\Omega^M e_M,  \qquad &\nabla_{M} \ell = -\Omega_M \ell  +\Xi_{M}^{\phantom{M}N} e_N,
\label{eq:connectionCoeff}
\end{flalign}
where the indices $M,N$ are raised and lowered with $q_{MN}$ and its inverse $q^{MN}$, $\nabla_n \equiv n^\mu\nabla_\mu$ and $\nabla_M\equiv e^\mu_M \nabla_\mu$. This is the most general form of the connection coefficients consistent with our conventions \eqref{eq:metricData}, as it can be easily  derived in the framework of \cite{Mars:2013qaa}. In particular it can be seen that the integral curves of the normal vector $n$ are  null geodesics parallel to the hypersurface, and the inaffinity parameter $\kappa$ is referred as the  \emph{surface gravity}  in the case of horizons. The surface gravity, together with the \emph{Hajicek one-form} $\Omega_M$ and $\Xi_{MN}$ can be conveniently encoded in the tensor field  $Y_{ab}$ defined  in terms of the basis $\hat \cB$ of $T_p\Sigma$ as follows
\be
Y_{ab}\equiv\ft12 \cL_\ell g (e_a,e_b)|_{\Phi(p)}=  \begin{pmatrix}
-\kappa & -\Omega_M\\
-\Omega_N & \Xi_{MN}
\end{pmatrix}.
\label{eq:tensorY}
\ee
The set of coefficients $\Xi_{MN}= \ft12e^\mu_M e^\nu_N \nabla_{(\mu} \ell_{\nu)}$ characterises the components of the Levi-Civita connection associated to directions  which are all transverse  to the normal vector $n$, and 
thus we will refer to it as the \emph{transverse connection}. For later convenience we will also introduce the rotation one-form $\omega_a$ which is defined by
\be
\omega_a \equiv - Y_{ab}\, \hat n^b = (\kappa,\Omega_M),
\label{eq:defRot}
\ee
where $\hat n$ is the vector on $T_p \Sigma$ which is identified with the null normal via the embedding $\dd \Phi(\hat n) \equiv n$.  
The remaining connection coefficients, $\Theta_{MN}$ and 
$\overline{\Gamma}_{BA}^C$, are fully determined by the intrinsic geometry via the equations   
\be
\ft{1}{2}\pd_n q _{MN}= \Theta_{MN}, \qquad \overline{\Gamma}_{MN}^L = \ft12 q^{LP} (\pd_M q_{NP} + \pd_N q_{MP} - \pd_P q_{MN}),
\ee
where  $\pd_n q \equiv \pd_{\xi^1} q$. Thus,  $\overline{\Gamma}_{BA}^C$ represents the Levi-Civita connection compatible with $q_{MN}$, and the quantity $\Theta_{MN}$ is   known as the \emph{second fundamental form}.\\  

Summarising, the intrinsic and extrinsic  geometry of a null hypersurface can be fully  encoded in the following set of fields defined on $\Sigma$
\be
\boxed{
\text{Hypersurface data:} \qquad  \mathscr{D}\equiv (q_{MN}, \quad \kappa, \quad \Omega_M, \quad \Xi_{MN}).}
\label{eq:dataDef}
\ee
Our choice of coordinates for $\Sigma\cong\mathbb{R}\times\mathbb{S}^2$, with $\xi^1$ running along the null direction $n$ of the hypersurface, and $\xi^M$ parametrising the sections with constant $\xi^1$, $\cS_{\xi^1}\cong\mathbb{S}^2$, allows to picture the data $\kappa$, $\Omega_M$ and $ \Xi_{MN}$ as tensor fields living on a manifold with the topology of a sphere and a Riemannian metric $q_{MN}$. In this picture, the dependence of these fields on the null coordinate $\xi^1$ is interpreted as a ``temporal'' evolution \cite{damour,Price:1986yy,Thorne:1986iy}. As we shall see in the next section,  the evolution of  these fields  along the null direction is not completely free, as it is restricted by the geometry of the ambient space.  Moreover, as we shall see in sections \ref{sec:redundancies} and \ref{sec:NEH}, this description still involves some residual gauge redundancies, which lead to further constraints on \eqref{eq:dataDef} after gauge fixing.\\

In the following we will often identify the abstract manifold $\Sigma$ with  the hypersurface $\cH$, and we will leave implicit the pull-back $\Phi^*$ operation in the formulae  to simplify the notation.

\subsection{Constraint equations of null hypersurfaces}
%%%%%%%%%%%%%%%%%%%%%%%%%%%%%%%%%%%%%%%%

The  spacetime connection on the hypersurface, given by \eqref{eq:connectionCoeff},  must be consistent with the geometry of the ambient space 
where $\cH$ is embedded. This requirement leads to the \emph{hypersurface constraint  equations} which relate the connection coefficients in \eqref{eq:connectionCoeff} with certain projections of the Ricci  tensor $R_{\mu\nu}$ at points of the hypersurface $\cH$. When expressed in terms of the hypersurface data $\mathscr{D}$ \eqref{eq:dataDef} these mathematical identities take the form of a set of equations of motion, which we will now review. A formal analysis of these equations can be found in \cite{Mars:1993mj,Mars:2013qaa}, and their application to null hypersurfaces is reviewed in detail in \cite{Gourgoulhon:2005ng}. Since our  conventions do not match the ones in  these references, for completeness we have included a derivation of these formulae in appendix \ref{app:constraints}. The relevant equations are:
\begin{itemize}
\item \emph{Raychaudhuri equation}:
\be
\hspace{-1cm}\boxed{
\pd_n \theta - \theta \kappa +  \Theta_{MN} \Theta^{MN} =J_{nn}.
\label{eq:Raychad}
}
\ee
\item \emph{Damour-Navier-Stokes equations}:
\be
\centering
\hspace{-1cm}\boxed{
\pd_n \Omega_M  - \pd_M \kappa + \theta \Omega_M  + D_N \Theta^N_{M} - D_M \theta =-J_{nM}. 
\label{eq:NS}}
\ee
\item Equation for the \emph{transverse connection}:
\begin{empheq}[box=\widefboxb]{align}
\begin{split}
\pd_n \Xi_{MN} =& - \ft12 D_{(M} \Omega_{N)} -\Omega_M \Omega_N - (\kappa + \ft12 \theta) \Xi_{MN}  + \Xi_{P(M} \Theta^P_{N)}\\
& -\ft12 \Theta_{MN} \theta^{\ell} 
+ \ft14 \cR q_{MN} +\ft12 J_{MN}.
\label{eq:Xi}
\end{split}
\end{empheq}
\end{itemize}
Here  $\theta\equiv \Theta_M^M$ is the \emph{expansion} of the null hypersurface,  and $\theta^{\ell} \equiv \Xi^M_M$.  The symbols $D_M$ and $\cR$ denote the Levi-Civita connection of $q_{MN}$ and the associated Ricci scalar respectively.  The equations also involve  the tensor $J_{ab}$ defined on $\Sigma$ in terms of its components in the basis $\hat \cB=\{\hat e_a\}$  
 \be
J_{nn} \equiv - \mathbf{R}(n, n)|_{\Phi(p)}, \quad J_{nM} \equiv - \mathbf{R}(n,e_M)|_{\Phi(p)}, \quad  
J_{MN} \equiv -\mathbf{R}(e_M, e_N)|_{\Phi(p)},
\label{eq:sourceTerms}
 \ee
where  $\mathbf{R}(e_a,e_b) = R_{\mu\nu} e^\mu_a e^\nu_b$ represent projections of the spacetime  Ricci tensor.   
The constraint equations (\ref{eq:Raychad}-\ref{eq:Xi}) simplify considerably in the particular case of \emph{non-expanding horizons} which are embedded in vacuum.  If  the  Ricci tensor $R_{\mu\nu}$ is consistent with the Einstein's field equations 
\eqref{eq:Einstein}, the quantities \eqref{eq:sourceTerms} can be associated to projections of the energy momentum tensor on the basis $\cB$. Therefore they are all vanishing $J_{nn}=J_{nM} = J_{MN}=0$ in vacuum $T_{\mu\nu}=0$.  By definition, a non-expanding horizon is a null hypersurface which has a vanishing expansion $\theta$ \cite{Ashtekar:1998sp,Ashtekar:2001jb} (see also \cite{Gourgoulhon:2005ng}), and then, 
due to the vacuum 
Raychaudhuri equation  \eqref{eq:Raychad}, we must have  \begin{empheq}[box=\widefboxb]{align}
\text{Non-expanding horizon:} \quad  \theta=0 \quad \Longrightarrow\quad  \ft12 \pd_n q_{MN} =  \Theta_{MN}=0.
\label{eq:nonExp}
\end{empheq}
As a consequence,  the spatial  metric $q_{MN}$ induced on the  sections of a non-expanding horizon $\cS_{\xi^1}$ is independent on the null coordinate $\xi^1$. Moreover, for a NEH the equations \eqref{eq:NS} and \eqref{eq:Xi} reduce to 
 \bea
\pd_n \Omega_M   &=& \pd_M \kappa,  \label{eq:DNS2}\\
\pd_n \Xi_{MN} &=& - \ft12 D_{(M} \Omega_{N)}  - \kappa \Xi_{MN} -\Omega_M \Omega_N + \ft14 q_{MN} \cR,
\label{eq:transverseCon2}
\eea
where we have already imposed the vacuum Einstein's equations.

As we shall review in section  \ref{sec:Scri}, null infinity $\cI$ can be described as an non-expanding hypersurface with $\Theta_{MN}=0$  using  Penrose's conformal framework, and its structure is also constrained by (\ref{eq:Raychad}-\ref{eq:Xi}), which are mathematical identities  satisfied by any null hypersurface.   However, non-expanding horizons and null infinity have very different dynamical behaviour, and in particular  the constraints (\ref{eq:DNS2}) and (\ref{eq:transverseCon2}) are not valid for $\cI$. One of the reasons is  that the geometric data of $\cI$ is only defined up to conformal transformations, what requires introducing appropriate 
 equivalence classes of data sets  \cite{Ashtekar:1981hw}.  The other important difference with NEHs is that the  Ricci tensor defined on the conformal completion of spacetime does not satisfy the ordinary Einstein's equations,  and thus a specific treatment is required for $\cI$. 

\subsection{Newman-Penrose null tetrad and Weyl scalars}
\label{sec:WeylScalars}
%%%%%%%%%%%%%%%%%%%%%%%%%%%%%%%%%%%%%%%%

The set of constraint equations (\ref{eq:Raychad}-\ref{eq:Xi}), ensures the consistency of the connection coefficients in \eqref{eq:connectionCoeff} with the trace part of the ambient-space Riemann tensor, i.e. the Ricci tensor $R_{\mu\nu}$. Therefore, it is possible to obtain further constraints  requiring that the extrinsic geometry of $\cH$ to be  compatible with the traceless part of the curvature, that is, with  the  Weyl tensor $C_{\mu\nu\rho\sigma}$.  The Weyl tensor has 10 independent components which can be collected in the form of 5  independent complex scalars $\Psi_n$, with $n=0,\ldots, 4$, the so called  \emph{Weyl scalars}. 

In order to define the Weyl scalars, first we have to introduce a  \emph{Newman-Penrose null tetrad} (see e.g. \cite{Newman:1981fn,Chandrasekhar:1985kt}), what can be done in our framework as follows. 
At any given point $\xi^M_0$ of the spatial sections $\cS_{\xi^1}$ it is possible to find a set of coordinates  $\xi^M$ such that the spatial metric $q_{MN}$ has the simple form $q_{MN}(\xi_0)= \delta_{MN}$. Note that for NEH horizons this choice is independent of $\xi^1$ as $\pd_n q_{MN}=0$.  In this way we ensure that the two basis vectors $e_M|_{\xi_0^M}$ are orthogonal to each other and have unit norm. Then, we can construct the Newman-Penrose null tetrad $\cB_{NP}=\{n,\ell, m, \overline m\}$ at $\{\xi^1,\xi^M_0\}$  comprised of the normal vector  $n$, the rigging $\ell$, and the two  complex null vectors 
\be
m|_{\xi_0^M} \equiv \frac{1}{\sqrt{2}} (e_2 + \rmi e_3)|_{\xi_0^M}, \qquad \overline m|_{\xi_0^M} \equiv \frac{1}{\sqrt{2}} (e_2- \rmi e_3)|_{\xi_0^M}.
\label{eq:mbarm}
\ee 
By defining the $\cB_{NP}$ in this way we avoid introducing the additional gauge freedom which is always associated to the choice of null tetrad. 
 The set of vectors  $\cB_{NP}$  also forms a basis of the spacetime tangent space, and it is composed of null vectors only. 
 Actually, at $\xi_0^M$ the scalar products of its elements read
\bea
g(n,n) &=& 0,\quad g(n,\ell) = 1, \quad g(n,m) = 0, \quad g(n,\overline m) = 0,\nonumber \\
g(\ell,\ell) &=& 0,\quad g(\ell,m) = 0, \quad g(\ell,\overline m) = 0,\nonumber \\
g(m,m) &=& 0, \quad g(m,\overline m) = 1,\nonumber \\
g(\overline m,\overline m) &=& 0.
\label{eq:nullTetradProds}
\eea
The Weyl scalars are defined in terms of the Newman-Penrose tetrad by\footnote{We use the  definitions in \cite{Gourgoulhon:2005ng} up to a difference in the conventions:  in \cite{Gourgoulhon:2005ng} the second element of the tetrad $\cB_{NP}$ is future directed, but in our conventions the rigging $\ell$ is past directed.}  
\bea
\Psi_0 = C_{\sigma \rho \mu \nu} n^\sigma m^\rho n^\mu m^\nu,Ê&\qquad& \Psi_1 = C_{\sigma \rho \mu \nu} n^\sigma m^\rho \ell^\mu n^\nu,Ê\nonumber \\
\Psi_2 = C_{\sigma \rho \mu \nu} n^\sigma m^\rho  \ell^\mu \overline m^\nu,Ê&\qquad& \Psi_3 = C_{\sigma \rho \mu \nu} n^\sigma \ell^\rho \overline m^\mu \ell^\nu,\nonumber \\
\Psi_4 = C_{\sigma \rho \mu \nu} \overline m^\sigma \ell^\rho \overline m^\mu \ell^\nu.  \ \ &&
\label{eq:defWeyl}
\eea
The computation of the Weyl scalars is useful to determine the Petrov type of the gravitational field (see \cite{Chandrasekhar:1985kt}), and to characterise its different contributions  \cite{Szekeres:1965ux}. In particular, $\Psi_0$ and $\Psi_4$ encode   transverse wave components travelling along  the directions\footnote{The vector $-\ell$ is future directed since, by convention, $g(n, \ell)= 1$. } $-\ell$ and $n$ respectively.  The scalars  $\Psi_1$ and $\Psi_3$  represent  longitudinal wave components propagating respectively parallel to   $-\ell$ and $n$, and  $\Psi_2$ can be associated with a Coulomb contribution of the gravitational field  \cite{Szekeres:1965ux}. 

The  general form of the Weyl scalars in terms of the hypersurface data \eqref{eq:dataDef} can be found appendix \ref{app:weyl}.  We will present the relevant formulae when discussing case of non-expanding horizons and null infinity.

\section{Gauge redundancies and horizon supertranslations}
\label{sec:redundancies}
%%%%%%%%%%%%%%%%%%%%%%%%%%%%%%%%%%%%%%%%%%%%%%%%%%%%%%%%
%%%%%%%%%%%%%%%%%%%%%%%%%%%%%%%%%%%%%%%%%%%%%%%%%%%%%%%%

In the present section we will discuss in detail  the gauge redundancies in  our description in the case of generic non-expanding  null hypersurfaces with $\Theta_{MN}=0$, which are of interest both  for the study of black hole horizons and null infinity.  
Part of this gauge freedom was already used in the last section to set the hypersurface metric  data in the  the form \eqref{eq:metricData}.
Then, in section \ref{sec:residualGaugeH} we will begin our analysis identifying  the residual gauge redundancies which are left after imposing the conventions \eqref{eq:metricData}.  Since the constraints (\ref{eq:Raychad}-\ref{eq:Xi}) are direct consequence of the Ricci identity and \eqref{eq:metricData}, these residual gauge transformations also leave invariant the form of the constraint equations.

In section \ref{sec:supertranslationsGauge}, we will turn our attention to horizon supertranslations.  We will define them  as spacetime diffeomorphisms preserving the metric tensor of a generic NEH.  We will characterise how the geometry of the horizon changes under the action of a  supertranslation, and show that the transformation of the hypersurface data \eqref{eq:dataDef} can be identified with a   gauge redundancy of the description. In other words, we prove that horizon supertranslations  preserve both the intrinsic and extrinsic geometry of the horizon  up to a gauge transformation. As a consequence, we find that supertranslations also leave invariant the form of the constraint equations  (\ref{eq:Raychad}-\ref{eq:Xi}). 

It is important to stress that, the fact that two data sets  
can be related to each other by a gauge transformation is \emph{necessary} for them 
to describe \emph{the same NEH geometry}. However, the gauge equivalence of two data sets is \emph{not  sufficient} to prove that they correspond to the same NEH. For this reason the results derived in this section do not yet prove that two data sets related by a horizon supertranslation represent the same horizon geometry. Although this might seem unnatural at first, recall that in the case of  null infinity there are geometrically distinct data sets which can be related to each other by a gauge transformation, namely by BMS supertranslations. In this sense, BMS supertranslations should be regarded as a large gauge transformation, i.e. a global symmetry of the constraint equations, rather than pure gauge. We will discuss again this point  in section \ref{sec:Scri}.

\subsection{Gauge redundancies of hypersurface data}
\label{sec:allRedundancies}
%%%%%%%%%%%%%%%%%%%%%%%%%%%%%%%%%%%%%%%%

As we described in section \ref{sec:nullGeometry}, the intrinsic and extrinsic geometry of a null hypersurface can be fully encoded in a set of tensor fields defined on the abstract manifold $\Sigma$, namely $\mathscr{D}=(\gamma_{ab}, \ell_a, \ell^{(2)}, Y_{ab})$ defined in \eqref{eq:metricData0} and \eqref{eq:tensorY} \cite{Mars:2005ca}.  The description of the geometry of a non-expanding horizon  in terms of these quantities   has  some ``built in''  gauge redundancies, that is, different data sets  $\mathscr{D}$ and   $\mathscr{D}'$ might represent equivalent geometries. 
The following two types of redundancies represent \emph{all the gauge ambiguities} in our framework:

\paragraph{Coordinate freedom on the abstract  manifold.} Recall that the hypersurface data $\mathscr{D}$ is given in terms of tensor fields  living on the abstract manifold $\Sigma$, and thus its definition 
is unaffected by coordinate reparametrisations of $\Sigma$. Nevertheless,  the explicit coordinate expressions of these tensor fields will have a different form in different coordinate systems. This implies that the same hypersurface geometry could be encoded in 
 two different data  sets   $\mathscr{D}$ and $\mathscr{D}'$ 
  which are related to each other through a \emph{diffeomorphism of the abstract manifold} $\Sigma$. These transformations should be regarded as  a gauge freedom  in the hypersurface data\footnote{For a detailed discussion on this type of gauge redundancies in the context of perturbation theory in general relativity see e.g.  \cite{Mars:2005ca}}. 
 Under an arbitrary diffeomorphism on the abstract manifold $\Sigma$, with the explicit form  $\zeta: \xi^a \rightarrow  \zeta^i(\xi^a)$ and $i = 1,2,3$, the coordinate representation of the  data  $\mathscr{D}=(\gamma_{ab}, \ell_a, \ell^{(2)}, Y_{ab})$ transforms as follows
\be
\gamma_{ab}'=\zeta^*\gamma_{ab}, \qquad 
\ell_a' = \zeta^*\ell_a, \qquad
\ell^{(2)}{}' = \zeta^*\ell^{(2)}, \qquad
Y'_{ab} = \zeta^*Y_{ab}.
\label{eq:red1}
\ee   
Here $\zeta^*$ is the pull-back map associated to $\zeta$ which, for example, acts explicitly on the first fundamental form as  $\zeta^*\gamma_{ab}=\gamma_{ij}|_{\zeta(\xi)} \, \zeta_a^i \zeta_b^j$, with  $\zeta_a^i \equiv \pd_a \zeta^i$. 

At this point  it is worth to emphasise that the abstract manifold $\Sigma$ is a redundant object 
detached from the ambient spacetime, which can be seen as a bookkeeping device used  to isolate the geometric data of the hypersurface $\cH\subseteq \cM$. Therefore the diffeomorphisms on $\Sigma$ should not be confused with the diffeomorphisms of the ambient spacetime $\cM$.

 \paragraph{Choice of rigging.} The rigging $\ell$ is an auxiliary vector field over $\cH$ introduced to specify a direction transversal to the hypersurface. 
Given a null hypersurface $\cH$ we could construct the hypersurface data using two different choices of rigging, leading in general  to two different results $\mathscr{D}$ and $\mathscr{D}'$  which obviously represent the same geometry. To characterise  the effect of an arbitrary change of rigging, consider two different choices, $\ell$ and $\ell'$, related to each other by
$\ell' = u (\ell + V)$,  where $u$ is a non-vanishing  scalar function on the hypersurface $\cH$, and $V$ is some vector field tangent to it. Then, from the definition of the  hypersurface data it follows that the elements characterising the metric tensor transform as  \cite{Mars:2013qaa}
\be
\gamma_{ab}'=\gamma_{ab}, \quad
\ell_a' = \uhat(\ell_a + \Vhat^b \gamma_{ab}), \quad
\ell^{(2)}{}' = \uhat{}^2(\ell^{(2)} + 2 \Vhat^a \ell_a + \gamma_{ab} \Vhat^a \Vhat^b). 
\label{eq:red21}
\ee
where $\uhat$ is a function in $\Sigma$ defined by $\uhat \equiv \Phi^*(u)$ and $\hat V \in T_p\Sigma$ is defined by the condition $\dd\Phi(\hat V) = V$. The tensor $Y_{ab}$ describing the transverse connection coefficients is sent to
\be
Y'_{ab}= \uhat Y_{ab} + \ft12 (\pd_a \uhat \, \ell_b + \pd_b \uhat\, \ell_a) + \ft12 \cL_{\uhat \hat V} \gamma_{ab}
\label{eq:red22}
\ee

\subsection{Partial gauge fixing and residual gauge redundancies}
\label{sec:residualGaugeH}
%%%%%%%%%%%%%%%%%%%%%%%%%%%%%%%%%%%%%%%%

In section  \ref{sec:nullGeometry} we introduced the conventions \eqref{eq:metricData}  to reduce the elements of the hypersurface data down to $\mathscr{D} = (q_{MN},\kappa, \Omega_M, \Xi_{MN})$,  fixing some of the redundancies described above. In addition, we can specify a particular form for the metric $q_{MN}$ to fix partially the coordinate reparametrisations on the spatial sections of the horizon.  However, despite all these conventions there is still  some residual gauge freedom.  Indeed, we are still allowed to perform a diffeomorphism on $\Sigma$ \eqref{eq:red1} followed by a change of rigging, \eqref{eq:red21} and \eqref{eq:red22},  as long as they preserve our choice of gauge  \eqref{eq:metricData}. Thus, the combined transformations must satisfy 
\begin{align}
\gamma_{ab}' &=\zeta^*\gamma_{ab} =  \gamma_{ab},\label{eq:gammaRed}\\
\ell_a' &= \uhat(\zeta^*\ell_a +  \zeta^*\gamma_{ab} \Vhat^b ) = \delta_a^1,\label{eq:riggRed1}\\
\ell^{(2)}{}'&=\uhat^2(\zeta^*\ell^{(2)} + 2 \zeta^*\ell_a \Vhat^a + \zeta^*\gamma_{ab} \Vhat^a \Vhat^b) =0. \label{eq:riggRed2}
\end{align}
where $\gamma_{mn}$, $\ell_n$ and $\ell^{(2)}$ are given by \eqref{eq:metricData}.
To determine the form of the allowed diffeomorphisms we begin solving the  first equation \eqref{eq:gammaRed}, which reads explicitly
\be
\zeta_a^i \zeta_b^j \gamma_{ij}|_{\zeta(\xi)} = \gamma_{ab}|_{\xi}\quad \Longleftrightarrow \quad \, \zeta_1^I \zeta_1^J q_{IJ}|_{\zeta(\xi)}= 0, \quad \;Ê\zeta_M^I \zeta_N^J q_{IJ}|_{\zeta(\xi)} = q_{MN}|_\xi,
\label{eq:solvingRed1}
\ee
where $I,J=2,3$. On the one hand, $q_{MN}$ is non-degenerate, and thus the first  equation on the right implies that the components of the diffeomorphism  $\zeta^i(\xi)$ are constant along the null direction,  $\zeta_1^i = 0$. 
On the other hand, since we are restricting the analysis to non-expanding horizons $\pd_n q_{MN}=0$,  the previous equations are independent of the  null coordinate $\xi^1$. Therefore the last equation in \eqref{eq:solvingRed1} implies  that $\zeta^i(\xi^M)$ must define an isometry of the metric $q_{MN}(\xi^M)$, while the component $\zeta^1(\xi) \equiv \hat f(\xi)$ of the diffeomorphism   can be any arbitrary function on $\Sigma$. 

Although the diffeomorphism $\zeta^i(\xi)$ leaves $\gamma_{ab}$ invariant, without a compensating change of rigging, it leads to a   non-trivial transformation of the components $\ell_a$  
\be
\zeta^*\ell_a(\xi) = \zeta_a^i \ell_i(\zeta)\quad \Longrightarrow\quad    \zeta^*\ell_1 =\pd_n \hat f \equiv \hat f_n,\qquad \zeta^*\ell_M =\pd_M \hat f\equiv  \hat f_M,
\ee
while  $\zeta^* \ell^{(2)}=0$, and thus $\ell^{(2)}$ is unchanged. The appropriate  rigging transformation which compensates this change  can be found solving the remaining equations \eqref{eq:riggRed1} and \eqref{eq:riggRed2}. The solution   for $\uhat$ and $\Vhat^a$  reads
\be
\uhat= \frac{1}{\hat f_n}, \qquad \Vhat_M = - \hat f_M, \qquad \Vhat^1 =- \frac{1}{2\hat f_n} \hat f_M \hat f^M.
\label{eq:neqRigging}
\ee
This result together with the form of the diffeomorphism, $\zeta^i(\xi) = (\, \hat f(\xi)\, , \, \zeta^I(\xi^M)\,)$, determines completely the residual gauge redundancies of our description.
Then, using  \eqref{eq:red1} and \eqref{eq:red22} we can  find the how the  form of the tensor $Y_{ab}$ \eqref{eq:tensorY} changes under this gauge transformations
\be
Y'_{ab}= \uhat (\zeta^*Y_{ab}) + \ft12 (\pd_a \uhat \, (\zeta^* \ell_b) + \pd_b \uhat\, (\zeta^*\ell_a)) + \ft12 \cL_{\uhat\hat V} \gamma_{ab} \label{eq:YRed}.
\ee
Here we have used that the induced metric is invariant under the action of $\zeta$ for the last summand (recall (\ref{eq:gammaRed})). 
 Since the freedom to reparametrise the spatial coordinates $\xi^M$ has no interest for our discussion, we  consider for simplicity  diffeomorphisms with $\zeta^I(\xi^M) = \xi^I$, and we find 
 \begin{empheq}[box=\widefboxb]{align}
\begin{split}
\kappa'(\xi) &= \kappa|_{\zeta(\xi)} \, \hat  f_n+ \pd_n \log \hat f_{n},\\
\Omega_M'(\xi)  &= \Omega_M|_{\zeta(\xi)} + \kappa|_{\zeta(\xi)} \, \hat  f_M + \pd_M \log \hat f_{n},\\
\Xi'_{MN}(\xi)  &= \ft{1}{\hat f_n}\big(\Xi_{MN}|_{\zeta(\xi)} -   \Omega_{(M}|_{\zeta(\xi)} \;\hat f_{N)}   -  \kappa|_{\zeta(\xi)} \,\hat  f_M \hat f_N  - D_M \hat f_N \big).
\label{eq:dataRedundancies}
\end{split}
\end{empheq}
Thus, the previous transformations represent the gauge freedom left in the hypersurface data after setting the conventions  \eqref{eq:metricData}.
As these 
transformations leave invariant our conventions \eqref{eq:metricData}, the form of the constraint equations (\ref{eq:Raychad}-\ref{eq:Xi}) is also left unchanged\footnote{The vanishing tensor $J_{ab}$ appearing in the  constraint equations (\ref{eq:Raychad}-\ref{eq:Xi}) can so be shown to be left  invariant under \eqref{eq:dataRedundancies} (see \cite{Mars:2013qaa}).}. This implies that, if a given data set $\mathscr{D}$ is a solution to the constraint equations, then the transformed data set $\mathscr{D}'$ obtained  via  \eqref{eq:dataRedundancies}   will also satisfy the constraints.

Before we close this discussion let us single out the  situation when the diffeomorphism is of the form $\zeta^i(\xi) = (\, \xi^1 + A (\xi^M)\, , \, \xi^I\,)$.  
Then the gauge transformations have the simpler form
\begin{empheq}[box=\widefboxb]{align}
\begin{split}
\kappa'(\xi) &= \kappa|_{\zeta(\xi)},\\
\Omega_M'(\xi) &= \Omega_M|_{\zeta(\xi)} +   \kappa|_{\zeta(\xi)} \, A_M ,\\
\Xi'_{MN}(\xi) &= \Xi_{MN}|_{\zeta(\xi)} -  \Omega_{(M}|_{\zeta(\xi)} \; A_{N)}-  \kappa  |_{\zeta(\xi)} \,  A_M A_N- D_M A_N.
\label{eq:dataST}
\end{split}
\end{empheq}
This case is  particularly interesting because, as we will see in the next subsection, it is closely related to horizon supertranslations. 

Note that a generic change of gauge  also induces a transformation on the elements of the  basis $\cB =\{n,\ell, e_M\}$, and in turn, of the components of any tensor which is expressed in terms of $\cB$. The new basis elements $n'=e_1'$ and $e_M'$ at the point $\xi^a$ can be obtained from their  definition after the acting with the pushforward $\dd \zeta$  on $\hat e_i$, that is, $e_m' = \dd \Phi(\dd \zeta(\hat e_m)) = \dd \Phi(\zeta^i_m \hat e_i)$, and the new rigging $\ell'$ from \eqref{eq:neqRigging}. In the case of transformations of the form  \eqref{eq:dataST}  the basis elements behave as
\be
n' = n|_{\zeta(\xi)}, \quad e'_M = e_M|_{\zeta(\xi)} +  n|_{\zeta(\xi)} \, A_M, \quad  \ell' = \ell|_{\zeta(\xi)} -  e_M|_{\zeta(\xi)}\, A^M - \ft12n|_{\zeta(\xi)}\,   A^M A_M.
\label{eq:basisChange}
\ee

\subsection{Horizon supertranslations and hypersurface data}
\label{sec:supertranslationsGauge}
%%%%%%%%%%%%%%%%%%%%%%%%%%%%%%%%%%%%%%%%

We will now study the effect of horizon supertranslations on the hypersurface data of a generic non-expanding horizon. We will describe supertranslations as \emph{active spacetime diffeomorphisms} $F: \cM \to \cM$ (as opposed to coordinate transformations), which act on the spacetime metric tensor as $g \to F^*g$. That is, $F$  induces a deformation of the metric tensor, while the coordinate charts are left invariant.  
More specifically, we will characterise horizon supertranslations as diffeomorphisms which leave invariant the horizon as a set of points, and which preserve the full metric tensor on it (see \cite{Koga:2001vq,Chung:2010ge,Majhi:2012tf,Hawking:2015qqa,Donnay:2015abr,Donnay:2016ejv,Hawking:2016sgy,Averin:2016hhm,Averin:2016ybl,Hawking:2016msc}). These two conditions can be expressed in a coordinate invariant way as follows    
\begin{empheq}[box=\widefboxb]{align}
F(\cH) = \cH,Ê\qquad \text{and} \qquad F^*g(X,Y)_p = g(X,Y)_p,
\label{eq:defSupertranslation}
\end{empheq}
for all pairs of vectors $X,Y \in T_p \cM$ in the spacetime tangent space at points $p \in \cH$ on the hypersurface. Actually, we shall see that  diffeomorphisms satisfying \eqref{eq:defSupertranslation} lead to a more general class of transformations than supertranslations, and we will refer to them  as \emph{hypersurface symmetries}\footnote{As a consistency check for this approach we have also rederived the full BMS group  at null infinity --including BMS supertranslations-- using similar techniques. See appendix \ref{app:BMSgroup}.}. We will prove that the action of these diffeomorphisms on the NEH data can be described by the transformations  \eqref{eq:dataST}. That is,  hypersurface symmetries, and in particular supertranslations,  leave invariant both the intrinsic and the extrinsic geometry  of the hypersurface up to a gauge redundancy of the description.

  Previous works have used an infinitesimal version of the definition  \eqref{eq:defSupertranslation}, which can be recovered expressing $F$ explicitly  in \eqref{eq:defSupertranslation}, i.e.   in terms of a coordinate system $F:x^\mu \to y^\alpha(x)$ where  $\alpha=0, \ldots, 3$. Then, setting   $y^\alpha(x) \approx x^\alpha + \epsilon \, k^\alpha(x)$ and working at linear order in $\epsilon\ll 1$ we find that the second condition in \eqref{eq:defSupertranslation} reduces to   $\cL_kg_{\mu\nu}|_{\cH}=0$, which is the definition used in \cite{Koga:2001vq,Chung:2010ge,Majhi:2012tf,Hawking:2015qqa,Donnay:2015abr,Donnay:2016ejv,Hawking:2016sgy,Averin:2016hhm,Averin:2016ybl,Hawking:2016msc}. The advantage of  \eqref{eq:defSupertranslation} is that it is coordinate independent, what clarifies the geometrical interpretation of these transformations, 
and that  it can be solved easily  leading  directly to the finite form  of the diffeomorphisms.

\paragraph{Spacetime coordinate system.} In order to find the set of diffeomorphisms satisfying the conditions \eqref{eq:defSupertranslation} we will first define a spacetime coordinate system adapted to $\cH$ to simplify the derivation. As the hypersurface $\cH$ is diffeomorphically identified with the abstract manifold $\Sigma$ via the embedding map $\Phi:\Sigma \to \cM$, we can use the coordinate system $\{\xi^a\}$ on  $\Sigma$ to parametrise points over the hypersurface. Then, we can extend these coordinates away from $\cH$ introducing a transverse coordinate $r$, which we define by the conditions $\ell= \pd_r$ and $r(\cH)=0$, and   requiring the coordinates $\{\xi^a\}$ to be constant along the integral lines of $\ell$. Strictly speaking this procedure would require to define how the rigging is extended off the hypersurface, but  the following calculation is independent of this extension, and thus we will leave it unspecified.  The complete spacetime  coordinate system is given by  $x^\mu\equiv\{u=\xi^1,r,x^M=\xi^M\}$, and therefore it follows that the embedding map is simply
\be
\Phi: \xi^a \longrightarrow x^\mu = \{u =\xi^1, r=0, x^M = \xi^M\}.
\label{eq:embeddingST}
\ee
With these choices the corresponding coordinate basis for the spacetime tangent space $\{\pd_u,\pd_r,\pd_M\}$ coincides with the basis $\cB=\{e_1,\ell,e_M\}$ defined in section \eqref{sec:nullGeometry}  
\be
n= \pd_u, \qquad \ell= \pd_r, \qquad e_M =\pd_M.
\label{STbasis}
\ee
In order to be consistent with our conventions, which require $\bn(\ell)=1$,  the normal form to the hypersurface must be given by $\bn=dr$. In addition, from   \eqref{eq:metricData}  it follows  that the metric tensor  is of the form
\be
g_{\mu\nu}(x)|_{r=0} =\begin{pmatrix}
0 \ \ &1& 0\\
1  \ \ &0&0\\
0  \ \ &0&q_{MN}
\end{pmatrix}, \qquad \text{with}Ê\qquad q_{MN} = q_{MN}(x^{M}),
\label{eq:metricN}
\ee
at points on the hypersurface $\cH$, which is located at $r=0$. In the following we will use ``$\hat =$'' to write equations which hold on the hypersurface $\cH$, that is, at $r=0$.  Note also that the normal vector is given $n = g^{-1}(\bn,\cdot) = e_1$, in consistency with the setting defined in section \ref{sec:nullGeometry}. For later reference, note that the first derivatives of the metric have the form
\be
\pd_u g_{\mu\nu} \; \hat =\; 0, \qquad \pd_L g_{\mu\nu} \; \hat =\;   \begin{pmatrix}
0 \ \ &0& 0\\
0  \ \ &0&0\\
0  \ \ &0&\pd_L q_{MN}
\end{pmatrix}, \qquad  \ft12\pd_r g_{\mu\nu}  \; \hat =\;  \begin{pmatrix}
\cdot \ \ &\cdot& \cdot\\
\cdot  \ \ &-\kappa&-\Omega_M\\
\cdot  \ \ &-\Omega_N &\Xi_{MN}
\end{pmatrix},
\label{eq:metricDers}
\ee
where the last equality follows directly  from the definition of the tensor $Y_{ab}$ \eqref{eq:tensorY}, and the empty entries are those  which cannot be determined from the hypersurface data alone. 
With this information we are ready to find those diffeomorphisms $F$ satisfying the conditions \eqref{eq:defSupertranslation}, and to characterise their action on the hypersurface data.

\paragraph{Hypersurface symmetries.} The first condition on \eqref{eq:defSupertranslation} requires that the  diffeomorphism $F$ maps the hypersurface to itself. Since the null normal $\bn$ of a hypersurface is unique up to a scale, $F$ can only change the normalisation of $\bn$, that is $F^*\bn = \lambda_1 \, \bn$,  where $\lambda_1$ is some function on $\cH$. In the coordinate system defined above this condition has the explicit form
\be
y^\alpha_\mu\, n_\alpha(y(x))  \; \hat =\;  \lambda_1\,  n_\mu(x)  \quad \Longleftrightarrow\quad y^1_\mu  \; \hat =\; \lambda_1\,  \delta_\mu^1 \quad \Longrightarrow \quad \lambda_1  \; \hat =\; y^1_r,
\label{ders1}
\ee
where we are using the shorthand $y^\alpha_\mu\equiv \pd_\mu y^\alpha$.
Moreover, since the metric tensor $g$ should also be preserved by $F$, the normal vector $n= g^{-1}(\bn,\cdot)$ can only change its normalisation under the action of   the diffeomorphism $F$, i.e. $dF(n) = \lambda_2 \,  n$, where $\lambda_2$ is some function over $\cH$.  The explicit form of this condition is
\be
 \lambda_2 \, n^{\alpha}(y(x))  \; \hat =\; y^\alpha_\mu \, n^\mu(x)  \quad \Longleftrightarrow\quad  \lambda_2 \, \delta^\alpha_u  \; \hat =\; y^\alpha_u \quad \Longrightarrow \quad \lambda_2  \; \hat =\; y^0_u,
\label{ders2}
\ee
The following sequence of identities shows that the transformation of the null normal and the normal vector are related by the condition $\lambda_1 = \lambda_2^{-1}$
\bea
1 & \; \hat =\;& g(\ell,  n)  \; \hat =\; F^*g(\ell, n)  \; \hat =\; g(dF(\ell),d F( n))  \; \hat =\; \lambda_2\,  g(dF(\ell),n) \; \hat =\; \nonumber \\
&&  \lambda_2\,  \bn(dF(\ell))  \; \hat =\; \lambda_2 \, F^*\bn(\ell) \; \hat =\; \lambda_1  \lambda_2\,  \bn(\ell) \; \hat =\; \lambda_1  \lambda_2,
\eea
where we have  used the transformation properties of the metric tensor, the normal form and the normal vector, and the definition of the pullback map (see e.g. \cite{Nakahara:2003nw}). Summarising, the first condition in \eqref{eq:defSupertranslation} implies that the diffeomorphism should satisfy the following constraints at points on the hypersurface
\be
\pd_u y^1  \; \hat =\; \pd_M y^1  \; \hat =\;0, \qquad \pd_u y^I  \; \hat =\;0, \quad \text{and}\quad \pd_u y^0  \; \hat =\; 1/\pd_r y^1.
\label{eq:solvST1}
\ee
The second condition in \eqref{eq:defSupertranslation} is satisfied if and only if the diffeomorphism preserves all the scalar products between the elements of the basis $\cB$. From the results obtained above it is straightforward to check that the requirements 
\be
F^*g(n,\ell)  \; \hat =\; g(n,\ell),\quad F^*g(n,n)  \; \hat =\; g(n,n),\quad \text{and} \quad F^*g(n,e_M) \; \hat =\;g(n,e_M)
\ee
are satisfied already without imposing further conditions. In order for $F$ to preserve the remaining scalar products the following equations must hold
\begin{align}
F^*g(e_M,e_N) \; \hat =\;g(e_M,e_N):& \qquad g_{\alpha\beta}' y^\alpha_M y^\beta_N  \; \hat =\; g_{MN} \;  \Longleftrightarrow \hspace{.4cm} q_{IJ}' \, Y^I_M Y^J_N Ê   \; \hat =\; q_{MN},
\label{ders3}\\
F^*g(\ell,e_M) \; \hat =\;g(\ell,e_M):& \qquad  g_{\alpha \beta}' y^\alpha_r y^\beta_M  \; \hat =\; 0 \hspace{.74cm}\Longleftrightarrow \hspace{.5cm}y^I_r  \; \hat =\; -\ft{1}{f_u} f^M Y_M^I,  
\label{ders4}\\
F^*g(\ell,\ell)  \; \hat =\;g(\ell,\ell):& \qquad  g_{\alpha \beta}' y^\alpha_r y^\beta_r \; \hat =\;0 \hspace{.88cm} \Longleftrightarrow \hspace{.51cm} y_r^0 \; \hat =\; - \ft{1}{2 f_u} f^M f_M,
\label{ders5}
 \end{align}
where we have defined $f(u, x^M) \equiv y^0(x)|_{r=0}$,   $Y^I(x^M) \equiv y^I(x)|_{r=0}$, and $f_u\equiv \pd_u f$.  The indices $M,N$ are raised and lowered with the metric $q_{MN}$ and its inverse $q^{MN}$, and  we also used a prime to denote quantities evaluated at $y(x)$, e.g. $g_{\alpha \beta}' = g_{\alpha \beta}(y(x))$.
    
We can conclude that the  action of hypersurface symmetries at points of the hypersurface $\cH$ is completely determined by the following functions
\begin{empheq}[box=\widefboxb]{align}
y^0(x)  \; \hat =\; f(u, x^M), \qquad y^I(x) \; \hat =\; Y^I(x^M),
\label{eq:solST}
\end{empheq}
where $f(u,x^M)$ has an arbitrary dependence on its variables, while the components $Y^I(x^M)$ are independent on $u$ due to \eqref{eq:solvST1}. Moreover equation  \eqref{ders3} implies that $Y^I$ must define an isometry of the $u-$independent metric $q_{MN}$. Note that the conditions \eqref{eq:defSupertranslation} only constrain the form of hypersurface symmetries at points of the hypersurface $\cH$, and then, their  extension away from $\cH$ is arbitrary.

The diffeomorphisms  we just found are more general than the supertranslations discussed in \cite{Chung:2010ge,Majhi:2012tf,Hawking:2015qqa,Donnay:2015abr,Donnay:2016ejv,Hawking:2016sgy,Averin:2016hhm,Averin:2016ybl,Hawking:2016msc}, which correspond to the case when the  functions \eqref{eq:solST} are of the form  
\be
f(u,x^M) = u + A(x^M),\qquad \text{and}\qquad Y^I(x^M) =x^I.
\label{eq:supertranslations}
\ee
Actually, supertranslations can be singled out noting   that, in addition to \eqref{eq:defSupertranslation}, they also also preserve the normalisation of the null normal, namely $F^* \bn =\bn$, which implies $f_u=1$.

\paragraph{Effect on the horizon geometry.} To conclude this section we will  discuss the effect that hypersurface symmetries have on  the   data   $\mathscr{D}=\{\gamma_{ab},\ell_a, \ell^{(2)},Y_{ab}\}$.  Note that, since the diffeomorphisms satisfying \eqref{eq:defSupertranslation} preserve the scalar products on $\cH$, they also leave  invariant  the metric data of the horizon $\{\gamma_{ab}, \ell_a,\ell^{(2)}\}$, given by \eqref{eq:metricData0} and \eqref{eq:metricData}, and in particular its intrinsic geometry. Therefore all that remains to compute is the effect of these diffeomorphisms on the geometric data encoded in the tensor $Y_{ab}$. The form  of the tensor $Y_{ab}$ after a horizon supertranslation can be obtained doing the substitution $g \to F^*g$ in its definition \eqref{eq:tensorY}
\bea
Y_{ab}' &=&\ft12 e_a^\mu e_b^\nu \cL_{\ell} (F^*g)_{\mu \nu} = \ft12 e_a^\mu e_b^\nu \pd_r\left( g_{\alpha\beta}(y)  \; y^\alpha_\mu y^\beta_\nu\right)  \nonumber \\
&=&\ft12 \left(g_{\alpha\beta,\gamma} y^\gamma_r  y^\alpha_\mu y^\beta_\nu + g_{\alpha\beta} y^\alpha_{r \mu} y^\beta_\nu  +g_{\alpha\beta} y^\alpha_\mu y^\beta_{r \nu}   \right)  e^\mu_a e^\nu_b,
\eea
with $e_a = \{n,e_M\}$.
The relevant derivatives of the metric tensor  have been given in \eqref{eq:metricDers}, the first and second derivatives of $y^\alpha(x)$ can be determined from \eqref{eq:solvST1} and (\ref{ders3} - \ref{ders5}), and the basis vector components $e^\mu_a$ are defined in \eqref{STbasis}.  
As in the previous subsection we also assume for simplicity that $Y^I(x^M) = x^I$. After a long but straightforward computation we obtain
\bea
Y_{nn}'(\xi) &=&-\kappa|_{\zeta(\xi)} \, \hat f_n - \pd_n \log \hat f_n,  \nonumberÊ\\
Y_{nM}'(\xi) &=& -\Omega_M|_{\zeta(\xi)} - \kappa|_{\zeta(\xi)} \, \hat f_M - \pd_M \log \hat  f_n, \nonumberÊ\\
Y_{MN}'(\xi) &=&  \ft{1}{\hat f_n}(\Xi_{MN}|_{\zeta(\xi)} - \Omega_{(M}|_{\zeta(\xi)}\;  \hat f_{N)} - \kappa|_{\zeta(\xi)} \, \hat  f_M \hat f_N  -D_M \hat f_N).
\label{eq:Ytransform}
\eea
where $\zeta^i(\xi) =(\hat f(\xi), \xi^I)$ is a diffeomorphism  of the abstract manifold $\Sigma$ defined as  $\zeta \equiv \Phi^{-1} \circ F \circ \Phi$, in terms of the  embedding map $\Phi$ of the hypersurface \eqref{eq:embeddingST}. Note that the diffeomorphism $\zeta(\xi)$ is well defined, as $F$ maps the hypersurface $\cH$ onto itself.
It is straightforward to identify the behaviour of the tensor $Y_{ab}$ under a hypersurface symmetry with the effect of a residual gauge transformation \eqref{eq:dataRedundancies}.  In particular,  the action of supertranslations \eqref{eq:supertranslations} on the horizon data is identical to  \eqref{eq:dataST}.  Thus, in the following  we will make no distinction between horizon supertranslations and the gauge transformations acting as \eqref{eq:dataST}.

The same conclusion can be reached analysing the effect of hypersurface symmetries in the tensor $Y_{ab}$ in terms of the diffeomorphism $\zeta$, without making use of the explicit spacetime coordinate system constructed above. Using the well-known property $\cL_{\ell}(F^*g) = F^*(\cL_{dF(\ell)} g)$
 (see e.g.  \cite{schouten1954ricci}) one finds the following equalities
 \bea
 Y' = \ft12 \Phi^* \cL_{\ell} (F^*g) = \ft12 (F \circ \Phi)^* \cL_{\dd F (\ell)} g = \ft12 (\Phi \circ \zeta)^* \cL_{\dd F (\ell)} g =  \zeta^*(\ft12 \Phi ^* \cL_{\dd F (\ell)} g). \nonumber
 \eea
The quantity in parentheses in the last term can be regarded as the tensor $Y|_{\zeta(\xi)}$ changed under a rigging transformation (\ref{eq:red22}), which is followed by a diffeomorphism in the abstract manifold. These are the two gauge redundancies considered in section \ref{sec:residualGaugeH}. \\

 The results presented in this section have two main consequences: on the one hand hypersurface symmetries, and in particular supertranslations, leave  invariant  both the intrinsic and extrinsic geometries of the horizon up to a gauge redundancy.   On the other hand, since the residual gauge redundancies  \eqref{eq:dataRedundancies}  leave the equations (\ref{eq:Raychad}-\ref{eq:Xi}) unchanged,  diffeomorphisms acting as hypersurface symmetries  also preserve the form of the constraint equations of null hypersurfaces. That is, supertranslations are a symmetry  of the NEH constraint equations. 

The next step is to determine the action of these transformations on the dynamical degrees of freedom on  the horizon. In other words, we need to find a free data set necessary and sufficient to describe the full horizon geometry, and then we have to characterise the action of supertranslations on such data set.
%In other words, we need to find a free data set which is both necessary and sufficient to describe the full horizon geometry, and then we have characterise the action of supertranslations on the free data. 
 Depending on whether these diffeomorphisms  have a non-trivial action on the NEH free data  or not, we will identify them as a large gauge transformations (i.e. global symmetries of the constraint equations) or pure gauge redundancies of our description.   In section \ref{sec:NEH}  we will consider the constraint equations (\ref{eq:Raychad}-\ref{eq:Xi}) for a NEH in order to single out  an  appropriate free data set,  and characterise the geometric nature of  horizon supertranslations. For completeness in section \ref{sec:Scri} we will perform a similar analysis for BMS supertranslations acting on null infinity.

\section{Evolution of non-expanding horizons}
\label{sec:NEH}
%%%%%%%%%%%%%%%%%%%%%%%%%%%%%%%%%%%%%%%%%%%%%%%%%%%%%%%%
%%%%%%%%%%%%%%%%%%%%%%%%%%%%%%%%%%%%%%%%%%%%%%%%%%%%%%%%

In the present section we will study the solutions to the constraint equations of non-expanding horizons embedded in the vacuum, (\ref{eq:transverseCon2}) and (\ref{eq:DNS2}).  Our main objective is to extract a free data set, $\mathscr{D}_{free}$,   necessary and sufficient to reconstruct the full NEH geometry, and then to determine the behaviour of the free data set under supertranslations. Our starting point is the data set \eqref{eq:dataDef} presented in  section \ref{sec:nullGeometry}, which contains \emph{sufficient} information to characterise completely the NEH geometry.   However the data set \eqref{eq:dataDef} still involves some residual gauge freedom, which we characterised in section \ref{sec:residualGaugeH}. Moreover, the data elements of \eqref{eq:dataDef} are not  freely specifiable, as they are subject to the constraints (\ref{eq:transverseCon2}) and (\ref{eq:DNS2}).

Our strategy will be,  first, to reduce the residual gauge freedom  \eqref{eq:dataRedundancies}  imposing  appropriate  gauge fixing conditions, 
so that the only remaining ambiguity  is  \eqref{eq:dataST}, which we associated to  supertranslations in section \ref{sec:supertranslationsGauge}.
Then,  we will turn to the resolution of the constraints (\ref{eq:transverseCon2}) and (\ref{eq:DNS2}), and we will present a free data set, $\mathscr{D}_{free}$, composed of quantities  invariant under supertranslations. Note that, since this free data set involves no unfixed  gauge redundancies, all of its elements are \emph{necessary} to describe the NEH geometry.  Finally, we will check explicitly that our free data set encodes all the information about the spacetime curvature which was contained in the original data set \eqref{eq:dataDef}.  With this result we conclude that the supertranslation invariant data set $\mathscr{D}_{free}$ is both \emph{necessary and sufficient} to  reconstruct the NEH geometry, and that  supertranslations act trivially on the dynamical variables of the NEH.\\

In order to simplify the analysis of the gauge redundancies and constraint equations  \eqref{eq:transverseCon2} let us introduce  the a more convenient set of variables than  \eqref{eq:dataDef}. The hypersurface data element  $\Omega_M$ defines a one-form on the spatial sections of the horizon  $\cS_{\xi^1} \cong \mathbb{S}^2$, and  therefore  we can decompose it uniquely  as the sum of an exact   part $\Omega^e_M$ and a  divergence free part $\Omega_M^0$  
\be
\Omega_M = \Omega^e_M + \Omega_M^0, \qquad \Omega^e_M\equiv  \pd_M \eta, \qquad D^M \Omega_M^0=0.
\label{def:hodge}
\ee
where   $\eta(\xi)$ is a smooth   function of the coordinates.  The previous equation represents the Hodge decomposition of $\Omega_M$ on $\cS_{\xi^1}$, which determines $\Omega_M^0$ uniquely and the potential  $\eta(\xi)$ is defined   up to a shift, $\eta \to \eta + \eta_0$, where $\pd_M \eta_0=0$.  Due to the properties of  the Hodge decomposition, the exact and divergence free part of each side of  the Damour-Navier-Stokes equations \eqref{eq:DNS2} are separately equal, leading to 
\be
 \pd_n  \Omega_M^0=0,  \qquad\qquad  \pd_n \pd_M \eta  = \pd_M \kappa.
\label{eq:NSdec}
\ee
In particular this implies  that the divergence free part of the Hajicek one form is constant along the null direction. Moreover, the  equation on the  right can be solved  in terms of $\eta(\xi)$ requiring it to satisfy $\kappa = \pd_n \eta$, what determines $\eta$  uniquely  up to an additive constant on the horizon. 
The potential $\eta$ can also be defined in a  more covariant way  in terms of a decomposition of the rotation one-form $\omega_a = (\kappa,\Omega_M)$, which we defined in \eqref{eq:defRot}. More specifically,  if $\omega_a$ is a solution of \eqref{eq:DNS2} it can be decomposed as (see \cite{Ashtekar:1999yj,Ashtekar:2001is})
\begin{empheq}[box=\widefboxb]{align}
\omega_a = \pd_a \eta + \omega_a^0, \quad \text{where} \quad  \omega^0(\hat n) =0 \quad \text{and} \quad D^M \pd_M \eta = D^M \Omega_M.
\label{eq:defEta}
\end{empheq}
Here  $\omega^0_a$ is uniquely determined to be  $\omega_a^0=(0,\Omega_M^0)$, and $\eta(\xi)$ is defined up to a constant shift. 
As we shall see in section \ref{sec:solvEq} the potential $\eta$ will play an essential role to solve \eqref{eq:transverseCon2} in terms of quantities which are  invariant under supertranslations.   Finally, let us  also introduce the following combination\footnote{The quantity $\Sigma_{MN}^0$ can be defined covariantly in terms of the rotation one form $\omega_a$. Using the connection $\overline \nabla$ defined by eq. (17) of reference \cite{Mars:2013qaa}  we have  $\Sigma_{ab}^0 \equiv  \ft12\overline  \nabla_{(a} \omega_{b)} + \omega_a \omega_b$.   } 
\begin{empheq}[box=\widefboxb]{align}
\Sigma_{MN}^0 \equiv \ft1 {2}  D_{(M}\Omega_{N)} +   \Omega_M \Omega_N + \kappa\, \Xi_{MN},
\label{eq:invariantXi}
\end{empheq}
together with   its trace $\theta^0 \equiv \Sigma^{0 \, M}_M$  and its traceless part $\sigma^0_{MN}$.
 From the definitions \eqref{eq:defEta} and \eqref{eq:invariantXi}, it is straightforward to check that the NEH data  \eqref{eq:dataDef} can be equivalently encoded in 
a new set of quantities $\mathscr{D}_s$
\be
\mathscr{D} = (q_{MN},\quad \kappa, \quad \Omega_M, \quad \Xi_{MN}) \quad \longrightarrow\quad \mathscr{D}_s = (q_{MN},\quad \eta, \quad \Omega_M^0, \quad \Sigma_{MN}^0).
\label{eq:invDataSet}
\ee  
As we will see below the new data set $\mathscr{D}_s$  has particularly simple transformation properties under supertranslations.

\subsection{Reduction of the gauge freedom.}
%%%%%%%%%%%%%%%%%%%%%%%%%%%%%%%%%%%%%%%%

 We now introduce the relevant gauge fixing conditions so that the residual gauge freedom \eqref{eq:dataRedundancies} reduces to the transformations \eqref{eq:dataST}, which we identified with supertranslations. Given an arbitrary non-expanding horizon with surface gravity $\kappa$, it is always possible to choose a gauge  where the  surface gravity is a constant $\kappa_0$ over the horizon 
\be
\text{\emph{Gauge condition 1:}}\qquad \qquad \pd_n \kappa_0 = \pd_M \kappa_0=0 \quad \text{and} \quad \kappa_0 >0  \quad \text{for all   $\xi \in \Sigma$}
\label{eq:GaugeFix1}
\ee
 This gauge can be achieved making a transformation of the form \eqref{eq:dataRedundancies} with the function $\hat f(\xi)$ satisfying
\be
\kappa(\xi)\,  \hat f_n + \pd_n \log \hat f_{n} =\kappa_0, 
\label{eq:GaugeFix1bis}
\ee
  which can   always be solved for  $\hat f(\xi)$.  It is straightforward to check that, in this gauge, the equation \eqref{eq:DNS2} implies that the full Hajicek one-form $\Omega_M$ --and not only $\Omega_M^0$-- must be constant along the null direction of the horizon, $\pd_n \Omega_M=0$.

The previous gauge fixing condition  still does not  reduce the redundancies down to supertranslations. Indeed, after imposing the  condition \eqref{eq:GaugeFix1} on the data, the remaining gauge freedom can be found solving again  equation \eqref{eq:GaugeFix1bis}, but this time  setting  $\kappa(\xi)=\kappa_0$, which gives
\be
\hat f(\xi) =\xi^1 +A(\xi^M)+ \ft{1}{\kappa_0} \log\Big( 1+ B(\xi^M)\rme^{-\kappa_0 \xi_ 1}\Big), 
\label{eq:resGauge1}
\ee
where $B(\xi^M)$ and $A(\xi^M)$ are smooth functions satisfying $\pd_n A = \pd_nB =0$.  At this point we can already identify $A(\xi^M)$ with the freedom to perform a supertranslation \eqref{eq:dataST}. Therefore it only remains to
find a convention to eliminate ambiguity associated to  $B(\xi^M)$, which reflects the fact that the gauge condition \eqref{eq:GaugeFix1bis} does not  determine completely the normalisation of the null normal $\bn$.  The analysis in the following sections is independent of the the actual method to fix the normalisation of $\bn$, what  was discussed for example in \cite{Ashtekar:1999yj,Ashtekar:2001jb,Gourgoulhon:2005ng}. Here we will follow the strategy of \cite{Ashtekar:2001jb}, which consists in imposing a gauge fixing condition on $\theta^0 = \Xi_M^M$. 
Contracting the constraint equation for $\Xi_{MN}$ \eqref{eq:transverseCon2} with  $q^{MN}$, and using that the surface gravity $\kappa_0$ and $\Omega_M$ are constant along $\xi^1$ we find
\be
\pd_n \theta^0 = - \kappa_0 \,( \theta^0  -\ft{1}{2}  \cR),
\label{eq:theta0}
\ee
 where  the  quantity\footnote{The quantity $\theta^0$  should not be confused with the trace of the second fundamental form, the expansion $\theta \equiv \Theta^{M}_M$, which is zero for non-expanding horizons.} $\theta^0$ was defined in \eqref{eq:invariantXi}, and $\cR$ is the scalar curvature associated to $q_{MN}$. 
Note that in this equation only $\theta^0$ has a non trivial dependence on $\xi^1$, since  $\pd_n q_{MN} =0$  also implies $\pd_n \cR=0$, and thus it can be integrated easily
\be
\theta^0(\xi)  = (\theta^0|_{\xi^1_0} - \ft12 \cR)\,  \rme^{-\kappa_0(\xi^1- \xi^1_0)} + \ft{1}{2} \cR.
\ee
As was discussed in \cite{Ashtekar:2001jb}, for a  subclass of non-expanding horizons, known as \emph{generic non-expanding horizons},  it is possible perform a transformation of the form \eqref{eq:resGauge1} in order to make $\theta^0$ stationary on the horizon
\be
\hspace{-4cm} \text{\emph{Gauge condition 2:}} \qquad  \qquad \pd_n \theta^0=0 \qquad \text{for all   $\xi \in \Sigma$.} 
\label{eq:GaugeFix2}
\ee
This condition is trivially satisfied by all black holes in the Kerr family, which are stationary, and thus they have  generic horizons  (see e.g. appendix D in \cite{Gourgoulhon:2005ng}). As we show in appendix  \ref{app:gaugeFixB}, if the hypersurface data of a NEH satisfies the  following condition on a spatial section $\cS_{\xi^1}$ of $\Sigma$
\be 
\Omega^0{}^M \Omega^0_M \le    \ft12 \cR\le \theta^0 \qquad \text{for all   $\xi \in \cS_{\xi^1}$},
\label{eq:generic}
\ee
the horizon can be shown to be    generic. In other words,  it is possible to find a gauge transformation of the form \eqref{eq:resGauge1} which allows to set $\pd_n \theta^0=0$ everywhere on $\Sigma$. Moreover,   the residual gauge freedom left after imposing this gauge fixing condition is precisely that of  supertranslations \eqref{eq:dataST}.  To the best of our knowledge
the condition \eqref{eq:generic} has not been presented before in the literature. In the following we will restrict ourselves to horizons where the gauge \eqref{eq:GaugeFix2} can be attained.\\

After imposing the gauge fixing conditions \eqref{eq:GaugeFix1} and \eqref{eq:GaugeFix2}, the only remaining gauge freedom are supertranslations  \eqref{eq:dataST}.  We will now characterise the behaviour   under supertranslations of the elements in the data set  $\mathscr{D}_s$  \eqref{eq:invDataSet}. The transformation properties of the spatial metric $q_{MN}$ were discussed in section \ref{sec:residualGaugeH}, and it was shown to be invariant under \eqref{eq:dataST}.    Since the surface gravity $\kappa_0$ has been set to a constant,  the exact and divergence free parts of  the Hajicek one-form transform under  supertranslations as
\be
\Omega_M^e{}'(\xi)  = \Omega_M^e|_{\zeta(\xi)} + \kappa_0  A_M, \qquad \qquad \Omega_M^0{}'(\xi)  = \Omega_M^0|_{\zeta(\xi)},
\label{eq:STdec}
\ee
where $\zeta^a = (\xi^1 + A(\xi^M), \xi^M)$. Actually, due to \eqref{eq:NSdec} 
 the  divergence free part of the Hajicek one-form satisfies $\pd_n\Omega_M^0=0$, and thus  $\Omega_M^0$ is completely invariant. In this gauge the functional form of the potential $\eta(\xi)$ reduces to  $\eta(\xi) = \kappa_0 \xi^1  + h(\xi^M)$, with $\pd_n h=0$,  and  it behaves   under supertranslations as\footnote{The behaviour of $\eta(\xi)$ under \eqref{eq:dataST} should be derived from  its definition in  \eqref{eq:defEta}.} 
\be
\eta(\xi)  \to \eta'(\xi) = \kappa_0 \xi^1 + h(\xi^M) + \kappa_0 A(\xi^M).
\label{eq:etaTransf}
\ee
The previous expression can also be written as $\eta'(\xi)=\eta(\zeta(\xi))$, which means that   $\eta$ transforms under a supertranslation as a scalar field.
 Finally, it is easy to check that the object  $\Sigma_{MN}^0$,  defined in \eqref{eq:invariantXi}, also   transforms  as a scalar under  \eqref{eq:dataST} 
 \be
\Sigma_{MN}^0{}'(\xi) = \Sigma_{MN}^0|_{\zeta(\xi)}.
\label{eq:SigmaSTtrans}
\ee
Thus,  from \eqref{eq:STdec}, \eqref{eq:etaTransf} and \eqref{eq:SigmaSTtrans}, it follows that all the elements of the NEH data set $\mathscr{D}_s$  \eqref{eq:invDataSet} are either invariant or transform as a scalar field under supertranslations.

\subsection{Resolution of the hypersurface constraint equations}
\label{sec:solvEq}
%%%%%%%%%%%%%%%%%%%%%%%%%%%%%%%%%%%%%%%%

In this subsection we will consider the constraint equations of the NEH, and we will present a free data set $\mathscr{D}_{free}$  composed of quantities which are all invariant under supertranslations. 

The data set $\mathscr{D}_s$  \eqref{eq:invDataSet} is subject to the following complete set of constraint equations and gauge fixing conditions
\be
\pd_n q_{MN} =0, \qquad \pd_n \Omega_M^0=0, \qquad  \theta^0  = \ft1{2} \cR, \qquad \pd_n  \sigma^0_{MN} = -\kappa_0 \,  \sigma^0_{MN},
\label{eq:finalEqs}
\ee
and the potential has to be of the form  $\eta(\xi) = \kappa_0 \xi^1 + h(\xi^M)$, where $\pd_a \kappa_0=\pd_n h=0$, due to \eqref{eq:GaugeFix1}.
Let us recapitulate the origin of these equations from left to right: the first one is a consequence of the non-expanding condition and the Raychaudhuri equation, \eqref{eq:nonExp}; the second one follows from the DNS equation \eqref{eq:NSdec}; the third and fourth ones are obtained  expressing  the eq. \eqref{eq:transverseCon2} for the transverse connection in terms of $\Sigma_{MN}^0$  \eqref{eq:invariantXi}, and then decomposing it in its trace $\theta^0$, and traceless $\sigma_{MN}^0$ parts. To obtain the last two equations we also used the gauge fixing conditions \eqref{eq:GaugeFix1} and \eqref{eq:GaugeFix2}.

The equations \eqref{eq:finalEqs} imply that the full NEH geometry can be reconstructed from an initial data set  specified on a spatial section $\cS_{\xi^1_0}$ of the horizon
\be
(q_{MN}|_{\cS_{\xi^1_0}}, \quad \Omega_M^0|_{\cS_{\xi^1_0}}, \quad \sigma^0_{MN}|_{\cS_{\xi^1_0}}),
\label{eq:freeData1}
\ee
and providing, in addition,  the scalar potential $\eta(\xi)= \kappa_0 \xi^1 + h(\xi^M)$.  This initial data set, and the quantities $\kappa_0$ and $h(\xi^M)$, can be chosen freely on a given spatial slice $\cS_{\xi^1_0}$, and then the geometry over the entire NEH can be obtained solving \eqref{eq:finalEqs}.  

As we discussed above, all the elements in the data set $\mathscr{D}_s$ transform as scalar fields under supertranslations \eqref{eq:dataST}, implying that the initial data \eqref{eq:freeData1} may still involve gauge dependent quantities. Let us examine the transformation properties of the elements in  \eqref{eq:freeData1} under supertranslations:
 \begin{itemize}
\item  The gauge freedom  \eqref{eq:dataST} is defined in terms of \emph{active diffeomorphisms} on $\Sigma$, which transform the NEH data but leave the coordinate system unchanged.  In consequence, the initial slice  $\cS_{\xi^1_0}$ (defined by $\xi^1 = \xi^1_0$) does not transform under  supertranslations. 
 \item  The objects $q_{MN}$ and  $\Omega_M^0$ 
 do not depend on the null coordinate  due to   \eqref{eq:finalEqs}, and thus   $q_{MN}|_{\cS_{\xi^1_0}}$ and  $\Omega_M^0|_{\cS_{\xi^1_0}}$   are invariant under supertranslations. 
 \item  The   potential $\eta(\xi) = \kappa_0 \xi^1 + h(\xi^M)$  has a non trivial dependence on the null coordinate,  and so does  $\sigma_{MN}^0$ unless it is strictly zero (see eq. \eqref{eq:finalEqs}).  Therefore, in general, both the potential $\eta(\xi)$, and the initial value $\sigma_{MN}^0|_{\cS_{\xi^1_0}}$ will transform non-trivially under supertranslations. 
\end{itemize}  
At this point, we could impose  an appropriate gauge fixing conditions to eliminate the  ambiguity associated with the transformations\footnote{This is the approach used in the  membrane paradigm for the description of black holes (e.g. see appendix D in  \cite{Price:1986yy}). Other examples of this method are  reviewed in \cite{Gourgoulhon:2005ng}.  } \eqref{eq:dataST}. However, following the strategy used for null infinity in \cite{Ashtekar:1981hw},  we will deal with this redundancy introducing a supertranslation invariant free data set, and proving that it contains the same information about the spacetime geometry as the original data  \eqref{eq:dataDef}. This is a rigorous way to ensure that we do not exclude physically allowed configurations of the NEH.

\paragraph{Supertranslation invariant data.} 
 
In order to define a  free data set which is composed of quantities invariant under supertranslations  it is convenient to parametrise the null direction of the horizon using the scalar potential $\eta(\xi)$. Note that this parametrisation is well defined, since  $\eta(\xi)$ increases monotonically  along the null direction everywhere in $\Sigma$, i.e. $\pd_n \eta = \kappa_0 >0$.  Then, using the fact that both $\sigma_{MN}^0(\xi)$ and $\eta(\xi)$ transform as scalar fields under supertranslations, we can construct  a supertranslation invariant variable expressing the evolution of $\sigma^0_{MN}$ along the null direction in terms  of $\eta$.
For this purpose, let us  write the null coordinate in terms of $\eta$ as $\xi^1 = H(\eta,\xi^M)$, 
\be
\eta(\xi)  = \kappa_0 \xi^1 + h(\xi^M) \qquad \Longrightarrow \qquad H(\eta,\xi^M) = \ft{1}{\kappa_0} \big(\eta - h(\xi^M)\big).
\label{eq:defH}
\ee
As the potential $\eta(\xi)$ changes under supertranslations \eqref{eq:etaTransf}, the inverse function $H(\eta,\xi^M)$ needs to transform accordingly
\be
H(\eta,\xi^M) \to H'(\eta,\xi^M) = \ft{1}{\kappa_0}\big(\eta - h(\xi^M)\big) - A(\xi^M).
\ee
Thus, we can characterise the evolution of $\sigma_{MN}^0$ along the null coordinate in terms of the following   \emph{supertranslation invariant} variable 
 \begin{empheq}[box=\widefboxb]{align}
s_{MN} (\eta, \xi^M) \equiv \sigma^0_{MN} \big(H(\eta,\xi^M), \xi^M\big).
\label{eq:defSigmaGI}
\end{empheq}
To prove that this object is invariant under  \eqref{eq:dataST} we just need to use its definition in combination with the transformation properties of $\sigma_{MN}^0$ and the function $H(\eta,\xi^M)$ under \eqref{eq:dataST}
\bea
s_{MN}'(\eta,\xi^M) &=& \sigma^0_{MN}{}' \big(H'(\eta,\xi^M), \xi^M\big) = \sigma^0_{MN}{} \big(H'(\eta,\xi^M) + A(\xi^M), \xi^M\big) = \nonumber\\
&=& \sigma^0_{MN}{} \big(H(\eta,\xi^M), \xi^M\big) = s_{MN}(\eta,\xi^M).
\label{eq:sMNinvariant}
\eea 
Therefore, the transformed form of  $s_{MN}' (\eta, \xi^M)$ after a supertranslation is the same function of $\eta$ as $s_{MN} (\eta, \xi^M)$, which proves that this object is completely  invariant under the action  of supertranslations.

 Following a similar line of argument it can also be shown that 
$s_{MN}$ is invariant under gauge transformations \eqref{eq:dataRedundancies} with $\zeta^a(\xi) =( \lambda \xi^1 , \xi^M)$, where $\lambda$ is a constant over the horizon.

\paragraph{Solution to the constraint equations.} The constraint equation  for $s_{MN}(\eta,\xi^M)$ is obtained expressing the last equation in  \eqref{eq:finalEqs} in terms of $\eta$ and $s_{MN}$. Using that $\pd_\eta H = 1/\kappa_0$ we find 
\be
\pd_\eta s_{MN} =  \pd_\eta H \,  \pd_n \sigma_{MN}^0|_{H(\eta)}\qquad  \Longrightarrow\qquad  \pd_\eta s_{MN} = - s_{MN},
\label{eq:GIequations}
\ee
which  has the general solution  
\be
s_{MN}(\eta,\xi^M) =  s_{MN}|_{\eta_0}  \; \rme^{- \, (\eta-\eta_0)},
\label{eq:finalSol}
\ee
and  $\eta_0$ is an arbitrary  constant which reflects the ambiguity in the definition of $\eta$. This  ambiguity can also be eliminated imposing an additional condition on the data, e.g. the normalisation
\be
\frac{1}{a_\cH} \oint_{\cS_{\eta=0}} d\xi^2  q\; s_{MN} s^{MN}= 1,
\ee
 where the integral is over the spatial section $\cS_{\eta=0}$ of the horizon, $q\equiv \sqrt{\det(q_{MN})}$ and  $a_\cH$ is the area of $\cS_{\eta=0}$.

 The result \eqref{eq:finalSol}, together with the equations \eqref{eq:finalEqs}, imply that the full NEH geometry can be encoded in the functional form of the potential $\eta = \kappa_0 \xi^1 + h(\xi^M)$, combined with the following free data set  
\begin{empheq}[box=\widefboxb]{align}
\text{Free horizon data:} \qquad  \mathscr{D}_{free}\equiv (q_{MN}|_{\cS_{\eta_0}},  \quad \Omega_M^0|_{\cS_{\eta_0}}, \quad s_{MN}|_{\cS_{\eta_0}}),
\label{eq:completeFreeData}
\end{empheq}
which is specified on a spatial slice  of the horizon $\cS_{\eta_0}$ defined by  $\eta = \eta_0$. The full NEH geometry can be recovered from these quantities  using that $q_{MN}$ and $\Omega_M^0$  are constant along the null  direction of the horizon and \eqref{eq:finalSol}. In particular $q_{MN}|_{\cS_{\eta_0}}$ determines the intrinsic geometry of the NEH, and  $\Omega_M^0|_{\cS_{\eta_0}}$  can be associated to its angular momentum aspect when $q_{MN}$ admits an $SO(2)$ isometry  (see \cite{Gourgoulhon:2005ng}). It is also interesting to note that  $s_{MN}$ (which is symmetric and traceless) has two independent components, matching the number of radiative degrees of freedom of the gravitational field.

Due to  \eqref{eq:finalEqs}, the first two elements  of \eqref{eq:completeFreeData}, $q_{MN}|_{S_{\eta_0}}= q_{MN}|_{S_{\xi^1_0}}$ and   $\Omega_{M}^0|_{S_{\eta_0}} =  \Omega^0_{M}|_{S_{\xi^1_0}}$, coincide with the first two elements in \eqref{eq:freeData1}, which we argued to be invariant  under supertranslations. Moreover, the third element $s_{MN}|_{\cS_{\eta_0}}$ is also invariant under \eqref{eq:dataST} due to \eqref{eq:sMNinvariant}, implying that none of the elements in \eqref{eq:completeFreeData} involve any unfixed gauge freedom.  As a consequence distinct data sets $\mathscr{D}_{free}$  generate gauge inequivalent NEH structures, and thus, the corresponding spacetime geometries must be different as well.   In other words, all the  elements in $\mathscr{D}_{free}$  are  \emph{necessary}  to characterise completely the geometry of the NEH.   Note, however, that  the potential $\eta = \kappa_0 \xi^1 + h(\xi^M)$ is also part of the NEH initial data set, and it transforms non-trivially under supertranslations \eqref{eq:etaTransf}.

\subsection{Free horizon  data  and the spacetime curvature.}
\label{sec:NewmanPenrose}
%%%%%%%%%%%%%%%%%%%%%%%%%%%%%%%%%%%%%%%%

We will now show explicitly that the free data set $\mathscr{D}_{free}$ \eqref{eq:completeFreeData} encodes all the information about the  curvature of the ambient spacetime $\cM$   contained in the original data set \eqref{eq:dataDef}.  In other words,  the supertranslation invariant data set \eqref{eq:completeFreeData} is both \emph{necessary and sufficient}  to  characterise the entire  NEH geometry, and thus no knowledge about the functional form of  the potential $\eta(\xi)$ is required.  In addition, using the Newman-Penrose formalism, we will argue  that $\mathscr{D}_{free}$ is sufficiently general to  describe radiative processes taking place at the horizon.

In  the case of horizons embedded in vacuum, $T_{\mu \nu}=0$, the curvature of the ambient spacetime is completely described by the Weyl tensor $R_{\mu\nu\rho\sigma} = C_{\mu\nu\rho\sigma}$, or equivalently, by  the five Weyl scalars $\Psi_n$, with  $n=0,\ldots,4$. Since the connection coefficients  \eqref{eq:connectionCoeff} only characterise the spacetime connection along the horizon they only constrain  four of the Weyl scalars. Indeed, the computation of $\Psi_4$ requires knowledge about the spacetime connection off the hypersurface, which is not available in  \eqref{eq:connectionCoeff}.  Therefore, all we need to show is that all the information about the spacetime curvature contained  in the  Weyl scalars $\Psi_n$, with $n=0,\ldots, 3$ is also encoded in the free data set $\mathscr{D}_{free}$ \eqref{eq:completeFreeData}. 
 
The computation of the Weyl scalars can be done as described in  section \ref{sec:WeylScalars}. First, without loss of generality, we specify an arbitrary point $\xi^M_0$ on the spatial sections $\cS_{\xi^1}$ of the horizon, choosing the coordinates $\xi^M$ such that\footnote{This choice determines  our gauge fixing conventions \eqref{eq:metricData} at $\xi^a=\xi_0^a$. As discussed in section \ref{sec:redundancies} this choice is preserved by the residual gauge redundancies \eqref{eq:dataRedundancies},  including supertranslations. } $q_{MN}(\xi_0) = \delta_{MN}$. Then, the Weyl scalars can be obtained from    \eqref{eq:defWeyl}  contracting the Weyl tensor 
with the elements  of the Newman-Penrose tetrad $\cB_{NP} = \{n,\ell,m,\overline  m \}$. 
The explicit expressions for the Weyl scalars $\Psi_n(\xi^1,\xi^M)$ are functions of the coordinates $\xi^a$, and therefore they are not invariant under diffeomorphisms of the abstract manifold.  However, if we impose the gauge fixing conditions   \eqref{eq:metricData}, \eqref{eq:GaugeFix1} and \eqref{eq:GaugeFix2} the only remaining gauge transformations are supertranslations, $\xi^1 \to \xi^1 + A(\xi^M)$. Thus, in order  to deal with this  freedom   we  introduce  the \emph{gauge corrected} Weyl scalars
\be
\Psi_n^{c}(\eta, \xi^M) \equiv   \Psi_n(H(\eta,\xi^M), \xi^M).
\label{eq:GCWdefs}
\ee
In the case  of non-expanding horizons embedded in vacuum these quantities read 
\begin{empheq}[box=\widefboxb]{align}
\Psi_0^c = \Psi_1^c=0,\qquad \Psi_2^c = -\ft14 \cR + \ft\rmi{2} \cJ, \qquad \text{with}  \qquad \cJ \equiv D_{[2} \Omega_{3]}^0,
\label{eq:Psi012final}
\end{empheq}
 and  
 \be
\Psi^c_3 =\ft1 {\kappa_0\sqrt{2}} \Big[ D  s |_{\eta_0} \; \rme^{-(\eta-\eta_0)}  + \hat D \Psi_2^c + 3 \, \hat \Omega^0 \Psi_2^c  + 3 \, \hat \Omega^e \Psi_2^c  \Big].
\label{eq:Psi3final}
\ee
Here we have used the notation $ \hat D \equiv  (D_{2} -\rmi D_{3})$, and $\hat \Omega \equiv (\Omega_2 - \rmi \Omega_{3})$. We  have also defined the complex field
\be
Ds(\eta,\xi^M)  \equiv D^M s_{M2} +  \Omega^{0|M} s_{M2} - \rmi  (D^M s_{M3} + \Omega^{0|M} s_{M3}).
\ee
The details of the computation 
can be found in  appendix \ref{app:weyl}.
The  Weyl scalars $\Psi_0^c$ and $\Psi_1^c$,  represent gravitational radiative modes 
which propagate  into the horizon, and their vanishing can be seen as a consistency condition for the horizon to be non-expanding.  The scalar $\Psi_2^c$ encodes the coulomb contribution of the gravitational field and,  when $q_{MN}$ admits an axial killing vector field,   $\cJ$ characterises  the angular momentum aspect of the NEH (see \cite{Gourgoulhon:2005ng}). Finally, the Weyl scalars     $\Psi_3^c$ and $\Psi_4^c$  can be associated to radiative modes propagating along the horizon. 

As we explained above, the structure of the NEH does not constrain the value of the fourth Weyl scalar, and  thus  a priori $\Psi_4^c$ can take any value on $\cH$.   Since we also have  $\Psi_0^c=\Psi_1^c=0$ and  $\Psi_2^c, \Psi_3^c \neq 0$, 
we can  conclude that  the Weyl tensor on a NEH will be generically of Petrov type II (see \cite{Chandrasekhar:1985kt,Gourgoulhon:2005ng}). Then, in general,  
the gravitational field on the NEH  will  contain a radiative component   \cite{Ashtekar:2000hw}. In other words, the NEH structure  is sufficiently general to allow for the presence of gravitational radiation on the horizon.  \\

It is straightforward to check that $\Psi_0^c$, $\Psi_1^c$ and $\Psi_2^c$ are invariant under supertranslations, and that $\Psi_2^c$ can be computed from the elements in $\mathscr{D}_{free}$. However, the expression \eqref{eq:Psi3final} for $\Psi^c_3$ still involves gauge dependent quantities which are not part of the free data set \eqref{eq:completeFreeData}, namely, the surface gravity $\kappa_0$, which depends on the normalisation of the null normal, $\bn$, and the exact part of the Hajicek one form $\Omega_M^e$ which transforms under supertranslations. We will now show that both quantities can be associated to well known gauge redundancies of the Newman-Penrose formalism, i.e. the freedom to perform  rotations of the null tetrad $\cB_{NP}$.  That is, neither $\kappa_0$ or $\Omega_M^e$   involve any information about the spacetime geometry. More specifically, we will prove that the expression \eqref{eq:Psi3final} for $\Psi_3^c$ represents the same spacetime geometry as
  \begin{empheq}[box=\widefboxb]{align}
\Psi^c_3 =\ft1 {\sqrt{2}} \Big[ D  s |_{\eta_0} \; \rme^{-(\eta-\eta_0)}  + \hat D \Psi_2^c + 3 \, \hat \Omega^0 \, \Psi_2^c\Big],
\label{eq:Psi3TrueFinal}
\end{empheq}
 which is completely determined by the elements of $\mathscr{D}_{free}$. Both expressions for the third Weyl scalar \eqref{eq:Psi3final}  and \eqref{eq:Psi3TrueFinal}  are projections of the \emph{same Weyl tensor} associated to two different null tetrads $\cB_{NP}$.
 
We will begin  considering the surface gravity $\kappa_0$. The Newman-Penrose formalism has an inherent gauge freedom associated to the choice of null tetrad $\cB_{NP}$ which, in a general setting, is only required to satisfy the orthogonality and normalisation conditions \eqref{eq:nullTetradProds}. Thus, when there are no further restrictions, it is possible to perform the following redefinition of the null tetrad $\cB_{NP}$ which preserves \eqref{eq:nullTetradProds} 
\be
n'  = \lambda n, \qquad \ell' = \lambda^{-1} \ell, \qquad m'=m, \qquad \overline m' =\overline m,
\label{eq:tetradScaling}
\ee
where $\lambda$ is real scalar field on $\Sigma$. From \eqref{eq:defWeyl} it is immediate to check that if we transform the null tetrad as in  \eqref{eq:tetradScaling} --keeping  the Weyl tensor fixed-- the gauge corrected  Weyl scalars behave as (see section 8 in \cite{Chandrasekhar:1985kt})
\be
\Psi_0^c{}' =  \Psi_0^c{}=0, \qquad \Psi_1^c{}'   =\Psi_1^c=0, \qquad \Psi_2^c{}'= \Psi_2^c, \qquad \Psi_3^c{}' =\lambda^{-1} \Psi_3^c. 
\label{eq:WeylScaling}
\ee 
That is,  $\Psi_n^c$ and $\Psi_n^c{}'$ represent contractions of the same Weyl tensor with the elements of two different tetrads, $\{n,\ell, m, \bar m\}$ and $\{n',\ell', m', \bar m'\}$ respectively, and thus the two sets of Weyl scalars describe the same spacetime geometry.
In our setting,  the null tetrad is fully determined by the  elements in the basis $\cB$ adapted to $\cH$, and thus the rotations \eqref{eq:tetradScaling} must always be associated  
to a gauge transformation \eqref{eq:dataRedundancies}  for consistency with the definition of $\cB$.
 Actually, it is possible to implement a rotation of the null tetrad of the  form \eqref{eq:tetradScaling}  performing a transformation \eqref{eq:dataRedundancies}  with  $\zeta^a(\xi) = (\lambda \xi^1, \xi^M)$, where $\lambda>0$ is an arbitrary positive constant. The corresponding change in  the hypersurface data can be derived from \eqref{eq:dataRedundancies}, and the definition of $s_{MN}$ \eqref{eq:defSigmaGI} 
\be
\cR'=\cR, \qquad   \kappa_0' =\lambda\,  \kappa_0,\qquad  Ds' = Ds, \qquad \hat \Omega'=\hat \Omega. 
\label{eq:rotationIII}
\ee 
Then, using this data to compute the transformed Weyl scalars \eqref{eq:Psi012final} and  \eqref{eq:Psi3final} it is straightforward to check that the behaviour of $\Psi_n^c$  under these transformations is precisely   \eqref{eq:WeylScaling}, i.e. it is indistinguishable from the effect of a null tetrad rotation. This proves  explicitly that  gauge transformations with  $\zeta^a(\xi) = (\lambda \xi^1, \xi^M)$ and  $\pd_a\lambda=0$ leave invariant the spacetime geometry, and thus $\kappa_0$ can be set to any arbitrary value in  \eqref{eq:Psi3final} without changing the geometric information encoded in   $\Psi_3^c$.

We will now discuss the role of the exact part of the Hajicek one-form $\Omega^e_M$ in \eqref{eq:Psi3final}. Consider the following redefinition  of the null tetrad $\cB_{NP}$ which preserves the scalar products \eqref{eq:nullTetradProds}
\be
n' = n, \qquad m' = m + a  \, n, \qquad  \overline m' = \overline m + \overline a  \, n, \qquad  \ell' = \ell - \overline a \,  m -  a \, \overline  m  - a  \overline a \, n,
\label{eq:WeyltetradST}
\ee
where  $a= a_2(\xi) + \rmi a_3(\xi)$ is a complex valued function on $\Sigma$. Under this change of null tetrad, and keeping the Weyl tensor fixed,  $\Psi_n^c$ behave as (see  \cite{Chandrasekhar:1985kt})
 \be
\Psi_0^c{}' =  \Psi_0^c{}=0, \qquad \Psi_1^c{}'   =\Psi_1^c=0, \qquad \Psi_2^c{}'= \Psi_2^c, \qquad \Psi_3^c{}' =\Psi_3^c +3 \overline a\,  \Psi_2^c. 
\label{eq:WeylST}
 \ee 
In our framework it can be shown, using \eqref{eq:mbarm} and \eqref{eq:basisChange},  that  supertranslations $\xi^1  \to \xi^1 + A(\xi^M)$  induce a rotation of the null tetrad $\cB_{NP}$ of the form \eqref{eq:WeyltetradST}   with $a \equiv (A_2 + \rmi A_3) /\sqrt{2}$. 
 Moreover,  from \eqref{eq:dataST}  it follows that supertranslations  act on the 
 data appearing in \eqref{eq:Psi012final} and  \eqref{eq:Psi3final}  as 
\be
\cR'=\cR,\qquad  \kappa_0' = \kappa_0,\qquad  Ds' = Ds, \qquad \hat \Omega^e{}'=\hat \Omega^e +\kappa_0 \sqrt{2}  \, \overline a. 
\label{eq:rotationI}
\ee 
These transformations  lead precisely to the  behaviour of the Weyl scalars described by  \eqref{eq:WeylST} when we apply them to \eqref{eq:Psi012final} and  \eqref{eq:Psi3final}. Therefore, the effect of a supertranslation in the Newman-Penrose formalism is entirely equivalent to a rotation of the null tetrad, which has no effect on the horizon geometry.  As a consequence, the quantity  $\hat \Omega^e$  could be changed to any value in \eqref{eq:Psi3final} without affecting the information about the spacetime curvature carried by $\Psi_3^c$, e.g. it could  be eliminated from  \eqref{eq:Psi3final} choosing $\overline a = - \hat \Omega^e/(\kappa_0\sqrt{2})$ in \eqref{eq:WeyltetradST}. This concludes our proof  that  the two expressions  \eqref{eq:Psi3final} and \eqref{eq:Psi3TrueFinal} for the third Weyl scalar $\Psi^c_3$ can be identified as projections of \emph{the same Weyl tensor} expressed in terms of two different null tetrads. As a consequence  \eqref{eq:Psi3final} and \eqref{eq:Psi3TrueFinal}   represent the same 
spacetime geometry, which can be entirely encoded in the supertranslation invariant free data set $\mathscr{D}_{free}$. \\

 Summarising,  we have argued that the  horizon free data set  \eqref{eq:completeFreeData}  is both \emph{necessary and sufficient} to reconstruct all the information about the spacetime geometry determined by the NEH structure:
\begin{itemize}
\item the free data set $\mathscr{D}_{free}$   \eqref{eq:completeFreeData} involves no gauge degrees of freedom,
\item  all the information about the  spacetime curvature  which is contained in \eqref{eq:dataDef}   is also encoded in  $\mathscr{D}_{free}$, and
\item  the corresponding data about the curvature tensor   can be recovered using the expressions of the Weyl scalars \eqref{eq:Psi012final} and  \eqref{eq:Psi3TrueFinal}, and the  relation  $R_{\mu\nu\rho\sigma} = C_{\mu\nu\rho \sigma}$ which holds in vacuum.  
\end{itemize}
Since the all the elements in the data set   $\mathscr{D}_{free}$ are invariant under supertranslations, this result completes the proof  that horizon supertranslations act trivially on the NEH geometry.

\section{Radiative vacua of null infinity}
\label{sec:Scri}
%%%%%%%%%%%%%%%%%%%%%%%%%%%%%%%%%%%%%%%%%%%%%%%%%%%%%%%%
%%%%%%%%%%%%%%%%%%%%%%%%%%%%%%%%%%%%%%%%%%%%%%%%%%%%%%%%

For completeness, in this section we review the role of BMS supertranslations at null infinity  using Penrose's conformal framework, and the intrinsic description of $\cI$ developed in \cite{Geroch1977,Ashtekar:1981hw} (see also \cite{Ashtekar:1981bq,Ashtekar:1987tt,Ashtekar:2014zsa}). In particular, we will reproduce the well known result that the radiative vacuum of asymptotically flat spacetimes is degenerate,  and we will discuss the connection of this degeneracy  with supertranslations.  At null infinity the would-be gauge degree of freedom  associated to  BMS supertranslations is \emph{necessary} to have a complete  characterisation of the dynamics of $\cI$, i.e. it cannot be gauged away.  Thus,  BMS supertranslations act non-trivially on the geometric data of $\cI$.

 The difference between the dynamical behaviour of a NEH and null infinity can be traced back to three main causes.  First,  the structure of null infinity is only defined up to conformal transformations, what requires that  the dynamical degrees of freedom of $\cI$ are encoded in appropriate \emph{equivalence classes  of data sets}. Second, the Ricci tensor of the conformal completion of the physical spacetime does not satisfy the ordinary Einstein's equations,  implying that the constraint equations we used for  NEH's, \eqref{eq:DNS2} and \eqref{eq:transverseCon2},  are no longer valid for null infinity. And finally, contrary to the case of horizons, the boundary conditions for the gravitational field  at null infinity allow for gravitational radiation propagating in a transverse direction to reach $\cI$.   

The geometry of null infinity will be described using the same formalism as in the case of non-expanding horizons, and thus the following analysis shall serve as  a non-trivial consistency check of our approach.

\paragraph{Asymptotically flat spacetimes.} Let us begin recalling the definition of asymptotic flatness and null infinity following \cite{Ashtekar:1981hw}. A spacetime $(\hat \cM,\hat g)$ is said to be asymptotically flat at null infinity if it is possible to find a  spacetime $(\cM, g)$, together with an embedding $\Psi: \hat \cM \rightarrow \cM$, and a function $\Omega$ on $\cM$ such that 
\begin{itemize}
\item[(i)] $\Psi^*g_{ab} = \Omega^2\,  \hat g_{ab}$ on $\hat \cM$.
\item[(ii)] 
 $I\cong \mathbb{S}^2 \times \mathbb{R} $ is the  boundary of $\Psi(\hat \cM)$ on $\cM$,  located at $\Omega=0$.
 \item[(iii)]   The normal form is given by $n_\mu \equiv \nabla_\mu \Omega\neq 0$ on $I$.
\item[(iv)] There is a neighbourhood  of $I$ on $\cM$, such that $\hat g_{ab}$ satisfies the vacuum Einstein equations, i.e.  $\hat R_{ab}=0$.
\end{itemize}
The spacetime $(\cM, g)$ is called the unphysical spacetime, and the hypersurface $I\subseteq\cM$, which is null as a consequence of $(i)$, $(ii)$ and $(iv)$, is referred as  null infinity. Note that, if the pair $(\Omega, g)$ defines an appropriate conformal completion, so does the pair $(\omega\, \Omega, \omega^2\, g)$ for some smooth positive function\footnote{the factors of $\omega$ are chosen so that the physical metric $\Omega^{-2} g$ remains the same.} $\omega$ on $\cM$.  Two asymptotic completions  related in this way are regarded as equivalent, and the freedom to perform such conformal transformations should be considered as a gauge redundancy.

\paragraph{Hypersurface data of null infinity.}To describe the geometry of the null hypersurface $I$  we can use the formalism introduced in section \ref{sec:nullGeometry}. Thus, we introduce an abstract manifold $\cI$, which acts as a diffeomorphic copy of $I\subseteq\cM$ detached from  the unphysical spacetime, 
 and the identification is performed via the embedding  $\Phi: \cI \rightarrow\cM$, such that $\Phi(\cI) =I$.   We will also choose the coordinate system $\xi^a$ for $\cI$ and the rigging $\ell$  following the conventions in  section \ref{sec:nullGeometry}, so that  the hypersurface data can be represented by the set of quantities  \eqref{eq:dataDef}. In the case of null infinity  it is possible to simplify  the hypersurface data taking advantage of the freedom to perform conformal transformations    $(\Omega, g) \to (\omega \Omega, \omega^2 g)$. Actually, under this change the normal form is rescaled  as $\bn \to \omega\, \bn$, what can be   used  to require that the normal vector $n^\mu$  satisfies \cite{Geroch1977} (see also \cite{Wald:1984rg})
\be
\nabla_\nu n^\mu =0 \ \ \text{on $I$}\qquad  \Longrightarrow\qquad  \kappa=\Omega_M=\Theta_{MN}=0,
\label{eq:divFree}
\ee
where the conditions on the connection coefficients follow  from \eqref{eq:connectionCoeff}.
In this gauge, the second fundamental form vanishes, and therefore null infinity $\cI$ admits a description as a non-expanding null hypersurface with $\pd_n q_{MN}=0$.  In addition, using the same conformal freedom, we can  impose that $q_{AB}$  describes a two dimensional metric of constant scalar curvature $\cR$ \cite{Geroch1977}.  In this setting the hypersurface data of $\cI$ has the form
\be
\gamma_{ab} = \begin{pmatrix}
0&0\\
0& q_{MN}
\end{pmatrix}, \qquad \ell^a=(1,0,0), \qquad \ell^{(2)}=0, \qquad Y_{ab} = \begin{pmatrix}
0&0 \\
0& \Xi_{MN}
\end{pmatrix}.
\label{eq:gauge2NI}
\ee
In the following, to simplify the notation, we will make no distinction between null-infinity $I$ and its abstract copy $\cI$.
\paragraph{Residual gauge redundancies.}    
The conventions introduced above  do not eliminate all the gauge redundancies of our description.  Regarding the conformal transformations, the present setting fixes completely    the normalisation of the null vector $n$. However, we are still allowed  to perform conformal transformations with a conformal factor that satisfies $\omega|_I=1$ at null infinity, but takes  arbitrary values away from it.
From the definition of the hypersurface data it is straightforward to check that under this residual conformal transformations the metric data $\{\gamma_{ab}, \ell_a, \ell^{(2)}\}$ remain invariant, but the transverse components of the tensor $Y_{ab}$ transform as 
\be
\Xi_{MN}' = \Xi_{MN} +\lambda\,  q_{MN},
\label{eq:scriConf}
\ee
 where  $\lambda(\xi)\equiv \cL_\ell \omega|_\cI$ can be any smooth function  on $\cI$.

In addition to these conformal transformations, our description of null infinity also involves redundancies associated to the  freedom to perform diffeomorphisms on the abstract manifold, and the choice of rigging.  Actually,  the analysis of the gauge redundancies of the hypersurface data that we presented in section \ref{sec:redundancies}  is also applicable here,  since the condition \eqref{eq:divFree} implies that null infinity can be described  a non-expanding null hypersurface, and the conventions \eqref{eq:gauge2NI}  are the same we used to study  horizons. Taking the results of section \ref{sec:redundancies} into account,  and recalling that the normalisation of the null normal $n = \pd_{\xi^1}$ is fixed by our choice of conformal gauge, 
 we find that  the only residual freedom of this type are BMS supertranslations\footnote{Similarly to the case of horizons, BMS supertranslations can be  described as hypersurface symmetries of null infinity $I\subseteq \cM$  (see \ref{app:BMSgroup}).   }  \eqref{eq:dataST}
 \be
 \Xi_{MN}'= \Xi_{MN} - D_M A_N,
\label{eq:scriST}
\ee
which corresponds to a diffeomorphism of the abstract manifold $\cI$ acting as $\zeta^a(\xi) = (\xi^1 + A(\xi^M),\xi^M)$.  Note that, since $\kappa=0$, these transformations leave invariant all other elements of the hypersurface data. We will denote the group of BMS supertranslations by $\cS$.

It might seem surprising that we did not encounter the BMS group when discussing the gauge freedom at null infinity. The reason is that   we fixed the scale of the null normal  (by  eqs. \eqref{eq:divFree}, and choosing $q_{MN}$ to describe a sphere with curvature $\cR$) before characterising the residual gauge freedom of our description.    Indeed,  the  BMS group can be recovered  as the set of gauge transformations that leave invariant the conventions \eqref{eq:metricData} up to a conformal transformation  $(\Omega,g) \to (\omega \Omega,\omega^2 g)$, with $\omega\neq1$ at $\cI$. When the conformal transformations are gauge fixed so that only those with $\omega|_{\cI}=1$ are allowed, the remaining residual gauge is given by the transformations \eqref{eq:scriConf} and \eqref{eq:scriST} that we described above. A derivation of the full BMS group can be found in appendix \ref{app:BMSgroup}, where we do a similar  analysis  to that of section \ref{sec:supertranslationsGauge} for non-expanding horizons. In the appendix we show that the BMS group can be described as the set of diffeomorphisms of the unphysical spacetime which leaves invariant the metric tensor $g_{\mu\nu}$ and the null normal $\bn$ of $\cI$ up to a conformal transformation, $g \to \omega^2\, g$ and $\bn \to \omega\,  \bn$. 

\paragraph{Constraint equations.} As in the case of non-expanding horizons, the hypersurface data of null infinity cannot be specified freely.  It must be consistent with the constraint equations (\ref{eq:Raychad}- \ref{eq:Xi}), which are mathematical identities satisfied by any null hypersurface.  In the previous paragraphs we have presented most of the elements involved in these equations, and it only remains to compute the terms \eqref{eq:sourceTerms}. 
The crucial difference with our previous discussion of non-expanding horizons is that, although the physical spacetime is vacuum in a neighbourhood of $\cI$ (condition ($iv$)), the \emph{unphysical Ricci tensor} $R_{\mu\nu}$ \emph{does not vanish}. This is just a direct consequence of the non-trivial transformation properties of the Ricci tensor under the conformal rescaling of the metric (see \cite{Wald:1984rg}).  Therefore,  \eqref{eq:DNS2} and \eqref{eq:transverseCon2} are not valid for $\cI$.

In the following subsection we will characterise the  terms  \eqref{eq:sourceTerms} of the constraint equations of null infinity, i.e. the unphysical Ricci tensor $R_{\mu\nu}$,  and then we will turn to  the resolution  of the constraints  in section \ref{sec:solEqScri}.

\subsection{The Ricci tensor at null infinity}
 %%%%%%%%%%%%%%%%%%%%%%%%%%%%%%%%%%%%%%%%

In order to compute the Ricci tensor at points of null infinity it is convenient to note that the Weyl tensor $C_{\rho\sigma\mu\nu}$ is vanishing at $\cI$. This implies  that the unphysical Riemann tensor at null infinity has the general form \cite{Geroch1977,Ashtekar:1981hw}
\be
R_{\sigma\rho\mu \nu } =\ft12(g_{\sigma[\mu} S_{\nu] \rho} - g_{\rho[\mu} S_{\nu]\sigma}),
\label{eq:defRscri}
\ee
where  the symmetric tensor $S_{\mu \nu}$ is the Schouten tensor defined at the beginning of section \ref{sec:Hypersurfaces}.  The tensor $S_{\mu\nu}$ has a particularly simple form when expressed in the basis  $\cB=\{n, \ell, e_M\}$ due to our gauge fixing conventions \eqref{eq:divFree} and \eqref{eq:gauge2NI}.  Indeed, the divergence free condition \eqref{eq:divFree} can be used in combination with the Ricci identity to prove that the four components $S_{na}\equiv n^\mu e^\nu_a S_{\mu \nu}$  vanish  at null infinity
\bea
0&=& \ell_\sigma n^\mu e^\nu_A \, \nabla_{[\mu} \nabla_{\nu]} n^\sigma =\ell_\sigma n^\mu e^\nu_A R^\sigma_{\rho \mu \nu} n^\rho = \ft 12 S_{An},\nonumber\\
0&=& q^{AB}\, e_{A |\sigma} n^\mu e^\nu_B \, \nabla_{[\mu} \nabla_{\nu]} n^\sigma =q^{AB}\, e_{A |\sigma}  n^\mu e^\nu_B R^\sigma_{\rho \mu \nu} n^\rho =- S_{nn}.
\label{eq:GaugeSchouten}
\eea
Moreover,  it is also possible to show that the components of the Schouten tensor satisfy
$S_M^M  = \cR$,
where $\cR$ is the scalar curvature of $q_{MN}$. This expression can be derived comparing the result of computing $R_{MANB} q^{AB} q^{MN}$ directly from  \eqref{eq:defRscri}, with the outcome of the same computation  using the identity \eqref{eq:appRicciIdentity1} (see appendix \ref{app:weyl}) together with  the gauge conditions \eqref{eq:divFree}. We can simplify $S_{\mu\nu}$ even further making use of the  residual conformal transformations with $\omega|_\cI=1$,  which act on the Schouten tensor as (see \cite{Geroch1977}) 
\be
S'_{ab} = S_{ab}, \qquad  S_{\ell a}' = S_{\ell a} - 2 \pd_a \lambda, \qquad S_{\ell \ell}'  = S_{\ell \ell}   -2 \mu + 4 \lambda^2,
\label{eq:schGauge}
\ee
where $\lambda(\xi) = \nabla_\ell \, \omega|_{\cI}$ and  $\mu(\xi) =\ell^\mu \ell^\nu \nabla_\mu  \nabla_\nu\,  \omega|_\cI$ are two arbitrary functions on $\cI$. Therefore, we can set $S_{\ell \ell}=0$ 
 by a suitable choice of the function $\mu$.
 
 Collecting these results, and using  the inverse metric \eqref{eq:inverseG} to compute the trace of the Schouten tensor, it is possible to derive  the Ricci tensor of the unphysical spacetime using $R_{\mu\nu} = S_{\mu \nu} + \ft13 S g_{\mu\nu}$. We find the  following    non-vanishing components
\begin{empheq}[box=\widefboxb]{align}
R_{n\ell} = S_{n \ell}, \quad R_{MN} = S_{MN} +\ft12 (2 S_{n \ell} + \cR) q_{MN},
\label{eq:EinsteinScri}
\end{empheq}
and $R_{nn} = R_{nM}= R_{\ell \ell}= 0$. 
This form for the unphysical  Ricci tensor at null infinity is  universal for any asymptotically flat spacetime. We will now derive the additional conditions satisfied by $R_{\mu \nu}$ in regions of  $\cI$ where  no outgoing is radiation present.

\paragraph{Geometry of the radiative vacuum.} In order to find the relevant boundary conditions 
we need to consider  
the  leading order contribution $K_{\rho \sigma\mu \nu}\equiv\Omega^{-1}C_{\rho \sigma \mu \nu}$ to the Weyl tensor, since the  unphysical Weyl tensor $C_{\rho \sigma \mu \nu}$ always  vanishes on $\cI$  \cite{Geroch1977}.  

The condition that there is no outgoing radiation in a region of $\cI$ is most easily expressed in terms of the leading order Weyl scalars, which are defined as  components of the tensor $K_{\sigma \rho \mu\nu}$ in the basis $\cB_{NP} = \{\ell,  n, m,\overline m\}$ (see \cite{Newman:1981fn}). Note that we have changed the order of the  first two elements of the null tetrad $\cB_{NP}$, $\ell$ and $n$, with respect to section \ref{sec:NewmanPenrose},  while  $m$ and $\overline m$ are  defined by \eqref{eq:mbarm}. The relevant Weyl scalars are given by
\be
\Psi_2^0 = K^\mu_{\nu \rho \sigma} \ell_ \mu m^\nu n^\rho  \overline  m^\sigma, \qquad \Psi_ 3^0 =  K^\mu_{\nu \rho \sigma} \ell_ \mu n^\nu n^\rho \overline m^\sigma,  \qquad \Psi^0_4 = K^\mu_{\nu \rho \sigma} \overline m_\mu n^\nu \overline m^\rho n^\sigma,
\label{eq:ScriPsi}
\ee
and the conditions for no outgoing radiation at a region of  $\cI$ read \cite{Ashtekar:1987tt}
\be
\text{\emph{Radiative vacuum:}}\qquad  \Im \Psi_2^0 = 0,  \qquad \Psi_3^0 =0,\quad \text{and} \quad  \Psi_4^0=0.
\label{eq:noScriRad}
\ee
The implications of these boundary conditions on the form of the  Schouten tensor can be derived from the equations 
\be
\hspace{-3.56cm} \text{\emph{Bianchi Identities:}}\qquad \ \ \quad Ê\nabla_{[\mu} S_{\nu]\sigma} = -K_{\mu \nu \sigma \rho} n^\rho,
\label{eq:ScriFieldEq}
\ee
which are a direct consequence of  Bianchi identities of the unphysical spacetime  \cite{Geroch1977,Ashtekar:1987tt}.  In order to solve the previous equations and boundary conditions, it is   convenient to express them in the basis $\cB=\{\ell, n, e_M\}$. Taking contractions on both sides of  \eqref{eq:ScriFieldEq} with appropriate combinations of the elements in $\cB$, and  using \eqref{eq:connectionCoeff} in combination with the gauge conditions \eqref{eq:divFree} to simplify the result, we find
\bea
\Im \Psi_2^0=0 \qquad & \Longrightarrow& \qquad D_{[M} S_{N]\ell} = \Xi_{[M}^P S_{N]P},\label{eq:sch1}\\
\Psi_3^0=0 \qquad & \Longrightarrow& \qquad  \pd_{[n} S_{M]\ell}=0,\quad  \text{and}\quad D_{[M} S_{N]P}=0,  \label{eq:sch2}\\
\Psi_4^0=0 \qquad & \Longrightarrow& \qquad \pd_n S_{MN}=0. \label{eq:sch3}
\eea
A detailed derivation can be found in appendix \ref{app:scriVacua}. The last equation \eqref{eq:sch3}  implies that the components $S_{MN}$  have to be  constant along the null direction of $\cI$. Moreover, the form of $S_{MN}$ can be found solving the  second constraint in \eqref{eq:sch2} in combination with $S_M^M=\cR$, and it has the unique solution
\be
S_{MN} = \ft12 \cR \, q_{MN}.
\label{eq:vacSchouten}
\ee
The original proof can be found in \cite{Geroch1977}, but given that the setting therein is slightly different from ours, for completeness  we have  written a summary of  it  in appendix \ref{app:scriVacua}. From the previous relation  it follows   that  the right hand side of    \eqref{eq:sch1} must vanish,  which  together with the  first equation in \eqref{eq:sch2}, also implies that the components $S_{a\ell}$   take the form  $S_{a\ell} = \pd_a S_\ell$ for some function $S_\ell$ on $\cI$.  Thus, from equation  \eqref{eq:schGauge} it is straightforward to check that 
in the radiative vacuum the components $S_{a\ell}$  are pure conformal gauge. 

With these results at hand,  we can finally  obtain the  components  of the unphysical Ricci tensor on $\cI$  \emph{in the absence of outgoing  radiation}  
  \begin{empheq}[box=\widefboxb]{align}
R_{n\ell} = \pd_n S_{\ell}, \qquad R_{MN} =  (\pd_n S_{\ell} +  \cR) q_{MN},
\label{eq:EinsteinScriVac}
\end{empheq}
and $R_{nn}=R_{nM}=R_{\ell \ell}=0$. This result will allow us to write down the  constraint equations at regions of $\cI$ where there is no outgoing radiation,  and whose solutions  represent the radiative vacua of asymptotically flat spacetimes.

\subsection{Constraint equations for the transverse connection}
\label{sec:solEqScri}
%%%%%%%%%%%%%%%%%%%%%%%%%%%%%%%%%%%%%%%%

Before discussing the constraint equations let us comment on the physical degrees of freedom contained in the transverse connection $\Xi_{MN}$. At the beginning of this section we identified  the trace of $\Xi_{MN}$ as a pure conformal gauge (see eq. \eqref{eq:scriConf}), and thus,
 in order to eliminate this redundancy  we will proceed as in  \cite{Ashtekar:1981hw}, identifying those connections related by a conformal transformation.  In other words, we will introduce  the equivalence relation 
\be
\Xi_{MN} \approx \Xi_{MN}' \quad \Longleftrightarrow\quad \Xi_{MN}' - \Xi_{MN} = \lambda\,  q_{MN},
\label{eq:equiv}
\ee 
where $\lambda(\xi)$ is an arbitrary smooth function  on $\cI$,  and we will  work with the resulting  equivalence classes.  
This amounts to neglecting  the trace part of the transverse connection $\Xi_{MN}$, leaving as the dynamical field its traceless part  $\Xi_{MN} - \ft12 \Xi_L^{\phantom{L} L} q_{MN}$.
Note that this quantity describes precisely two degrees of freedom, which  can be identified with the two radiative degrees of freedom of  gravitational radiation \cite{Ashtekar:1981hw,Ashtekar:1987tt}.

\paragraph{General form of the constraint equations.} We begin discussing the constraint equations for a general  situation in the presence of radiation. In particular we will show that the two components in the traceless part of $\Xi_{MN}$ are both necessary and sufficient to describe the radiative degrees of freedom of the gravitational field  at null infinity.

In the presence of radiation at null infinity, the  terms \eqref{eq:sourceTerms} in  the constraint equations can be computed from the Ricci tensor given in \eqref{eq:EinsteinScri}
\be
J_{nn} = J_{nM} = 0,\qquad J_{MN} = - S_{MN}-S_{n\ell} q_{MN} - \ft12 \cR q_{MN}.
\label{eq:sourceTermsScri}
\ee
This result together with  the gauge conditions \eqref{eq:divFree}  imply that the  Raychaudhuri \eqref{eq:Raychad} and Damour-Navier-Stokes equations \eqref{eq:NS}  are trivially satisfied on $\cI$. The only non-trivial equations are those for  the transverse components of the connection \eqref{eq:Xi}, which can be expressed as
 \begin{empheq}[box=\widefboxb]{align}
\pd_n \Xi_{MN} - \ft12 (S_{MN} + S_{n\ell} q_{MN}) \approx- \ft12 N_{MN},
\label{eq:XiScri}
 \end{empheq}
where we have used the equivalence relation \eqref{eq:equiv}.
Here  $N_{MN} \equiv S_{MN} - \ft12 \cR q_{MN}$, is the \emph{news} tensor, which vanishes in the absence of radiation passing through  $\cI$, i.e. when eqs.  \eqref{eq:vacSchouten} hold. In addition to the previous equation, the connection must satisfy one more  constraint coming from the  identity \eqref{eq:appRicciIdentity1} 
\be
D_{[M} \Xi_{N]P} = \ft12 q_{P[M} S_{N] \ell}.
\label{eq:XiScri2}
\ee
Using the  Bianchi identities satisfied  by the Schouten tensor  \eqref{eq:ScriFieldEq} it can be checked that this condition is consistent with the time evolution given by \eqref{eq:XiScri} (see appendix \ref{app:scriVacua}). In other words, if the previous equation is satisfied at any given value of the null coordinate, then equation \eqref{eq:XiScri} ensures that it will hold for all values of $\xi^1$. \\

We can now show that the two components in the traceless part of $\Xi_{MN}$ encode  the radiative modes of gravitational radiation at null infinity.  Recall that the information about the outgoing radiative modes at $\cI$ is described by the   Weyl scalars $\Im \Psi_2^0$, $\Psi_3^0$ and $\Psi_4^0$ \cite{Ashtekar:1981hw,Ashtekar:1987tt}. As we review in appendix \ref{app:scriVacua}, using the Bianchi identity \eqref{eq:ScriFieldEq}, it is possible to express these scalars as
\bea
\Im \Psi_2^0 &=& -\ft12 ( D_{[3} S_{2] \ell} - \Xi_{[3}^M S_{2]M}),\\
\Psi_3^0&=&  \ft1{\sqrt{2}} (D_{[3}S_{2]3}  - \rmi D_{[2}S_{3]2} ),\\
\Psi_4^0 &=& \ft12 (\pd_n S_{22} - \pd_n S_{33}) - \rmi \pd_n S_{23},
\eea
implying that  they are completely determined by the components of the Schouten tensor  $S_{M \ell}$ and  $S_{MN}$,  and  by the transverse connection $\Xi_{MN}$. Actually, it is easy to prove that  only the traceless part of $\Xi_{MN}$ contributes in the first equation, and that these expressions are invariant under the conformal transformations \eqref{eq:scriConf} and \eqref{eq:schGauge}. 
Then,    the constraint equations \eqref{eq:XiScri} and \eqref{eq:XiScri2} can be solved for $S_{MN}$ and $S_{M\ell}$ giving
\be
S_M^M= \ft12 \cR, \qquad S_{MN} \approx-2 \pd_n \Xi_{MN}, \qquad S_{N \ell} = q^{PM}D_{[M}\Xi_{N]P}.
\ee
In these equations too only the traceless part of $\Xi_{MN}$ contains relevant information about the geometry at $\cI$, as the contribution of the trace $\Xi_M^M$  can be identified as pure conformal gauge.  Thus, we can conclude that the traceless part of    $\Xi_{MN}$ is both \emph{necessary and sufficient}  to recover completely the  information about the radiative modes at  $\cI$ (for a more detailed derivation see \cite{Ashtekar:1981hw}).

\paragraph{Degeneracy of the radiative vacuum.} Finally, we turn to the discussion of the degeneracy of the radiative vacua
 in asymptotically flat spacetimes. The constraint equations for regions of $\cI$ with no outgoing radiation can obtained from \eqref{eq:XiScri} and  \eqref{eq:XiScri2} together with the boundary conditions \eqref{eq:EinsteinScriVac}, which imply $S_{a \ell} = \pd_a S_\ell$, and  $N_{MN}=0$. Using the equivalence relation  \eqref{eq:equiv} they read 
 \be
\pd_n \Xi_{MN} \approx0, \qquad 
D_{[M} \Xi_{N]P} \approx 0.
\label{eq:vacuumScri3}
 \ee
Note that $\Xi_{MN}$ still transforms under supertranslations. In order to characterise the set of radiative vacua  avoiding possible gauge artifacts  we need to introduce a new gauge invariant dynamical variable.  Due to our choice of conformal gauge the Hajicek one-form and the surface gravity are vanishing, and therefore we cannot  proceed as in the case of NEHs and  construct a supertranslation invariant variable analogous to $s_{MN}(\eta)$ in eq. \eqref{eq:defSigmaGI}. Instead, following \cite{Ashtekar:1981hw},  we choose a reference vacuum  $\mathring{\Xi}_{MN}$ and then we consider the differences $\Sigma_{MN}=\Xi_{MN} - \mathring{\Xi}_{MN}$, between a generic vacuum connection $\Xi_{MN}$ and the fiducial connection $\mathring{\Xi}_{MN}$. 
It is easy to check that $\Sigma_{MN}$ is invariant under supertranslations. Then, given a fixed fiducial connection $\mathring{\Xi}_{MN}$, the set of distinct $\Sigma_{MN}$ consistent with the equations  \eqref{eq:vacuumScri3} is isomorphic to the set of radiative vacua.  Since both of the connections $\Xi_{MN}$ and $\mathring{\Xi}_{MN}$  describe a radiative vacuum,  their  difference $\Sigma_{MN}$  also   satisfies \eqref{eq:vacuumScri3} due to the linearity of the equations. The general solution to \eqref{eq:vacuumScri3}, and therefore, the set of radiative vacua of null infinity is characterised by the expression
 \begin{empheq}[box=\widefboxb]{align}
\Sigma_{MN} \approx D_M f_N - \ft12 \Delta f  q_{MN},
\label{eq:vacuaScriSol}
 \end{empheq}
where $f(\xi^M)$ is any smooth function of the coordinates $\xi^M$  (see appendix \ref{app:scriVacua}).  It is important to stress that the smoothness $f(\xi^M)$ is essential  for the derivation,  which uses the fact that the spatial sections of  $\cI$ are compact and simply connected. We will denote the set of vacuum connections by $\mathring{\Gamma}$. 

The previous expression \eqref{eq:vacuaScriSol} already indicates clearly that the set of vacuum connections is infinitely degenerate.   Comparing \eqref{eq:vacuaScriSol} with \eqref{eq:scriST} it is straightforward to check that, given a fiducial vacuum $\mathring{\Xi}_{MN}$,  the most general vacuum connection is given by
 \be
\Xi_{MN} \approx \mathring{\Xi}_{MN} + D_M f_N - \ft12 \Delta f  q_{MN},
\label{eq:setVacua}
 \ee
and therefore, the difference between any two vacuum connections has the form of a supertranslation. In other words,  we can construct  the full set $\mathring{\Gamma}$ acting on $\mathring{\Xi}_{MN}$ with all the elements  of the group of BMS supertranslations $\cS$ \eqref{eq:scriST}. 
Note that  \emph{BMS translations}, i.e.  the four  dimensional subgroup $\cT \subseteq \cS$ of supertranslations satisfying
\be
D_M f_N - \ft12 \Delta f  q_{MN} =0,
\ee
 acts trivially on the  connections of null infinity.  Thus,  the set of radiative vacua is isomorphic to the group of supertranslations modulo BMS translations $\mathring{\Gamma} \cong \cS/\cT$. 
 
 It is interesting to see how the presence of a non-vanishing news tensor induces a change of the radiative vacuum.   Consider a  solution  of \eqref{eq:XiScri} where the news $N_{MN}$ is non-zero in the interval $\xi^1 \in( \xi^1_i,\xi^1_f)$ and vanishes everywhere else.
 Then, the initial and final states of the connection, $\Xi_{MN}|_{\xi_i^1}$ and  $\Xi_{MN}|_{\xi_f^1}$ respectively, represent radiative vacua of $\cI$, and have to be of the form \eqref{eq:setVacua}. 
Integrating \eqref{eq:XiScri} we find that the difference between the final and initial transverse connections $\Xi_{MN}$ is given by the expression 
\be
\delta \Sigma_{MN} \equiv \Sigma_{MN}|_{\xi_f^1} - \Sigma_{MN}|_{\xi_i^1}   =\Xi_{MN}|_{\xi_f^1} - \Xi_{MN}|_{\xi_i^1}  \approx-   \ft12 \int_{\xi_i^1}^{\xi_f^1} d\xi^1 N_{MN},   
\label{eq:newsShift}
\ee
which is invariant under supertranslations.  For a generic source of radiation the configuration of the news tensor $N_{MN}$ will be such that $\delta \Sigma_{MN}\neq0$ and therefore, in general, the initial and final transverse connections correspond to \emph{distinct radiative  vacua}. In particular, it is now clear that if we imposed a gauge fixing condition on $\Xi_{MN}$ to eliminate the freedom to perform supertranslations \eqref{eq:scriST} we would be restricting  the allowed dynamics at null infinity. \\

 From the discussion in the previous paragraphs we can see that, in contrast  to the case of horizons,  null infinity supertranslations transform the dynamical variables of $\cI$. On the one hand,  supertranslations have been shown to act non-trivially on the traceless part of the transverse connection $\Xi_{MN}$ (see eq. \eqref{eq:scriST}). On the other hand, the two components in the traceless part of $\Xi_{MN}$ are both \emph{necessary and sufficient} to describe the two degrees of freedom of  gravitational radiation at $\cI$.  Thus, connections related by a supertranslation  cannot be identified with each other, as this would require gauging away one further component of the traceless part of $\Xi_{MN}$. 
  As a consequence, BMS supertranslations must be regarded as  large gauge transformations, i.e. as \emph{global symmetries} of the  constraint equations of null infinity,  which  act non-trivially on the geometric data of $\cI$.

\section{Results and Discussion}
\label{sec:discussion}
%%%%%%%%%%%%%%%%%%%%%%%%%%%%%%%%%%%%%%%%%%%%%%%%%%%%%%%%
%%%%%%%%%%%%%%%%%%%%%%%%%%%%%%%%%%%%%%%%%%%%%%%%%%%%%%%%

One of the most interesting features about asymptotically flat spacetimes is the  infinite dimensional
 asymptotic symmetry group at null infinity, the BMS group.  The BMS symmetries, and in particular null infinity supertranslations, were originally characterised  as  diffeomorphisms which preserved certain coordinates conventions in a neighbourhood of null infinity \cite{Bondi:1962px,Sachs:1962wk,Sachs:1962zza}. Many years later, the study of the geometrical structure  of null infinity led to the isolation of the radiative degrees of freedom of the gravitational field \cite{Ashtekar:1981hw}, and it was understood that BMS supertranslations act  non-trivially on the radiative degrees of freedom. Actually, the radiative vacuum of asymptotically flat spacetime was shown to be infinitely  degenerate, and that it was possible to transform each of these vacua into any other with a supertranslation.

Recently it has been argued that the ASG of spacetimes containing a non-extremal black hole should be enhanced with horizon supertranslations.   These diffeomorphisms  would transform  the state of the black hole horizon in an analogous way as BMS supertranslations act on the geometric data of null infinity. According to this proposal, the  multiplicity of black hole states generated by horizon  supertranslations could  provide a partial explanation for the Bekenstein-Hawking  entropy formula.

The task of characterising  the ASG of the near horizon geometry for non-extremal black holes has been addressed in many works. However, there is  no consensus  regarding the structure of the ASG, or the physical interpretation of these diffeomorphisms. In the present paper we have presented a detailed characterisation of the geometric properties of  supertranslations defined on a generic  non-expanding horizon embedded in vacuum. For this purpose we have used a coordinate independent approach analogous to the one used in \cite{Geroch1977,Ashtekar:1981hw} to study the structure of null infinity in exact, non-linear, general relativity. In this framework, the intrinsic and extrinsic geometry of the horizon are encoded in tensor fields living on an abstract three dimensional manifold $\Sigma$, which acts as a diffeomorphic copy of horizon separated from the physical spacetime. In particular,  the corresponding set of tensor fields,  known as the  \emph{the horizon data set}, contains the dynamical degrees of freedom of the horizon, i.e. the freely specifiable and gauge invariant   data of the horizon.   In order to extract the  dynamical degrees of freedom from the  data set, and determine their behaviour under supertranslations, we have followed the strategy   described below:
\begin{itemize}
\item[1.] First, we have characterised in detail \emph{all the gauge redundancies} in our description of the NEH.
\end{itemize}
In particular we have shown that the action of   supertranslations   on the horizon data is identical to that of a gauge redundancy:  they are  associated with a reparametrisation of the null direction of the horizon, and a change of the transversal direction used to define the extrinsic geometry, (i.e. the rigging).   Thus, supertranslations leave invariant both the intrinsic and extrinsic geometry of the horizon up  to a gauge redundancy of the description.
\begin{itemize}
\item[2.] To determine the free data of the horizon we have solved the constraints imposed by the vacuum Einstein's equations on the NEH geometry.  
\end{itemize}
As a result of this analysis we have identified the set of geometric quantities which can be freely specified on the horizon, and which encode  all the  information about the spacetime curvature  contained on the NEH  geometry.  This \emph{free data set} encodes the  dynamical degrees of freedom of the horizon, but typically still involves some gauge redundancies. 
\begin{itemize}
\item[3.] The previous two analyses can be combined to characterise the gauge redundancies on the free data set. This  allows to extract a set of quantities which are both  \emph{necessary and sufficient} to reconstruct the full NEH geometry.
\end{itemize}
This procedure has led us to find a free data set which  contains no unfixed gauge degrees of freedom, and in particular,  which only involves  objects     \emph{invariant under supertranslations}.  More specifically, this free data set is composed of  quantities defined on a particular  spatial section $\cS_{\eta_0}$ of the horizon
\begin{empheq}[box=\widefboxb]{align}
\text{Free horizon data:} \qquad  \mathscr{D}_{free}\equiv (q_{MN}|_{\eta_0},  \quad \Omega_M^0|_{\eta_0}, \quad s_{MN}|_{\eta_0}).\nonumber
\end{empheq}
where $q_{MN}$ represents the induced metric on the spatial sections of the horizon, and $\Omega_M^0$ determines its angular momentum aspect when $q_{MN}$ has  an $SO(2)$ isometry. In those situations when there is gravitational radiation propagating along the horizon  the symmetric traceless tensor $s_{MN}$  can be  associated to radiative degrees of freedom of the  gravitational field.
 %Recall that the NEH structure is sufficiently general to describe  gravitational radiation propagating along the horizon (but not crossing it). In those situations the symmetric traceless tensor $s_{MN}$  can be  associated to radiative degrees of freedom of the  gravitational field.
  Since the elements of the  free data set $\mathscr{D}_{free}$ are all invariant under supertranslations, we conclude that    \emph{supertranslations act trivially on the NEH geometry}, i.e. they must be regarded as pure gauge. In particular,  the stationary state of the NEH, which corresponds to the case $s_{MN}|_{\eta_0}=0$,  is uniquely determined by  $q_{MN}$ and $\Omega_M^0$, and it does not transform under supertranslations.

A fundamental step to obtain the   supertranslation-invariant data set $\mathscr{D}_{free}$  is the choice of an appropriate  parametrisation for the null direction of the horizon.  Rather than using an arbitrary coordinate,  the null direction is parametrised by the value of a potential $\eta$,  which is defined in a coordinate invariant way (see  eq. \eqref{eq:defEta}). The horizon can be foliated by the  level sets of the potential $\eta$,  the spatial sections $\cS_\eta$, and  the evolution of the geometry along the null direction can be expressed in terms of the dependence on $\eta$ of the horizon data.   In particular, $q_{MN}$ and $\Omega_M^0$ are both constant along the null direction of the horizon, while $s_{MN}$ behaves as
\be
 s_{MN}=  s_{MN}|_{\eta_0}  \; \rme^{- \, (\eta- \eta_0)},\nonumber
\ee
which shows that any deviation away from the stationary state, i.e.  $s_{MN}=0$,  relaxes exponentially fast to it.  The use of the potential to express the evolution of the horizon data  avoids the ambiguity associated to supertranslations,  which are related to  coordinate reparametrisations of the null direction.

It is important to remark that the present work is restricted to the case the non-expanding horizons embedded in vacuum, and thus we have not considered processes  involving matter or radiation falling across the horizon.  To check if our results can be extended to more general situations we have considered  the possibility of ``implanting'' supertranslation hair on an event horizon with a non-spherical shock-wave of null matter or radiation, as proposed in \cite{Hawking:2016sgy}. We have found that,  consistently with the conclusions of this paper, the shock-wave cannot excite the degree of freedom associated to supertranslations. In other words,  the supertranslation degree of freedom cannot encode any ``memory'' about the energy momentum tensor of the shock-wave, what  is in harmony  with our identification of supertranslations as a gauge redundancy.  The corresponding analysis will be presented in a companion paper \cite{shockwave}.

\section*{Acknowledgments}
%%%%%%%%%%%%%%%%%%%%%%%%%%%%%%%%%%%%%%%%%%%%%%%%%%%%%%%%
%%%%%%%%%%%%%%%%%%%%%%%%%%%%%%%%%%%%%%%%%%%%%%%%%%%%%%%%
 
We would like to  thank  J. Barb\'on, P. Benincasa,  J.J. Blanco-Pillado,  G. Dvali, A. Garcia-Parrado, C. G\'omez, A. Helou,  M. Mars,  J. Mart\'in, M. Panchenko and  R. Vera  for useful discussions. We are grateful to J.M.M. Senovilla  for comments on the draft and for discussions.  This work is supported by  the Basque Government grants IT-956-16, POS-2016-1-0075 and  IT-979-16, the Spanish Government Grant FIS2014-57956-P,  the project FPA2015-65480-P,  the Centro de Excelencia Severo Ochoa Programme under grant SEV-2012-0249,  and by the ERC Advanced Grant 339169 ÓSelfcompletionÓ.
\appendix

\section{Constraint equations for null hypersurfaces}
\label{app:constraints}
%%%%%%%%%%%%%%%%%%%%%%%%%%%%%%%%%%%%%%%%%%%%%%%%%%%%%%%%
%%%%%%%%%%%%%%%%%%%%%%%%%%%%%%%%%%%%%%%%%%%%%%%%%%%%%%%%

In this appendix we will present a derivation of the constraint equations for null hypersurfaces (\ref{eq:Raychad}-\ref{eq:Xi}).

\subsubsection*{Raychaudhuri equation}
We begin discussing the constraint equation \eqref{eq:Raychad} for the expansion $\theta$ of the hypersurface. From the definition of the second fundamental form $\Theta_{MN} = e^\mu_M e^\nu_N \nabla_\mu n_\nu$, and using  Leibnitz rule to expand the derivative we find
\be
\pd_n \Theta_{MN} =
e^\nu_N  \nabla_n e^\mu_M  \nabla_\mu n_\nu  + e^\mu_M \nabla_n e^\nu_N   \nabla_\mu n_\nu  + e^\mu_M e^\nu_N  n^\sigma \nabla_\sigma \nabla_\mu n_\nu,
\ee
where $\nabla_n = n^\mu \nabla_\mu$. Substituting the connection coefficients \eqref{eq:connectionCoeff} we have
\bea
\pd_n \Theta_{MN} &=& \Theta_{(M}^P \Theta_{N)P} + e^\mu_M e^\nu_N  n^\sigma \nabla_\sigma \nabla_\mu n_\nu =  \nonumber \\
&=&\Theta_{(M}^P \Theta_{N)P} - R_{nMnN} + e^\mu_M e^\nu_N  n^\sigma \nabla_\mu \nabla_\sigma n_\nu,
\eea
where the second equality follows from using the Ricci identity. Using again the Leibnitz  rule andÊ\eqref{eq:connectionCoeff}  the last term can be written as
\bea
e^\mu_M e^\nu_N  n^\sigma \nabla_\mu \nabla_\sigma n_\nu &=& \pd_M (e_{N}^\mu \nabla_n n_\mu) - \nabla_M e_N^\nu \nabla_n n_\nu - 
e^\nu_N \nabla_M n^\sigma \nabla_\sigma n_\nu = \nonumber \\
&=& \kappa \Theta_{MN} - \Theta_{M}^P \Theta_{NP}.
\eea
The contribution from the Riemann tensor  can be expressed in terms of the Weyl and Ricci tensors
\be
 R_{nMnN} = C_{nMnN} + \ft12 S_{nn}q_{MN} =  C_{nMnN} + \ft12 R_{nn} q_{MN},
\ee
 where we are using the shorthand $R_{nMnN} = R_{\mu \nu \rho \sigma} n^\mu e_M^\nu n^\rho e_N^\sigma$, and similar expressions to denote the contraction of spacetime tensors with the elements of the basis $\cB = \{n,\ell, e_M\}$. 
Using the last equation we obtain
\be
\pd_n \Theta_{MN} =  \kappa \Theta_{MN}  + \Theta_{M}^P \Theta_{NP} -C_{nMnN} - \ft12 R_{nn} q_{MN}, 
\label{eq:fundamentaForm}
\ee
which known as the \emph{tidal force equation} (see \cite{Gourgoulhon:2005ng}).
The equation for the expansion $\theta = \Theta_M^M$ can be calculated from the expression
\be
\pd_n \theta = \pd_n( q^{MN} \Theta_{MN}) = \pd_nq^{MN} \Theta_{MN} + q^{MN} \pd_n \Theta_{MN}.
\ee
Note also that $\pd_n q^{MN} = - 2\Theta^{MN}$, what follows from differentiating $q_{MN} q^{NP} = \delta_M^P$ with respect to $\xi^1$ and using the relation $\Theta_{MN} = \ft12 \pd_n q_{MN}$.
Collecting all these results we   arrive to the final expression for the Raychaudhuri equation
\be
\pd_n \theta -  \kappa \theta +\Theta_{M}^P \Theta_{NP} =- R_{nn}.
\ee
\subsubsection*{Damour-Navier-Stokes equations}
In order  to derive the Damour-Navier-Stokes  equation \eqref{eq:NS} we will compute the component $R_{nA}$ of the Ricci tensor in terms of the hypersurface data. Using the expression \eqref{eq:inverseG} for the inverse metric we find
\be
R_{nA} = R_{\mu n \nu A} g^{\mu\nu}=R_{\ell n n A}  + q^{MN} R_{MnNA}.
\ee
The two terms  can be rewritten using the Ricci identity
\bea
R_{\ell n n A} &=& \ell^\sigma n^\mu e^\nu_A \nabla_{\mu} \nabla_{\nu} n_\sigma - \ell^\sigma n^\mu e^\nu_A \nabla_{\nu} \nabla_{\mu} n_\sigma,\\
R_{MnNA} &=& e_M^\sigma e_N^\mu e_A^\nu  \nabla_{\mu} \nabla_{\nu} n_\sigma - e_M^\sigma e_N^\mu e_A^\nu  \nabla_{\nu} \nabla_{\mu} n_\sigma.
\eea
Each of these four terms can be expressed in terms of the connection coefficients using the definitions \eqref{eq:connectionCoeff} and the Leibniz rule for the covariant derivative. For example, noting that $\Omega_M = \ell^\sigma e^\nu_M \nabla_\nu n_\sigma$ we have
\bea
\ell^\sigma n^\mu e^\nu_M \nabla_{\mu} \nabla_{\nu} n_\sigma  &=& \pd_n \Omega_M - \nabla_n e_M^\nu  \ell^\sigma \nabla_\nu n_\sigma -  e^\nu_M \nabla_n \ell^\sigma   \nabla_\nu n_\sigma = \nonumber \\
&&  \pd_n \Omega_M - \kappa \Omega_M -  \Omega^N \Theta_{MN}  + \kappa \Omega_M  + \Omega^N \Theta_{MN}. 
\eea
Also since $\kappa= \ell^\sigma n^\nu \nabla_\nu n_\sigma$ we have
\bea
\ell^\sigma n^\mu e^\nu_M \nabla_{\nu} \nabla_{\mu} n_\sigma &=& \pd_M \kappa - \nabla_M n^\mu \ell^\sigma \nabla_\mu n_\sigma -  n^\mu \nabla_M \ell^\sigma \nabla_\mu n_\sigma\nonumber\\
& = &\pd_M \kappa - \Omega_MÊ\kappa - \Theta_M^N \Omega_M + \kappa \Omega_M.
\eea
Collecting terms we find
\be
R_{\ell n n M} =R_{ n M \ell n}= \pd_n \Omega_M - \pd_M \kappa +   \Theta_{M}^N \Omega_N.
\label{eq:identity1}
\ee
Similarly it can be shown that 
\be
R_{MnNA} = D_{[N} \Theta_{A]M} + \Theta_{M[N} \Omega_{A]}.
\label{eq:identity2}
\ee
Then we have
\be
R_{nA} =  \pd_n \Omega_A - \pd_A \kappa + D_{N} \Theta_{A}^N - D_A \theta + \theta \Omega_{A},
\ee
which can be identified with the Damour-Navier-Stokes equations   \eqref{eq:NS}.

\subsubsection*{Equation for the transverse connection}
We now describe the  derivation of the constraint  equation \eqref{eq:Xi} for the transverse connection $\Xi_{MN}$. From the definition of the transverse connection  $\Xi_{MN} = \ft12 e^\mu_{(M} e^\nu_{N)} \nabla_\mu \ell_\nu$, and using  Leibnitz rule to expand the derivative we find
\be
2\pd_n \Xi_{MN} = \nabla_n e^\mu_{(M} \, e^\nu_{N)} \nabla_\mu \ell_\nu + e^\mu_{(M} \nabla_n e^\nu_{N)} \nabla_\mu \ell_\nu + e_{(M}^\mu e^\nu_{N)} \nabla_n \nabla_\mu \ell_\nu.
\ee
Substituting the expressions for the connection coefficients \eqref{eq:connectionCoeff} we arrive to
\bea
\pd_n \Xi_{MN} &=& -2 \Omega_M \Omega_N + \Theta_{(M}^P \Xi_{N)P}+ \ft12 e_{(M}^\mu e^\nu_{N)} \nabla_n \nabla_\mu \ell_\nu = \nonumber\\
&=& -2 \Omega_M \Omega_N + \Theta_{(M}^P \Xi_{N)P}+\ft12 R_{(N\ell nM)}+ \ft12 e_{(N}^\nu n^\sigma \nabla_{M)} \nabla_\sigma \ell_\nu,
\label{eq:appXider1}
\eea
where  we have  also used the Ricci identity to derive the second equality. Using the Leibnitz rule for the connection and   \eqref{eq:connectionCoeff} we can  rewrite the last term as
\bea	
 e_{(N}^\nu n^\sigma \nabla_{M)} \nabla_\sigma \ell_\nu &=& \nabla_{(M} ( n^\sigma e_{N)}^\nu \nabla_\sigma \ell_\nu) - \nabla_n e_\nu \nabla_{(M} e_{N)}^\nu - \nabla_{(M} n^\sigma \nabla_\sigma \ell_\nu e^\nu_{N)} =\nonumber\\
&=& -  D_{(M} \Omega_{N)} - 2 \kappa\Xi_{MN}  + 2\Omega_M \Omega_N  - \Theta^P_{(M} \Xi_{N)P}
\label{eq:appXider1half}
\eea
where $D_M$ is the Levi-Civita connection of the  spatial metric $q_{MN}$. Thus,  substituting the previous expression in it we have
\be
\pd_n \Xi_{MN} =   - \ft12 D_{(M} \Omega_{N)} - \Omega_M \Omega_N - \kappa \Xi_{MN}   + \ft12 \Theta_{(M}^P \Xi_{N)P}+\ft12( R_{N\ell nM}+ R_{M\ell nN} ).
\label{eq:appXider2}
\ee
Using the symmetries of the Riemann tensor, and the form \eqref{eq:inverseG} for the inverse metric,  we find that
\be
\ft12( R_{N\ell nM}+ R_{M\ell nN} ) = - \ft12 R_{MN} +\ft12 R_{MANB} q^{AB},
\label{eq:appXider3}
\ee
where $R_{AB} = g^{\mu\nu} R_{\mu A \nu B}$. We will now rewrite $R_{MANB}$ in terms of the connection coefficients. From the Ricci identity
we have
\be
R_{MANB} =   e_{M|\sigma}\,   \nabla_{[\mu} \nabla_{\nu]} e^\sigma_A\,   e^\mu_N e^\nu_B = \pd_{[N} ( e_{M|\sigma}  \nabla_B e_{A]}^\sigma)
-\nabla_{[N} e_{M|\sigma} \nabla_B e_{A]}^\sigma - \nabla_{[N} e_{B]}^\nu \nabla_\nu e_A^\sigma e_{M|\sigma}.
\ee
Substituting the expressions for the connection coefficients  \eqref{eq:connectionCoeff} we obtain 
\be
R_{MANB} =  \pd_{[N} ( \overline{\Gamma}_{A]B}^L q_{M L}) - \Theta_{M[N} \Xi_{A]B}  - \Xi_{M[N} \Theta_{A]B} - \overline{\Gamma}_{M[N}^L \overline{\Gamma}_{A]B}^P q_{PL},
\ee
which after contracting with $q^{AB}$ can be written as 
\be
q^{AB} R_{MANB} = \cR_{MN} + \Xi_{P(M} \Theta^P_{N)} - \theta \Xi_{MN} - \Theta_{MN} \theta^\ell.
\label{eq:identity3}
\ee
Here $\cR_{MN}$ is the Ricci tensor associated to the spatial metric $q_{MN}$.
This result can be used together with  \eqref{eq:appXider2} to express \eqref{eq:appXider1} as follows
\bea
\pd_n \Xi_{MN} &=&   - \ft12 D_{(M} \Omega_{N)} - \Omega_M \Omega_N -( \kappa +\ft12 \theta) \Xi_{MN}  + \Theta_{(M}^P \Xi_{N)P}\nonumber\\ &=& - \ft12\Theta_{MN} \theta^\ell  +\ft12 \cR_{MN}  - \ft12 R_{MN},
\eea
which leads to the constraint equation for the transverse connection \eqref{eq:Xi} after using the identity $\cR_{MN}  = \ft12 \cR q_{MN}$.

\section{Calculations for non-expanding horizons}
%%%%%%%%%%%%%%%%%%%%%%%%%%%%%%%%%%%%%%%%%%%%%%%%%%%%%%%%
%%%%%%%%%%%%%%%%%%%%%%%%%%%%%%%%%%%%%%%%%%%%%%%%%%%%%%%%

\subsection{Fixing the normalisation of the null normal}
\label{app:gaugeFixB}
%%%%%%%%%%%%%%%%%%%%%%%%%%%%%%%%%%%%%%%%

In this appendix we will show that when the NEH data \eqref{eq:invDataSet} satisfy  one of the following  conditions on a spatial slice $\cS_{\xi^1}$
\bea
\hspace{-3cm}(i)&&\qquad \qquad \Omega^0{}^M \Omega_M^0 \leq \ft12 \cR \leq \theta^0, \nonumber\\
\hspace{-3cm}(ii)&&\qquad \qquad \Omega^0{}^M \Omega_M^0\leq \theta^0 \leq \ft12 \cR, 
\label{eq:conditionsBgauge}
\eea
it is possible to find a gauge transformation  \eqref{eq:dataRedundancies} that sets $\pd_n\theta^0=0$  on the horizon $\Sigma$. Moreover,  the gauge freedom that remains after imposing this condition is precisely that of supertranslations.

\textbf{Conditions to set} \bm{$\pd_n\theta^0=0$}. Among the residual gauge freedom \eqref{eq:dataRedundancies}, the transformations $\zeta^a(\xi)= (\hat f(\xi),\xi^M)$ that keep $\kappa$ constant (i.e. gauge condition 1 \eqref{eq:GaugeFix1})  are given by
\be \label{fTransform}
\hat f(\xi)=\xi^1 + A(\xi^M)+\ft{1}{\kappa_0}\log \( 1+B(\xi^M)\rme^{-\kappa_0 \xi^1} \) ,
\ee
\be 
\hat f_n = \frac{1}{1 + B \rme^{-\kappa_0 \xi^1}}, \hspace{.5 cm}
\hat  f_M = A_M + \frac{B_M \rme^{-\kappa_0 \xi^1}}{\kappa_0 (1+B \rme^{-\kappa_0 \xi^1})}, \hspace{.5 cm}
\hat  f_{nM} = -\frac{B_M\rme^{-\kappa_0 \xi^1}}{(1+ B \rme^{-\kappa_0 \xi^1})^2}.\nonumber
\ee
Under these transformations the data changes as follows
\bea \label{dataTransf}
\kappa'(\xi) &=& \kappa|_{\zeta(\xi)},\\
\Omega'_M(\xi) &=& \Omega_M|_{\zeta(\xi)} + \kappa_0 A_M , \\
\Xi'_{MN}(\xi) &=& (1+B\rme^{-\kappa_0\xi^1}) \big(\Xi_{MN}|_{\zeta(\xi)} -\Omega_{(M}|_{\zeta(\xi)}\;  A_{N)}  - \kappa_0 A_MA_N -D_M A_N\big)-\nonumber\\
&& - \frac{\rme^{-\kappa_0 \xi^1} }{\kappa_0} \big(\Omega_{(M}|_{\zeta(\xi)}\;  B_{N)} + \kappa_0 A_{(M} B_{N)} + D_M B_N \big).
\eea
The transformation of the quantity $\Sigma_{MN}^0$, which is invariant under \eqref{eq:dataST}, can be found by plugging the previous equations into its definition (\ref{eq:invariantXi}), giving
\bea
{\Sigma_{MN}^0}' &=& \Sigma_{MN}^0 + B\rme^{-\kappa_0 \xi^1} \(  \Sigma_{MN}^0 -\frac{1}{2}  D_{(M}\Omega'_{N)}-\Omega'_M \Omega'_N  \) -\nonumber \\
&& - \rme^{-\kappa_0 \xi^1}\(\Omega'_{(M}B_{N)}+D_M B_N \).
\eea
Taking the trace we find
\be \label{thetaTransf}
{\theta^0}' =\theta^0 + B\rme^{-\kappa_0\xi^1} \big({\theta^0} -\Omega'^M \Omega'_M - D^{M} \Omega'_{M}\big)- \rme^{-\kappa_0 \xi^1} \big(2 \Omega'^{M} B_{M} + D^M B_M \big).
 \ee
This quantity evolves according to \eqref{eq:theta0},
and thus, if at any given time (e.g. $\xi^1=0$) we can set ${\theta^0}'=\ft1{2}\cR$, then $\pd_{n'}{\theta^0}'=0$ for all $\xi^1$, where $n'$ is the vector resulting from the transformation of $n$ under the map $\zeta$. This requires to solve
\be
D^M B_M + 2 \Omega'^{M} B_{M} + \big(D^{M} \Omega'_{M} + \Omega'^M \Omega'_M -\theta^0 \big) B- (\theta^0 - \ft12\cR)=0 
\label{eq:Braw}
\ee
To simplify this expression we decompose the  Hajicek one-form $\Omega_M$  in its exact $\Omega_M^e = \pd_M \eta$ and divergence free parts $\Omega_M^0= {\varepsilon_{M}}^N \pd_N g$ as in \eqref{def:hodge}. Here  $\eta(\xi)$ and $g(\xi)$ are two smooth functions on $\Sigma$ satisfying $\pd_n \eta =\pd_n g=0$, and $\epsilon_{MN}$ is the volume one form associated to $q_{MN}$. The transformed  Hajicek one-form $\Omega_M'$ is then determined by the functions $\eta' = \eta + A$, and $g' = g$. 
With the change of variables $B = \rme^{-\eta'} \tilde B$ we find 
\bea
B_M &=& \tilde B_M \rme^{-\eta'} - \eta'_M \tilde B \rme^{-\eta'},\nonumber \\
D^M B_M &=& D^M \tilde B_M \rme^{-\eta'} - 2 \eta'_M \tilde B^M \rme^{-\eta'} + \eta'^M \eta'_M \tilde B \rme^{-\eta'}- D^M \eta'_M \tilde B \rme^{-\eta'}, \nonumber
\eea
which allows us to rewrite \eqref{eq:Braw} as 
\be \label{ellipticEqn}
D^M \tilde B_M + 2{\varepsilon_{M}}^N \tilde B^M g_N + \tilde B(g^M g_M  - \theta^0) - ( \theta^0 - \ft12 \cR ) \rme^{\eta'}=0.
\ee 
In this form we can easily identify two solutions to this equation
\be
\tilde B=-1, \hspace{1cm} A=\log \( \frac{\theta_0-g^Mg_M}{\theta^0-\ft12 \cR} \)-\eta ,
\ee 
and
\be
\tilde B=1, \hspace{1cm} A=\log \( \frac{g^Mg_M-\theta^0}{\theta^0-\ft12 \cR} \)-\eta ,
\ee
provided the argument of the logarithms in these equations are non-negative, and that  \eqref{fTransform} is well defined at $\xi^1=0$. We find the following possibilities
\bea
\tilde B=-1&:& \quad   g^Mg_M \ge  \ft12 \cR\ge \theta^0, \quad \text{or}\quad g^Mg_M \le  \ft12 \cR\le \theta^0, \label{sol2}\\
\tilde B=1&:& \quad   g^Mg_M \leq \theta^0 \leq \ft12 \cR, \quad \text{or}\quad g^Mg_M \geq \theta^0 \geq \ft12 \cR. \label{sol1}
\eea
Thus, each of these four sets of conditions (to be met a at the spatial slice $\cS_{\xi^1=0}$) is sufficient to ensure that gauge  $\pd_n \theta^0=0$ exists.

\textbf{Uniqueness of the gauge}. Now we will show that this gauge is unique up to supertranslations provided $ g^M g_M \leq \ft12 \cR$ on $\Sigma$. For this, suppose that we are already in this gauge. We would like to know what  the possible transformations which maintain this gauge are. They would have to solve (\ref{ellipticEqn}) with $\theta^0=\ft1{2}\cR$, 
\be
  D^M \tilde B_M + 2 \epsilon^{MN} \tilde B_{M} g_N + (g^M g_M  -\ft12 \cR)\tilde B =0.
\ee
We can multiply by $\tilde B$ and integrate over the sphere. After integrating the first two terms by parts and dropping the boundary terms (the spatial sections of $\cS_{\xi^1} \cong \mathbb{S} ^2$ are compact and simply connected) we get
\be
\int_{S^2} d^2\xi \big(\tilde B^M \tilde B_M + (\ft12 \cR - g^M g_M )\tilde B^2\big) =0.
\ee
This implies that, if  
\be \label{12Rgeqgg}
\ft12 \cR \geq g^M g_M
\ee
is satisfied 
the integral is the sum of two positive contributions, so we must have $B=\tilde B=0$. But $A$ is unconstrained, so the condition  $\pd_n \theta^0=0$ fixes the gauge up to transformations for which $f(\xi)=\xi^1 + A(\xi^M)$, i.e. supertranslations.  Then we arrive to the following conditions which guarantee that    the gauge \eqref{eq:dataRedundancies} can be reduced down to supertranslations
\bea
\tilde B=-1&:& \quad   g^Mg_M \le    \ft12 \cR\le \theta^0, \\
\tilde B=1&:& \quad   g^Mg_M \leq \theta^0 \leq \ft12 \cR.
\eea
Noting that $ g^Mg_M = \Omega^0{}^M \Omega^0_M$ we arrive to \eqref{eq:conditionsBgauge}.  

A priory it might seem that the first situation, for which $B<0$, is ill behaved because the domain of $\zeta(\xi)$ covers only the range
\be
\xi^1 \in \big[ -\ft1\kappa_0 \log (-1/B), \infty\big),
 \ee
 but this is not the case. Actually, the consistency condition that we should impose on $\zeta(\xi)$ is that given a tensor field defined  on $\Sigma$, e.g. $\gamma_{ab}$, the coordinate representation of the transformed tensor $\zeta^*\gamma(\xi)$ contains the same information as the original one $\gamma(\xi)$. Thus, we must require that the image of $\zeta$ is the full abstract manifold $\Sigma$. This condition ensures that scanning over the domain of $\zeta^*\gamma(\xi)$ we will access the full domain where $\gamma(\xi)$ is defined. It is straightforward to see that the image of $\zeta$ for  the two cases in \eqref{eq:conditionsBgauge} covers the following ranges of the null coordinate
\bea
(i) \; B<0&:& \quad  \hat f(\xi) \in (-\infty, \infty),\\
(ii) \; B>0 &:& \quad   \hat f(\xi) \in (A + \ft1 \kappa_0 \log B, \infty).
\eea
Then, only diffeomorphisms satisfying the condition $(i)$ are well behaved in the sense explained above. This is the condition we presented in the main text \eqref{eq:generic}.

\textbf{Supertranslation independent condition}. We will now prove that the condition $(i)$ in \eqref{eq:conditionsBgauge} 
is preserved by supertranslations, even before setting the gauge $\pd_n \theta^0=0$.
 To see this, first note that the condition (\ref{12Rgeqgg}) is preserved by supertranslations since both $\cR$ and $g_M$ transform as scalar fields under supertranslations, but neither of the two depend on the null coordinate, so they are actually invariant. 
  Finally,    $\theta^0$ also transforms as a scalar under supertranslations $\theta^0(\xi)\to{\theta^0}'(\xi)=\theta^0(\zeta(\xi))$, as it can be checked setting $B=0$ in (\ref{thetaTransf}). 
 However,  due to the form of \eqref{eq:theta0}, the RHS cannot change sign, implying that if for some $\xi^1$ the inequality  $\theta^0 \ge \ft12 \cR$ holds, then it will hold for all $\xi^1$, including the supertranslated one $\hat f(\xi)=\xi^1+A(\xi^M)$. Therefore, it is impossible to cross the bound $\theta^0 \ge \ft12 \cR$ with a supertranslation.

In conclusion, we have shown that if a given data set  satisfies the condition  $(i)$ in \eqref{eq:conditionsBgauge}, then there is always a gauge transformation of the form (\ref{fTransform}) which allows us to set $\pd_n \theta^0=0$ everywhere on the horizon. Furthermore, the gauge freedom that remains once we have done so is that of supertranslations.

\subsection{Weyl scalars}
\label{app:weyl}
%%%%%%%%%%%%%%%%%%%%%%%%%%%%%%%%%%%%%%%%

In the present section we will compute the  Weyl scalars $\Psi_n$, with $n=\{0,1,2,3\}$, at a generic point $\xi^a = \xi^a_0$ of a non-expanding horizon which is embedded in vacuum, i.e. $T_{\mu\nu}=0$.  As explained in the main text, the scalar $\Psi_4$  involves information about the spacetime connection off the hypersurface, and thus it cannot be computed from the connection coefficients \eqref{eq:connectionCoeff}. 

To be more precise, we will compute the pullback of $\Psi_n$ to the abstract manifold $\Sigma$, but we will keep the pullback operation implicit in order to ease  the notation. We will use the setting described in section \eqref{sec:WeylScalars}: we choose a coordinate system for the abstract manifold $\Sigma$ such that $q_{MN}(\xi_0) = \delta_{MN}$, and  we will define the Weyl scalars in terms of the Newman-Penrose tetrad $\cB_{NP}=\{n,\ell,m,\overline{m}\}$, where $m$ and $\overline{m}$ are defined by \eqref{eq:mbarm}. The Weyl scalars are given by the expressions
\bea
\Psi_0 &=& C_{\mu \nu \rho \sigma} n^\mu m^\nu n^\rho m^\sigma =\ft12  (C_{n2n2} - C_{n3n3}) + \rmi C_{n2n3},	\nonumber\\
\Psi_1 &=& C_{\mu \nu \rho \sigma} n^\mu m^\nu \ell^\rho n^\sigma =\ft{1}{\sqrt{2}}  (C_{n2\ell n} + \rmi C_{n3\ell n}), \nonumber \\
\Psi_2 &=& C_{\mu \nu \rho \sigma} n^\mu m^\nu \ell^\rho \overline m^\sigma = \ft12 (C_{\ell2n2} + C_{\ell3n3}) +\ft{\rmi}{2} (C_{\ell 2n3} - C_{\ell 3n2}), \nonumber \\
\Psi_3 &=&  C_{\mu \nu \rho \sigma} n^\mu \ell^\nu \overline m^\rho \ell^\sigma  = \ft{1}{\sqrt{2}} (C_{n\ell2\ell} - \rmi C_{n\ell3\ell}).
\label{eq:appPsidefs}
\eea
First we will express the scalars $\Psi_n$ in terms of the connection coefficients, and then we will compute the gauge corrected Weyl scalars, \eqref{eq:Psi012final} and \eqref{eq:Psi3final}, which we  introduced in section \ref{sec:NewmanPenrose}.

\begin{lemma}
Let  $\mathscr{D}=\{q_{MN}, \kappa, \Omega_M,\Xi_{MN}\}$ be the hypersurface data of a non-expanding horizon $\cH$ embedded   in the vacuum, and let  $\Psi_n$, $n=0,1,2,3$,  be the Weyl scalars defined with respect to the Newman-Penrose tetrad $\cB_{NP}=\{n,\ell,m,\overline{m}\}$ on $\cH$. Then,  the following equations hold
\be
\Psi_0=\Psi_1=0, \qquad \Psi_2 = -\ft14 \cR + \ft\rmi2 \cJ,
\label{eq:appPsi12a}
\ee
and 
\bea
\Re \Psi_3 &=& \ft{1}{\sqrt{2}} (D_{[M} \Xi_{2]N} +\Omega_{[M} \Xi_{2]N} )q^{MN},  \nonumber \\
 \Im \Psi_3 &=&-  \ft{1}{\sqrt{2}} (D_{[M} \Xi_{3]N} +\Omega_{[M}\Xi_{3]N}) q^{MN},
\label{eq:appPsi3a}
\eea
where  $D_{[M} \Omega_{N]} = \epsilon_{MN} \cJ$,  $\cR$ is the Ricci scalar of $q_{MN}$ and $\epsilon_{MN}$ the volume form.
\end{lemma}
\begin{proof}
Due to the  Einstein's equations  the Ricci tensor vanishes in vacuum $R_{\mu\nu}=0$, and the Weyl tensor is equal to the Riemann curvature tensor $C_{\mu\nu\rho \sigma} = R_{\mu\nu\rho \sigma}$.  Taking this into account we can compute $\Psi_0$ and $\Psi_1$ using  
the identities \eqref{eq:fundamentaForm} and  \eqref{eq:identity1} 
\bea
\Re(\Psi_0) &=& \ft12\Big(-\pd_n(\Theta_{22} - \Theta_{33}) +  \kappa(\Theta_{22} - \Theta_{33}) + (\Theta_{22}^2  - \Theta_{33}^2 )\Big), \\
  \Im(\Psi_0)&=& - \pd_n \Theta_{23}+ \kappa \Theta_{23} +  \Theta_2^C \Theta_{C3},\\
\Re(\Psi_1) &=& \ft{1}{\sqrt{2}}( \pd_n \Omega_2 - \pd_2 \kappa + \Theta_2^N \Omega_N),\\
Ê\Im(\Psi_1) &=& \ft{1}{\sqrt{2}} ( \pd_n \Omega_3 - \pd_3 \kappa + \Theta_3^N \Omega_N).
\eea
Note that, for non-expanding horizons embedded in  vacuum all these quantities vanish, $\Psi_0=\Psi_1=0$,
since the second fundamental form is zero $\Theta_{MN}=0$ (see section \ref{sec:redundancies}), and as a consequence of  the Damour-Navier-Stokes equation \eqref{eq:NS}. This proves the left equation in \eqref{eq:appPsi12a}. It is worth mentioning that this result could also have been obtained using the Goldberg-Sachs theorem (see \cite{Chandrasekhar:1985kt}), and noting the existence of a geodesic and shear free null vector, namely the null normal  $n$.

To compute $\Psi_2$ we can use the following two identities
\bea
R_{\ell N nM} q^{NM}&=& (- \pd_n \Xi_{MN}  - D_M \Omega_N - \Omega_M \Omega_N -\kappa \Xi_{MN} + \Theta_N^P \Xi_{MP})q^{NM},\nonumber\\
R_{\ell [N n M]} &=& D_{[N} \Omega_{M]} - \Xi_{[N}^L \Theta_{M]L}.
\eea
The first one follows from \eqref{eq:appXider1}, and the second one from  
the  Ricci identity $R_{\ell2 n 3} =\ell_\rho n^\mu e_3^\nu \nabla_{[\mu} \nabla_{\nu]} e_2^\rho$. From them we obtain
\be
\Re(\Psi_2) = \ft12(- \pd_n \theta^\ell - D_M \Omega^M  - \kappa \theta^\ell - \Omega^A \Omega_A + \Theta^{AB} \Xi_{AB} ), \qquad 
\Im(\Psi_2) = \ft12  D_{[2} \Omega_{3]}. 
\ee
Here we have also used the fact that the real part of $\Psi_2$ can also be written as $\Re \Psi_2 = \ft12 q^{AB} C_{\ell A n B}$.   The expressions \eqref{eq:Psi012final} can be recovered when we impose the constraint equations of a non-expanding horizon embedded in vacuum. Setting $\Theta_{MN}=0$, and  from  the constraint equation for the trace of $\Xi_{MN}$, \eqref{eq:Xi} we arrive to
\be
\Re(\Psi_2) = - \ft14\cR, \qquad 
\Im(\Psi_2) = \ft12  D_{[2} \Omega_{3]}. 
\ee
At the point where we are evaluating the expressions,  $\xi_0^a$,  the spatial metric has the canonical form $q_{MN}=\delta_{MN}$, and therefore the volume form reduces to the Levi-Civita symbol, which satisfies  $\epsilon_{23}=1$. Thus $D_{[2} \Omega_{3]} \equiv \epsilon_{23} \cJ=\cJ$, which proves \eqref{eq:appPsi12a}.

To compute the  last Weyl scalar $\Psi_3$ it is convenient to use the equation $R_{n \ell A\ell}=R_{\ell MAN}q^{MN}$  which holds in vacuum. It follows from
\be
0=R_{\ell A}= R_{\ell\mu A \nu} g^{\mu \nu} = R_{\ell MAN} q^{MN} + R_{\ell nA \ell} + R_{\ell\ell A n}= R_{\ell MAN} q^{MN} - R_{n\ell A\ell}
\ee
with $M \neq A$. 
Here we used the form for the inverse metric \eqref{eq:inverseG}, and  the symmetries of the Riemann tensor, which imply $R_{\ell\ell A n} =0$.
Then, the contractions of the Riemann curvature of the form $R_{\ell M A N}$ can be calculated from the relation 
\be
R_{\ell M A N} = D_{[N} \Xi_{A]M} + \Omega_{[N} \Xi_{A]M} ,
\label{eq:appRicciIdentity1}
\ee
which is a direct consequence of  the Ricci identity $R_{\ell M A N}  = \ell_\sigma e^\mu_N e^\nu_M  \nabla_{[\mu} \nabla_{\nu]} e^\sigma_A$,  and the definitions   of the connection coefficients \eqref{eq:connectionCoeff}.  Recalling that the Riemann and Weyl tensors are equal in vacuum we have that $C_{n \ell A\ell} = R_{\ell MAN} q^{MN}$, we can obtain \eqref{eq:appPsi3a} using \eqref{eq:appRicciIdentity1} and the definitions \eqref{eq:appPsidefs}.
\end{proof}
\finn

 Recall, that the gauge corrected Weyl scalars are given by 
 \be
 \Psi^c_n(\eta,\xi^M)\equiv\Psi^c_n(H(\eta,\xi^M),\xi^M),\nonumber
 \ee
  where $H(\eta,\xi^M)$ is defined in \eqref{eq:defH}. Therefore, since  the Weyl scalars 
 $\Psi_0$, $\Psi_1$ and $\Psi_2$ do not depend on the null coordinate $\xi^1$,  their gauge corrected expressions are identical to those in \eqref{eq:appPsi12a}, which proves \eqref{eq:Psi012final}. 
 
It only remains derive the expression  \eqref{eq:Psi3final} for the gauge corrected Weyl scalar $\Psi_3^c(\eta,\xi^M)$.
\begin{proposition}
Let $\mathscr{D}=\{q_{MN}, \kappa, \Omega_M,\Xi_{MN}\}$ represent the hypersurface data of a generic non-expanding horizon embedded   in  vacuum, and let  $\Psi_3^c(\eta,\xi^M)$ be the gauge corrected Weyl scalar defined in \eqref{eq:GCWdefs}. Then, in the gauge defined by \eqref{eq:GaugeFix1} and \eqref{eq:GaugeFix2},  $\Psi_3^c(\eta,\xi^M)$ is given by \eqref{eq:Psi3final}.
\end{proposition}

\begin{proof}
 We will first write $\Psi_3$ in terms of the object $\Sigma_{MN}^0$ defined  \eqref{eq:invariantXi}. Substituting the definition \eqref{eq:invariantXi} into \eqref{eq:appRicciIdentity1}, after a straightforward calculation we find
\be
\kappa C_{n \ell A\ell} = (D_{[N} \Sigma_{A]M}^0 + \Omega_{[N} \Sigma_{A]M}^0)  q^{MN} +\ft12 \epsilon_{AM} D^M \cJ + \ft32 \epsilon_{AC} \Omega^C  \cJ  - \ft12 \Omega_A \cR,
\ee
where  $\cR$ and $\epsilon_{MN}$ are respectively the curvature scalar and volume form of $q_{MN}$, and $\cJ = D_{[2} \Omega_{3]}$. In order to simplify this expression we can use the assumption that the horizon is generic, and that the gauge redundancies \eqref{eq:dataRedundancies}
 have been partially fixed by the conventions \eqref{eq:GaugeFix1} and \eqref{eq:GaugeFix2}. Then, the trace of $\Sigma_{MN}^0$ satisfies $\theta^0 = \ft12 \cR$, and thus the first term in the previous equation takes the form
\be
 (D_{[N} \Sigma_{A]M}^0 + \Omega_{[N} \Sigma_{A]M}^0)  q^{NM}  = D^M \sigma^0_{AM} + \Omega^M \sigma_{AM}^0 - \ft14 \pd_A \cR -\ft14 \Omega_A \cR,
\ee
where $\sigma_{MN}^0$ is the traceless part of $\Sigma_{MN}^0$. This leads to 
\be
\kappa C_{n \ell A\ell} =D^M \sigma^0_{AM} + \Omega^M \sigma_{AM}^0  - \ft14 \pd_A \cR +\ft12 \epsilon_{AM} D^M \cJ + \ft32 \epsilon_{AC} \Omega^C \cJ   - \ft34 \Omega_A \cR.
\ee
Then from the definition of $\Psi_3$ we find
\bea
\Re \Psi_3 &=&\ft1{\kappa\sqrt{2}}(D^M \sigma^0_{2M} + \Omega^M \sigma_{2M}^0  - \ft14 \pd_2 \cR +\ft12  D_3 \cJ + \ft32  \Omega_3 \cJ   - \ft34 \Omega_2 \cR)\nonumber\\
\Im \Psi_3 &=&\ft1{\kappa\sqrt{2}}(-D^M \sigma^0_{3M} - \Omega^M \sigma_{3M}^0  + \ft14 \pd_3 \cR +\ft12 D_2 \cJ + \ft32  \Omega_2 \cJ   + \ft34 \Omega_3 \cR),
\eea
or equivalently
\be
\Psi_3 =\ft1{\kappa\sqrt{2}}(D\sigma^0  +\hat D \Psi_2 + 3\, \hat\Omega \, \Psi_2),
\ee
where we used the shorthands $\hat D \equiv D_2 - \rmi D_3$ and $\hat \Omega \equiv \Omega_2 - \rmi \Omega_3$, and we defined the complex scalar
\be
D\sigma^0 \equiv D^M \sigma_{2M}^0 +  \Omega^M \sigma_{2M}^0 - \rmi  (D^M \sigma_{3M}^0 + \Omega^M \sigma_{3M}^0).
\ee
Then, the gauge corrected Weyl scalar is given by
\bea
\Psi_3^c(\eta,\xi^M) &=& \ft1{\kappa\sqrt{2}}(D\sigma^0  +\hat D \Psi_2 + 3\, \hat\Omega \, \Psi_2)|_{\xi^1=H(\eta,\xi^M)} 	\nonumber \\
&=& \ft1{\kappa\sqrt{2}}D\sigma^0|_{\xi^1=H(\eta,\xi^M)}  +\ft1{\kappa\sqrt{2}}(\hat D \Psi_2^c + 3\, \hat\Omega \, \Psi_2^c),
\label{eq:appPsi3int}
\eea
where the second equality follows from the fact that  neither $\Psi_2$, $\Omega_M$, or $q_{MN}$ depend on $\xi^1$, and thus their functional form  in  unchanged after evaluating them in $\xi^1=H(\eta,\xi^M)$.
Therefore we just have to compute the first term on the right in the last equation, which reads
\bea
D\sigma^0|_{\xi^1 = H(\eta,\xi^M)} &=& D^M \sigma_{2M}^0|_{\xi^1=H(\eta,\xi^M)} +  \Omega^M \sigma_{2M} \nonumber \\ 
&&- \rmi  (D^M \sigma_{3M}^0|_{\xi^1=H(\eta,\xi^M)} + \Omega^M \sigma_{3M}).
\eea
In the previous expression we already made the substitution  $\sigma_{MN}^0|_{\xi^1=H(\eta,\xi^M)} = \sigma_{MN}$, using the  definition of the supertranslation invariant variable $\sigma_{MN}$, i.e.  \eqref{eq:defSigmaGI}.  From the definition \eqref{eq:defSigmaGI} it is immediate to check that  the following relation holds
\bea
 D^M \sigma_{2M} &=& (D^M \sigma_{2M}^0 + \pd_n \sigma_{2M}^0 \, \pd^M H)|_{\xi^1=H(\eta,\xi^M)} \nonumber\\
&=&  (D^M \sigma_{2M}^0 + \sigma_{2M}^0 \, \Omega^{e|M})|_{\xi^1=H(\eta,\xi^M)} \nonumber\\
&=& D^M \sigma_{2M}^0|_{\xi^1=H(\eta,\xi^M)} + \sigma_{2M} \, \Omega^{e|M}
\eea
where, for the second equality, we have used the equations \eqref{eq:finalEqs} and \eqref{eq:defH}, and that the exact part of the Hajicek one form is given by $\pd_M \eta = \Omega_M^e$. The last result allows us to express $\Psi_3^c$ in terms of the gauge invariant variable $\sigma_{MN}$ as follows
\be
\Psi_3^c(\eta,\xi^M) = \ft1{\kappa\sqrt{2}}(D\sigma  +\hat D \Psi_2^c + 3\, \hat\Omega \, \Psi_2^c),
\label{eq:appFinalPsi3}
\ee
where $D\sigma$ is now is defined as
\be
D\sigma(\eta,\xi^M)  \equiv D^M \sigma_{2M} +  \Omega^{0|M} \sigma_{2M} - \rmi  (D^M \sigma_{3M} + \Omega^{0|M} \sigma_{3M}).
\ee
It can be seen  that the exact part of the Hajicek one form $\Omega_M^e$ has been cancelled out, and thus $D\sigma$  only involves the divergence free part $\Omega^0_M$. Substituting the solution to the constraint equations \eqref{eq:finalSol} in   \eqref{eq:appFinalPsi3}  we arrive  to our final result, which is given by \eqref{eq:Psi3final}.
\end{proof}
\finn

\section{Calculations for null infinity}
%%%%%%%%%%%%%%%%%%%%%%%%%%%%%%%%%%%%%%%%%%%%%%%%%%%%%%%%
%%%%%%%%%%%%%%%%%%%%%%%%%%%%%%%%%%%%%%%%%%%%%%%%%%%%%%%%

\subsection{Derivation of the BMS group.}
\label{app:BMSgroup}
%%%%%%%%%%%%%%%%%%%%%%%%%%%%%%%%%%%%%%%%

In this appendix we present a derivation of the  BMS group of transformations at null infinity. We show that it can be described as the set of  diffeomorphisms of the unphysical spacetime which preserve null infinity as a set of points,  and that leave invariant both the metric tensor and the null normal on $I$ up to  a conformal transformation.  
The analysis is done in a similar fashion as in section \ref{sec:supertranslationsGauge}, where we studied the set of diffeomorphisms preserving the metric tensor at a  non-expanding horizon, i.e. the hypersurface symmetries, which include horizon supertranslations. Therefore, the present computation also serves as a check for our approach in section \ref{sec:supertranslationsGauge} to characterise horizon supertranslations. 

We will now characterise the set of  diffeomorphisms of the unphysical spacetime $F: \cM \to \cM$ which preserve the structure of null infinity implied by the definition of asymptotically flat spacetimes given in section \ref{sec:Scri}, i.e. conditions (i), (ii) and (iii).  More specifically, any diffeomorphism $F$ in this set should satisfy the following conditions: 
\begin{itemize}
\item[(i)] Leave invariant the scalar products  at $\cI$ up to a conformal transformation. That is, denoting  $g_{\mu \nu }'\equiv (F^*g)_{\mu \nu}$, and $ \Omega'\equiv F^*\Omega$ we should have
\be
g_{\mu \nu }' \;\hat{=}\; \omega^2 g_{\mu\nu}, \quad \text{and} \quad \Omega' \;\hat{=}\; \omega \Omega,
\label{eq:conformalMap}
\ee
so that $\Omega'{}^{-2} g_{\mu\nu}' \;\hat{=}\; \hat g_{\mu\nu}$, where  $\hat =$ denotes equality on $\cI$.
\item[(ii)] Map null infinity to itself,  $F(\cI) = \cI$, or equivalently $F^* \Omega \eqscr \Omega \eqscr 0$.
\item[(iii)] Preserve the definition of the null normal, $\bn \equiv d\Omega$.
\end{itemize}
The properties of this set of diffeomorphisms are more easily studied  using  a coordinate system of the unphysical spacetime adapted to $\cI$. We proceed as in section \ref{sec:supertranslationsGauge} for the case of non-expanding horizons. Since  null infinity $I$ is diffeomorphically identified with the abstract manifold $\cI$ via the embedding map $\Phi$, we can use the  coordinates on $\cI$ the later one, $\xi^a$, to parametrise  the hypersurface  $I$. The coordinate system on $I$ is then extended off the hypersurface introducing a transverse coordinate $r$, which is defined in terms of the rigging as $\ell = \pd_r$, with $r(\cI)=0$, and then keeping the coordinates $\xi^a$ constant along  the integral curves of $\ell$. Thus, the  coordinate system for the unphysical spacetime reads $x^\mu = \{\xi^1, r, \xi^M\}$, so that the embedding map takes the simple form
\be
\Phi: \xi^a \longrightarrow x^\mu = \{u =\xi^1, r=0, x^M = \xi^M\}.
\ee
The elements of coordinate basis $\cB=\{n, \ell, e_M\}$ have the following explicit form 
\be
 n = \pd_0, \qquad \ell= \pd_1,  \qquad e_M = \pd_M.
\ee
Then the null normal has the coordinate form $\bn = dr$, what follows  from our conventions in section \ref{sec:nullGeometry}, $\bn(\ell)=1$, and the properties of  the null  normal   $\bn(e_M) = \bn(n)=0$. Moreover, on this coordinate system the metric tensor has the following form at $\cI$
\be
g_{\mu\nu}|_{\cI} = \begin{pmatrix}
0 & 1 & 0 \\
1 & 0 & 0 \\
0 & 0 & q_{MN}
\end{pmatrix}.
\ee
Let us turn to the characterisation of the properties of the diffeomorphisms $F$ satisfying the conditions ($i$), ($ii$) and ($iii$) above.   From the condition $(iii)$ it is immediate to find the required behaviour of the null normal under the pull back $F^*$. Indeed, since $\Omega\;\hat{=}\;0$, we have 
\be
d\Omega' \;\hat{=}\;(\omega d\Omega + \Omega d\omega)\;\hat{=}\; \omega d\Omega \qquad \Longrightarrow \qquad  \bn' \;\hat{=}\; \omega \bn.
\label{eq:BMScond1}
\ee
where $\bn'\equiv F^*\bn$.
From this we can also find the transformation of the normal vector $n = g^{-1}(\bn,\cdot)$ under the  pushforward of $F$.
On the one hand, from the definition of $n$ we have that for any $k \in T_p\cI$
\be
F^*g(n , k)  \eqscr\omega^2 g(n , k) \eqscr \omega^2 \bn(k).
\label{eq:pullbn}
\ee
On the other hand, since $F$ maps $\cI$ to itself, it follows that $n$ can change at most by a rescaling $dF(n) =\alpha n$. Then 
\bea
F^*g(n , k) &\eqscr& g(\dd F( n) , \dd F(k)) \eqscr \alpha g( n , \dd F(k)) \nonumber \\
&\eqscr& \alpha \bn(\dd F(k)) \eqscr \alpha F^*\bn(k) \eqscr \alpha \omega \bn(k).
\label{eq:pushn}
\eea
Comparing the two previous expressions we find $\alpha=\omega$. Let $y^\alpha=y^\alpha(x)$ be  the explicit form for the diffeomorphism $F$, in a given set of coordinates, then from the condition on the pullback of $\bn$ \eqref{eq:pullbn} we find the following constraints on the mapping $F$
\be
F^*\bn_\mu \eqscr y_\mu^\alpha \bn_\alpha \qquad \LongrightarrowÊ\qquad y_\mu^1 \eqscr\omega \delta_\mu^1,
\ee
where we use the short hands $y_\mu^\alpha = \pd_\mu y^\alpha$.
From the condition on the pushforward of $n$ \eqref{eq:pushn} we find 
\be
\dd F(n)^\alpha \eqscr y^\alpha_\mu n^\mu \qquad \Longrightarrow \qquad y_0^\alpha \eqscr \omega \delta_0^\alpha.
\ee
Collecting both results we have
\be
y_0^0 \eqscr y_1^1 \eqscr \omega, \qquad y_0^1 \eqscr y_M^1 \eqscr y_0^I\hat =0.
\label{eq:STcond}
\ee
The mapping $F$ satisfies  conditions \eqref{eq:conformalMap}, if and only if the scalar products on the basis $\cB$ are mapped as follows
\bea
(F^*g)(n,n) \eqscr0, &\qquad& (F^*g)(e_M,e_N) \eqscr\omega^2 q_{MN}\\
(F^*g)( n,e_M) \eqscr0, &\qquad& (F^*g)(e_M,\ell) \eqscr0\\
(F^*g)(n,\ell) \eqscr\omega^2, &\qquad& (F^*g)(\ell,\ell) \eqscr0
\eea
In components they read
\bea
g_{\alpha \beta} y^\alpha_1 y^\beta_1\eqscr g_{11} \omega^2\eqscr0 &\qquad&  g_{\alpha \beta} y^\alpha_M y^\beta_N \eqscr q_{IJ} y^I_M y^J_N \eqscr \omega^2 q_{MN}\\
g_{\alpha \beta} y^\alpha_0 y^\beta_M \eqscr \omega\,  y^1_M \eqscr0 &\qquad&  g_{\alpha \beta} y^\alpha_M y^\beta_1  \eqscr  y^0_M \omega +q_{IJ} y_1^I y_M^J \eqscr 0\\
g_{\alpha \beta} y^\alpha_0 y^\beta_1 \eqscr\omega^2 &\qquad&  g_{\alpha \beta} y^\alpha_1 y^\beta_1 \eqscr2y^0_1 \omega +  q_{IJ} y^I_1 y^J_1 \eqscr 0
\eea
The second equation of the first line implies that, on $\cI$, $Y^I(x^M) \equiv y^I(x^M)|_{r=0}$ define a conformal symmetry of the metric $q_{MN}$ with conformal factor $\omega$, 
\begin{empheq}[box=\widefboxb]{align}
q_{IJ} Y^I_M Y^J_N = \omega^2 q_{MN}.
\label{eq:nullInfLorentz}
\end{empheq}
Note that, since $Y^I$ are constant along the null coordinate $u$, the conformal factor  must satisfy $\cL_n \omega\eqscr0$. If we restrict ourselves to globally well defined transformations, that is, to one-to-one mappings of the spatial sections of $\cI$ on to themselves, then the functions $Y^I$ generate a group isomorphic to the homogeneous orthochronous Lorentz group (see \cite{Sachs:1962zza}).

The action of the diffeomorphism on the null coordinate at $\cI$ is determined by the function $f(u,x^M) \equiv y^0(u,x^M)|_{r=0}$, which is constrained by the first equality in  \eqref{eq:STcond}, namely $y_0^0 \eqscr\omega$. Thus, the function $f(u,x^M)$ has the general form
\begin{empheq}[box=\widefboxb]{align}
f(u,x^M) = \omega(x^M) \big( u + A(x^M)\big),
\label{eq:nullInfST}
\end{empheq}
where $A(x^M)$  can be any smooth function of the spatial coordinates $x^M$.
The remaining non-trivial conditions can be solved in terms of $f$ and $Y^I$ to give
\be
y_1^I \eqscr  -\frac{1}{\omega} f_M q^{MN} Y_N^I,  \qquad y^0_1 \eqscr -  \frac{1}{2\omega}  f_M f^M.
\ee
The set of transformations determined by the functions $(f,Y^I)$ given by \eqref{eq:nullInfLorentz} and \eqref{eq:nullInfST} define the BMS group (see e.g. chapter 1 in  \cite{Newman:1981fn}). \emph{Null infinity supertranslations} can be  identified as those transformations of the BMS group  with $\omega\eqscr1$, and $Y^I(x) = x^I$, that is
\begin{empheq}[box=\widefboxb]{align}
f(u,x^M) = x^0 + A(x^M).
\end{empheq}
The infinitesimal version of the defining conditions of BMS transformations  can be recovered setting $y^\alpha(x) \approx x^\alpha + \epsilon k^\alpha(x)$, in \eqref{eq:conformalMap} and \eqref{eq:BMScond1}, where  the vector field $k^\alpha$ is the corresponding generator and $\epsilon\ll1$  is a small real parameter. We obtain
\be
\cL_{k} g_{\mu\nu} = 2 \lambda  g_{\mu\nu}, \qquad \cL_k  n = - \lambda  n,
\ee 
where $\omega \approx 1 + \lambda$, and $\cL_n\lambda=0$.  This is precisely the definition used to characterised the BMS group in the works by Geroch and Ashtekar  \cite{Geroch1977,Ashtekar:1987tt} (see also \cite{Ashtekar:2014zsa}). For the second equality we have used that the definition of the Lie derivative of a vector field  involves  the pushforward of the inverse mapping $F^{-1}$, this explains the minus sign on the second expression. 

From the previous equations it is also  straightforward to find the action of a supertranslation on the tensor $Y_{ab}$ at null infinity. Indeed, if we adopt  the gauge conventions in section \ref{sec:Scri}, null infinity can be described as a non-expanding null hypersurface. Moreover, since supertranslations preserve exactly the metric tensor and the null normal at null infinity they can be identified with a hypersurface symmetry of $\cI$. Therefore,  we can use the results in section \ref{sec:supertranslationsGauge} to find the transformation properties of the tensor $Y_{ab}$ under null infinity supertranslations, which hold for arbitrary hypersurface symmetries of a generic non-expanding null hypersurface.  Setting $\kappa= \Omega=0$ and $\hat f_1=1$ on equation \eqref{eq:Ytransform} we obtain
\be
Y_{ab}' =\begin{pmatrix}
0  & 0\\
0  &\Xi_{MN}
\end{pmatrix} - \begin{pmatrix}
0  & 0\\
0  &D_M  A_N 
\end{pmatrix},
\ee 
which is the relation we have used in the main text \eqref{eq:scriST}.

\subsection{Boundary conditions for no-outgoing radiation}
\label{app:scriVacua}
%%%%%%%%%%%%%%%%%%%%%%%%%%%%%%%%%%%%%%%%

In this appendix we will derive the equations (\ref{eq:sch1}-\ref{eq:sch3}), that is, the constraints satisfied by the Schouten tensor $S_{\mu\nu}$ of the unphysical spacetime in the absence of out going gravitational radiation at null infinity. Our starting  point are the boundary conditions \eqref{eq:noScriRad} expressed in terms of the leading order Weyl scalars $\Psi^0_n$, and the Bianchi identities \eqref{eq:ScriFieldEq}, relating the leading order unphysical Weyl tensor $K_{\mu\nu\rho\sigma}$ with $S_{\mu\nu}$. 

We prepare our set up as described in section \ref{sec:WeylScalars}.   Given point $\xi^a_0$ at null  infinity, the Weyl scalars $\Psi_n^0$ can be expressed in terms  of the  Newman-Penrose null $\cB_{NP} = \{\ell,n, m,\overline m\}$, where $\ell$ is the rigging vector, $n$  the null normal vector to $\cI$, and the complex null vectors $m$ and $\overline m$ are defined as in \eqref{eq:mbarm}. In addition we choose the coordinates on $\cI$ such that, at the point $\xi^a = \xi^a_0$, the spatial metric has the canonical form $q_{MN}= \delta_{MN}$.

\paragraph{Condition on the second Weyl scalar.} We begin deriving the constraint on $S_{\mu\nu}$ which follows from imposing $\Im \Psi^0_2=0$ at null infinity. The second Weyl scalar $\Psi_2^0$ has the form
\be
\Psi_2^0 =K_{\ell m n\overline m}, \quad \Longrightarrow \quad \Im \Psi_2^0 = \ft12 (K_{\ell3n2} - K_{\ell2 n3})
\ee
Here we will  use the notation $K_{\ell m n\overline m} = K^\mu_{\nu \rho \sigma} \ell_ \mu m^\nu n^\rho\overline  m^\sigma$, and similar expressions for  contractions of tensor with the elements of a basis. The last term can be rewritten using the symmetries of the Weyl tensor, and the first (algebraic) Bianchi identity
\be
\Im \Psi_2^0 = \ft12 (K_{\ell3n2} - K_{\ell2 n3}) =  \ft12 (K_{n2\ell3} + K_{n3 2\ell }) =- \ft12 K_{n\ell32} = \ft12 K_{32\ell n}.
\ee
Using  the Bianchi identity \eqref{eq:ScriFieldEq}, together with the definitions for the connection coefficients \eqref{eq:connectionCoeff}, and the gauge conventions \eqref{eq:divFree} we find
\be
\Im \Psi_2^0 = \ft12 K_{32\ell n} = -\ft12  \nabla_{[\mu} S_{\nu] \rho} \ell^\rho e^\mu_3 e^\nu_2 = -\ft12 ( D_{[3} S_{2] \ell} - \Xi_{[3}^C S_{2]C}),
\ee
where $S_{M \ell} = S_{\mu\nu} e^\mu_M \ell^\nu$ and $S_{MN} = S_{\mu\nu} e^\mu_M e^\nu_N$. It is easy to check that only the traceless part of $\Xi_{MN}$ contributes in the previous equation. The  condition $\Im \Psi_2^0=0$ implies that the previous expression should vanish, what can be written in more covariantly as 
\be
D_{[M} S_{N] \ell} = \Xi_{[M}^P S_{N]P}.
\label{eq:appSc1}
\ee
\paragraph{Condition on the third Weyl scalar.} The vanishing of the third Weyl scalar on $\cI$, $\Psi_3^0=0$, leads to two constraints on the Schouten tensor.
First let us write third Weyl scalar as
\be
\Psi_3^0 = K_{\ell  n n \overline m} = K_{n \overline m \ell  n} = \ft1{\sqrt{2}} (K_{n 2 \ell n} - \rmi K_{n3\ell n})
\ee
Using the Bianchi identity \eqref{eq:ScriFieldEq}, together with  \eqref{eq:connectionCoeff} and \eqref{eq:divFree}, we can rewrite this expression as
\be
K_{n M \ell n} = - n^\mu e_M^\nu \nabla_{[\mu} S_{\nu] \rho} \ell^\rho = - (\pd_n S_{M\ell} - \pd_M S_{n\ell}),
\ee
and therefore
\be
\Psi_3^0 =- \ft1{\sqrt{2}} (\pd_{[n} S_{2]\ell}  - \rmi \pd_{[n} S_{3]\ell} ).
\ee
In the absence  out out going radiation on $\cI$ the previous expression vanishes, or equivalently 
\be
\pd_{[n} S_{M]\ell} =0.
\label{eq:appSc2}
\ee
The second constraint for the Schouten tensor can be found writing the third Weyl scalar as follows
\be
\Psi_3^0= \ft1{\sqrt{2}} (K_{n 2 \ell n} - \rmi K_{n3\ell n}) = -  \ft1{\sqrt{2}} (K_{323n} - \rmi K_{232n}) =  \ft1{\sqrt{2}} (D_{[3}S_{2]3}  - \rmi D_{[2}S_{3]2} ).
\label{eq:WeylWeyl3}
\ee
The second equality is a consequence of the Weyl tensor being traceless,  $g^{\mu\nu} K_{\mu 2 \nu n}=0$. In particular, using  equation \eqref{eq:inverseG} for the inverse metric, with $q_{MN}=\delta_{MN}$, and the symmetries of the Weyl tensor we find
\be
g^{\mu\nu} K_{\mu 2 \nu n} = K_{\ell 2 n n} + K_{n 2 \ell n} + K_{2 2 2 n} + K_{3 2 3 n} =K_{n 2 \ell n}  + K_{3 2 3 n} =0.
\ee
The last equality in \eqref{eq:WeylWeyl3}  is obtained after using the Bianchi identity \eqref{eq:ScriFieldEq}, together with  \eqref{eq:connectionCoeff} and \eqref{eq:divFree}. Thus, the vanishing of the third Weyl scalar also implies  the tensor equation 
\be
D_{[M}S_{N]P} =0.
\label{eq:appSc3}
\ee
\paragraph{Condition on the fourth Weyl scalar.} The last constraint on $S_{\mu\nu}$ is obtained from the vanishing of $\Psi_4^0$, which reads
\be
\Psi_4^0 = K_{\overline mn \overline m n} = \ft12 (K_{2n2n} - K_{3n3n}) - \rmi K_{2n3n}
\ee
The Bianchi identities  \eqref{eq:ScriFieldEq},   \eqref{eq:connectionCoeff} and \eqref{eq:divFree} imply $K_{MnNn} = \pd_n S_{MN}$, and 
thus 
\be
\Psi_4^0 = \ft12 (\pd_n S_{22} - \pd_n S_{33}) - \rmi \pd_n S_{23}=0.
\ee
In addition, the Schouten tensor satisfies $S_M^M = \cR$, which together with  $q_{MN}= \delta_{MN}$, implies  $ \pd_n\cR = \pd_n S_{22} +\pd_n S_{33} = 0$. Therefore, we can summarise the constraints which follow from $\Psi_4^0=0$ in the tensor equation
\be
\pd_n S_{MN}=0.
\label{eq:appSc4}
\ee
This completes our proof of equations (\ref{eq:sch1}-\ref{eq:sch3}).

\subsection{Constraint equations at null infinity and radiative vacua}
%%%%%%%%%%%%%%%%%%%%%%%%%%%%%%%%%%%%%%%%

In this appendix we prove various results needed in section \ref{sec:Scri}  to derive the constraint equations of null infinity, and to find their solutions  in the absence of outgoing radiation.  

\subsubsection*{Consistency check for the  constraint equations}
%%%%%%%%%%%%%%%%%%%%%%%%%%%%%%%%%%%%%%%%

In this section we check that  the dependence of $\Xi_{MN}$ on the null coordinate $\xi^1$  implied by the constraint equations \eqref{eq:XiScri}   is consistent with the identity \eqref{eq:XiScri2}. 
\begin{proposition}
Let $\Xi_{MN}$ be a solution to the constraint equations of null infinity \eqref{eq:XiScri}. Then, if the relation \eqref{eq:XiScri2} is satisfied at a particular value of the null coordinate  $\xi^1_0$, then it will be satisfied for all values of $\xi^1$. 
\end{proposition}
\begin{proof}
The constraint equation for $\Xi_{MN}$ at null infinity can be obtained from \eqref{eq:Xi} substituting the gauge fixing conditions  \eqref{eq:divFree}, and  the form of the source term \eqref{eq:sourceTermsScri}. The result is
 \be
\pd_n \Xi_{MN} = -\ft12 ( S_{MN} + S_{n\ell}q_{MN} ), 
\label{eq:appScriVac}
 \ee
Note that this equation reduces to  \eqref{eq:XiScri} when we express it in terms of the equivalence relation \eqref{eq:equiv}. The relation \eqref{eq:XiScri2}  is satisfied at a given value of the null coordinate $\xi^1=\xi^1_0$, they will be satisfied for all values of $\xi^1$ provided the following expression vanishes  on $\cI$
\be
\pd_n (D_{[M} \Xi_{N]P} - \ft12 q_{P[M} S_{N] \ell}) = D_{[M} \pd_n \Xi_{N]P} - \ft12 q_{P[M} \pd_n S_{N] \ell},
\ee
where the equality is obtained using that  the metric $q_{MN}$ and its Levi-Civita connection are independent of $\xi^1$. Thus, we need to prove that the previous expression is zero. Substituting the constraint equation \eqref{eq:appScriVac}   we find 
\be
D_{[M} \pd_n \Xi_{N]P}- \ft12 q_{P[M} \pd_n S_{N] \ell} = -\ft12 D_{[M}  S_{N]P}- \ft12 \pd_{[M} S_{n \ell} \, q_{N]P}- \ft12 q_{P[M} \pd_n S_{N] \ell}.
\ee
As the indices $M,N,P=\{1,2\}$ and $N \neq M$, then $P$ must be either equal to $M$ or $N$. Without loss of generality we choose $P=M$. Moreover, in the following we will also assume that we have chosen the coordinates so that $q_{MN} =\delta_{MN}$ locally.  We obtain
\be
 -\ft12 (D_{[M}  S_{N]M}   +   \pd_{[n} S_{N]\ell}) = \ft12( K_{MNMn} + K_{nN\ell n}),
\ee 
where we have also used the Bianchi identity  \eqref{eq:ScriFieldEq} (contracted with  elements of the basis $\cB=\{n,\ell, e_M\}$) in the second equality. Using  the symmetries of the Weyl tensor, and that $q_{MN} = \delta_{MN}$, we find
\be
\ft12 (K_{MNMn} +K_{nN\ell n})=\ft12(q^{AB} K_{ANBn} + K_{nN\ell n}),
\ee
where $A,B$ run over $\{1,2\}$. The previous expression 
can be written in the form
\be
\ft12 (q^{AB} K_{ANBn} +K_{nN\ell n} )=\ft12g^{\mu \nu} K_{\mu N \nu n},
\ee
where we used the formula \eqref{eq:inverseG} for the inverse metric. Summarising, we have found the relation
\be
\pd_n (D_{[M} \Xi_{N]M} - \ft12 q_{M[M} S_{N] \ell})  = \ft12 g^{\mu \nu} K_{\mu N \nu n}
\ee
which can be easily checked to be always zero, 
 since  the contraction of any two indices of the Weyl tensor is always zero.
\end{proof}
\finn

\subsubsection*{Schouten tensor at null infinity with no outgoing radiation}

We will begin proving the relation \eqref{eq:vacSchouten} which is satisfied by the Schouten tensor in  the absence of outgoing radiation through $\cI$. 
According to our discussion above, the boundary  conditions \eqref{eq:noScriRad}  for no out going radiation imply the  constraints \eqref{eq:appSc3}, \eqref{eq:appSc4} for $S_{\mu\nu}$. Moreover,  as we proved in section \ref{sec:Scri}, in the divergence free conformal gauge \eqref{eq:divFree} the Schouten tensor  also satisfies   $S^M_M = \cR$.  
\begin{proposition} The tensor $S_{MN} = \ft12\, \cR q_{MN}$ is the unique solution of the constraints  \eqref{eq:appSc3}, \eqref{eq:appSc4} whose trace is given by  $S^M_M = \cR$.
\label{thm:Sch}
\end{proposition}
\begin{proof}
The tensor  $S_{MN}$ can be decomposed in its trace and traceless parts as follows
\be
S_{MN} = \sigma_{MN} + \ft12 \cR q_{MN},
\ee
where $\sigma^M_M=0$. Then $S_{MN}$ is a solution to  the constraints  \eqref{eq:appSc3}, \eqref{eq:appSc4}  is and only if the traceless part $\sigma_{MN}$ satisfies
\be
\pd_n \sigma_{MN}=0, \quad D_{[M} \sigma_{N] P} =0, \quad \sigma_{MN} q^{MN}=0. 
\label{eq:sigmaEqs}
\ee
To complete our proof we just need to use a result by  Geroch \cite{Geroch1977},  who showed that the unique solution to these equations is $\sigma_{MN}=0$. 
\end{proof}
\finn

Since our setting is slightly different to that of \cite{Geroch1977}, we will reproduce here the proof of the last statement:
\begin{lemma}
The only solution to the system of equations \eqref{eq:sigmaEqs} is the trivial one, that is  $\sigma_{MN}=0$.
\label{lem:geroch}
\end{lemma}
\begin{proof}
Let $\sigma_{MN}$ be a tensor satisfying \eqref{eq:sigmaEqs} and $k^M$ a killing vector\footnote{Recall that,  in our conformal gauge, $q_{MN}$ represents the geometry of a two dimensional sphere with constant curvature} of $q_{MN}$. For convenience we work in a coordinate system such that  $q_{MN}=\delta_{MN}$ locally. We will begin proving that the spatial tensor  $\lambda_{MN} \equiv  D_{[M} (\sigma_{N]P} k^P)=0$ is vanishing. Due to the  antisymmetry of $\lambda_{MN}$ we already have $\lambda_{11}=\lambda_{22}=0$ and $\lambda_{12} =- \lambda_{21}$, and thus, it only remains to show that $\lambda_{12}=0$.  Using \eqref{eq:sigmaEqs} the it an be checked that  the components of $\lambda_{MN}$ satisfy
\be
\lambda_{MN} =  k^P D_{[M} \sigma_{N]P} + \sigma_{[NP} D_{M]} k^P =\sigma_{[NP} D_{M]} k^P, 
\ee
and therefore we also have
\bea
\lambda_{12} &=&  \sigma_{1}^1 D_{2} k_1+ \sigma_{1}^2 D_{2} k_2-  \sigma_{2}^2 D_{1} k_1 -  \sigma_{2}^2 D_{1} k_2\nonumber\\
&=&  \sigma_{1}^1 D_{2} k_1 -  \sigma_{2}^2 D_{1} k_2 = ( \sigma_{1}^1 + \sigma_{2}^2 )D_{2} k^1 =0.
\eea
Here we have used $\sigma_M^M=0$, and also the following properties of the killing vector  $k^M$ 
\be
D_{(M} k_{N)}=0\quad  \Longrightarrow \quad D_1 k_1 = D_2 k_2 =0, \quad  D_1 k_2 = - D_2 k_1,
\ee
which are a direct consequence of the killing equation. 
Therefore we have that  $D_{[M} (\sigma_{N]P} k^P)=0$ and, since the spatial sections of $\cI$ are topologically equivalent to  $\mathbb{S}^2$,    this implies that  $\sigma_{NP} k^P = \pd_N \alpha$ for some smooth function $\alpha= \alpha(\xi^M)$. Actually,  it is straight forward to check that $\alpha$ must be harmonic 
\be
\Delta \alpha = D^M(\sigma_{MP} k^P) =  k^P  q^{MN} D_N \sigma_{MP}  +    \sigma_{MP} D^N k^P=0. 
\ee
The last term vanishes because  $D_M k_N$ is antisymmetric and $\sigma_{MN}$ symmetric, and the first one can be shown to be zero using  the second of the  equations \eqref{eq:sigmaEqs} together with $\sigma_{MN} = \sigma_{NM}$ and $\sigma_M^M=0$
\be
k^P  q^{MN} D_N \sigma_{MP}=-k^P D_P \sigma^M_M=0.
\ee
The only harmonic functions on the spatial sections of $\cI$ (which are by assumption compact and simply connected) are constants, and thus  we have $\sigma_{MN} k^M=\pd_M \alpha=0$. Finally, since the killing vectors $k^M$ of the sphere span all of the tangent space, we can conclude that  $\sigma_{MN}=0$. 
\end{proof}
\finn
\subsubsection*{Constraint equations in the absence of outgoing radiation at null infinity}

We will now prove that, in the absence of outgoing radiation thought $\cI$, the general solution to the constraint equations  \eqref{eq:Raychad}, \eqref{eq:NS} and \eqref{eq:Xi} at null infinity is given by \eqref{eq:vacuaScriSol}. 

As discussed in section \ref{sec:Scri}, in the conformal gauge \eqref{eq:divFree} the only non-trivial constraint equation is the one of the transverse connection $\Xi_{MN}$ \eqref{eq:Xi}, together with   \eqref{eq:XiScri2}. The final form of these constraint equations  can be found substituting in them the expression for the gauge conditions \eqref{eq:divFree},  the form of the source term \eqref{eq:sourceTermsScri}, and using that in the absence of radiation the Schouten tensor satisfies \eqref{eq:XiScri2} and $S_{a\ell} = \pd_a S_\ell$ (see section \ref{sec:Scri}). The result is
 \be
\pd_n \Xi_{MN} = -\ft14 ( \cR +  2\pd_n S_{\ell} )q_{MN}, \qquad 
D_{[M} \Xi_{N]P} = \ft12 q_{P[M} \pd_{N]} S_{\ell}.
\label{eq:appScriVac}
 \ee
 In order to eliminate the redundancy associated  with supertranslations we specify at fiducial vacuum connection  $\Xi_{MN}^0$, and then 
characterise a generic vacuum connection $\Xi_{MN}$ by the difference  $\Sigma_{MN} \equiv \Xi_{MN}- \Xi_{MN}^0$. It is straightforward to check that this quantity is invariant under supertranslations \eqref{eq:scriST}. Thus, given a fixed fiducial vacuum $\Xi_{MN}^0$, we can characterise  the full set  of radiative vacua at null infinity finding the most general form of $\Sigma_{MN}$ which is consistent with the constraint equations \eqref{eq:appScriVac}.
\begin{theorem}
Let $\Sigma_{MN}\equiv\Xi_{MN} -\Xi_{MN}^0$  be the difference between two transverse connections, $\Xi_{MN}$ and $\Xi_{MN}^0$, which solve the constraint equations \eqref{eq:appScriVac} in the absence of outgoing radiation through $\cI$.
Then $\Sigma_{MN}$   has the general form
\be
\Sigma_{MN} =  D_M f_N + \ft12 \cR f q_{MN} -\ft12 (S_{\ell} - S^0_\ell) q_{MN},
\ee
where $f(\xi)$ is  smooth function on $\cI$ satisfying $\pd_n f=0$, and the potentials $S_{\ell}$ and $S_{\ell}^0$  are defined by $S_{a\ell} = \pd_a S_\ell$ and $S_{a\ell}^0 = \pd_a S^0_\ell$, in terms of the Schouten tensor associated to the connections $\Xi_{MN}$ and $\Xi_{MN}^0$  respectively.
\end{theorem}
\begin{proof}
Due to the linearity of the equations \eqref{eq:appScriVac},  $\Sigma_{MN}$ should satisfy
 \be
\pd_n \Sigma_{MN} =- \ft12 \pd_n (S_{\ell} - S_{\ell}^0 )q_{MN}, \qquad 
D_{[M} \Sigma_{N]P} = - \ft12 D_{[M} (  q_{N]P} (S_{\ell} -S_{\ell}^0 ) ).
\label{eq:appScriVac2}
 \ee
With the  change of variables $\hat \Sigma_{MN} \equiv \Sigma_{MN} + \ft12 (S_{\ell}- S^0_{\ell})q_{MN}$, the previous equations take the simpler form
\be
\pd_n \hat \Sigma_{MN} =0, \qquad 
D_{[M} \hat \Sigma_{N]P} = 0.
\label{eq:hatSigma}
\ee
Contracting the second equation with $q^{MP}$ we also find
\be
q^{MP}D_{[M} \hat \Sigma_{N]P} =0\qquad \Longrightarrow \qquad D^N \hat \Sigma_{MN} = D_M \hat \Sigma_N^N.
\label{eq:tranverse}
\ee
In order to solve  \eqref{eq:hatSigma}  and \eqref{eq:tranverse} consider the following decomposition of $\hat \Sigma_{MN}$ \cite{Stewart:1990fm,ehlers1979isolated}
\be
\hat \Sigma_{MN} = D_{MN} \chi + D_{(M} A_{N)} + W_{MN} + t q_{MN},
\label{eq:SigmaDecomp}
\ee
where $\chi$ and $t$ are two scalar fields on $\cI$ satisfying $\pd_n \chi = \pd_n t=0$, and $t= \hat \Sigma_M^M$. The vector  $A_M$  and the tensor $W_{MN}$ are both independent of the null coordinate $\pd_n A_M = \pd_n W_{MN}=0$, and moreover $W_{MN}$  is also transverse and traceless
\be
D^N W_{MN}=0, \quad \text{and}\quad W^M_M=0.
\ee 
The operator $D_{MN}\chi$ is defined by
\be
D_{MN} \chi \equiv D_M D_N \chi - (\Delta + \ft12 \cR) \chi, 
\ee
and has the property that $D_{MN}\chi$ is transverse for all scalar fields $\chi$. 
As was proven in \cite{Stewart:1990fm}, in the previous decomposition  the  scalar $\chi$ can only be determined up  to $\chi \to \chi + \lambda$ where $\lambda$ is a solution to $(\Delta + \ft12 \cR)\lambda=0$, and  the vector $A_M$ is fixed up to $A_M \to A_M + k_M$, where  $k_M$ is a killing vector of $q_{MN}$. The tensor $W_{MN}$ and $t$ are both uniquely determined  by $\hat \Sigma_{MN}$.
 Inserting this decomposition \eqref{eq:SigmaDecomp} in the equation \eqref{eq:tranverse}  we find
\be
D^N D_{(M} A_{N)} = \pd_M t,
\ee
which can be solved by $A_M = \pd_M \phi$,  where $\phi$ is a scalar satisfying $\Delta \phi + \ft12 \cR \phi = \ft12 t$.
The condition $\hat \Sigma_M^M =t$ leads to the equation
\be
q^{MN} D_{MN} \chi - 2 D^M A_M=0 \qquad \Longrightarrow \qquad  \Delta \chi +\cR \chi =2 \Delta \phi.
\ee
The decomposition \eqref{eq:SigmaDecomp}   is more conveniently expressed in terms of the combination   $f \equiv \chi + 2 \phi$, which satisfies $\Delta f + \cR f = 2 t$.  We obtain the expression 
\be
\hat \Sigma_{MN} = D_M D_N f + \ft12 \cR f q_{MN}+ W_{MN}.
\ee
Substituting this result into \eqref{eq:hatSigma} we find the following constraint for $W_{MN}$ 
\be
\qquad D_{[M} W_{N] P} =0.
\ee
Thus we can see the lemma \ref{lem:geroch} applies to $W_{MN}$, since this tensor is constant along the null direction and traceless, and therefore we can conclude that $W_{MN}=0$.
Collecting these results, the final form of the solution to the equations \eqref{eq:appScriVac} is 
\be
\Sigma_{MN} =  D_M f_N + \ft12 \cR f q_{MN} -\ft12 (S_{\ell} -  S^0_\ell)q_{MN},
\label{eq:generalSolScri}
\ee
which is the expression we were looking for. Note that the freedom  to shift $A_M$ by a killing vector leaves this expression invariant, while the ambiguity to shift $\chi\to \chi + \lambda$ amounts to shifting $f \to f + \lambda$. 
\end{proof}
\finn

It is immediate to check that  \eqref{eq:generalSolScri} reduces to the form of the solution we presented in the main text, eq. \eqref{eq:vacuaScriSol},  when we express it in terms of the equivalence relation \eqref{eq:equiv}.

\bibliographystyle{JHEP}
\bibliography{HypersurfaceSymmetries}

\end{document}